\documentclass[11pt]{article}
\usepackage{amsmath,amssymb}
\usepackage{mathtools}
\usepackage{graphicx}
\usepackage[
  backend=biber,
  style=alphabetic,
  maxalphanames=99,
  backref=true
]{biblatex}
\usepackage{enumitem}

\usepackage{amsthm}
\usepackage{aliascnt}
\usepackage[hypertexnames=false]{hyperref}
\usepackage[nameinlink,noabbrev,capitalise]{cleveref}

\usepackage{braket}
\usepackage[margin=1in]{geometry}
\usepackage{amsmath,amssymb,amsfonts,setspace,url}
\usepackage{bm}
\usepackage{comment}
\usepackage{latexsym}
\usepackage{subfig}
\usepackage{float}
\usepackage{fancybox}
\usepackage{bigstrut,array,multirow}
\usepackage{here}
\usepackage{color}
\usepackage{mathpazo}
\usepackage{xcolor}
\usepackage{tikz}
\usetikzlibrary{arrows.meta}
\usepackage[disable]{todonotes}
\usepackage[T1]{fontenc}
\usepackage[bb=boondox]{mathalfa}

\newtheorem{theorem}{Theorem}[section]

\newaliascnt{corollary}{theorem}

\aliascntresetthe{corollary}

\newaliascnt{proposition}{theorem}

\aliascntresetthe{proposition}

\newaliascnt{definition}{theorem}
\newtheorem{definition}[definition]{Definition}
\aliascntresetthe{definition}

\newaliascnt{construction}{theorem}
\newtheorem{construction}[construction]{Construction}
\aliascntresetthe{construction}
\newaliascnt{assumption}{theorem}

\aliascntresetthe{assumption}

\newaliascnt{fact}{theorem}
\newtheorem{fact}[fact]{Fact}
\aliascntresetthe{fact}

\newaliascnt{conjecture}{theorem}

\aliascntresetthe{conjecture}

\newaliascnt{remark}{theorem}
\newtheorem{remark}[remark]{Remark}
\aliascntresetthe{remark}

\newaliascnt{claim}{theorem}

\aliascntresetthe{claim}

\newaliascnt{observation}{theorem}

\aliascntresetthe{observation}

\newaliascnt{lemma}{theorem}
\newtheorem{lemma}[lemma]{Lemma}
\aliascntresetthe{lemma}

\crefname{theorem}{Theorem}{Theorems}
\crefname{corollary}{Corollary}{Corollaries}
\crefname{proposition}{Proposition}{Propositions}
\crefname{definition}{Definition}{Definitions}
\crefname{property}{Property}{Properties}
\crefname{fact}{Fact}{Facts}
\crefname{conjecture}{Conjecture}{Conjectures}
\crefname{remark}{Remark}{Remarks}
\crefname{problem}{Problem}{Problems}
\crefname{claim}{Claim}{Claims}
\crefname{open}{Open Problem}{Open Problems}
\crefname{observation}{Observation}{Observations}
\crefname{lemma}{Lemma}{Lemmas}

\Crefname{theorem}{Theorem}{Theorems}
\Crefname{corollary}{Corollary}{Corollaries}
\Crefname{proposition}{Proposition}{Propositions}
\Crefname{definition}{Definition}{Definitions}
\Crefname{property}{Property}{Properties}
\Crefname{fact}{Fact}{Facts}
\Crefname{conjecture}{Conjecture}{Conjectures}
\Crefname{remark}{Remark}{Remarks}
\Crefname{problem}{Problem}{Problems}
\Crefname{claim}{Claim}{Claims}
\Crefname{open}{Open Problem}{Open Problems}
\Crefname{observation}{Observation}{Observations}
\Crefname{lemma}{Lemma}{Lemmas}

\crefname{construction}{Construction}{Constructions}
\Crefname{construction}{Construction}{Constructions}
\crefname{assumption}{Assumption}{Assumptions}
\Crefname{assumption}{Assumption}{Assumptions}

\newcommand{\proj}[1]{\ket{#1}\!\bra{#1}}
\newcommand{\norm}[1]{\left\lVert #1\right\rVert}

\newcommand{\E}{\mathbb E}
\newcommand{\negl}{\operatorname{negl}}
\newcommand{\poly}{\operatorname{poly}}
\newcommand{\Tr}{\operatorname{Tr}}

\newcommand{\ind}{\mathbf 1}
\newcommand{\Tpub}{T_{\mathrm{pub}}}

\newcommand{\Dig}{\mathsf{Dig}}
\newcommand{\enc}{\mathsf{enc}}

\newcommand{\Ran}{\operatorname{Ran}}

\newcommand{\QPT}{\mathsf{QPT}}

\newcommand{\TFGenOracles}{\mathsf{TF.GenOracles}}
\newcommand{\TFSetup}{\mathsf{TF.Setup}^{\mathcal{O}}}
\newcommand{\TFSpark}{\mathsf{TF.Spark}^{\mathcal{O}}}
\newcommand{\TFVer}{\mathsf{TF.Ver}^{\mathcal{O}}}
\newcommand{\TFClone}{\mathsf{TF.Clone}^{\mathcal{O}}}

\newcommand{\OSSGenOracles}{\mathsf{OSS.GenOracles}}
\newcommand{\OSSGenQKey}{\mathsf{OSS.GenQKey}}
\newcommand{\OSSVer}{\mathsf{OSS.Ver}}
\newcommand{\OSSSign}{\mathsf{OSS.Sign}}

\newcommand{\FireGenOracles}{\mathsf{Fire.GenOracles}}
\newcommand{\FireSetup}{\mathsf{Fire.Setup}^{\mathcal O}}
\newcommand{\FireSpark}{\mathsf{Fire.Spark}^{\mathcal O}}
\newcommand{\FireVer}{\mathsf{Fire.Ver}^{\mathcal O}}
\newcommand{\FireClone}{\mathsf{Fire.Clone}^{\mathcal O}}

\newcommand{\FireGen}{\mathsf{Fire.Gen}^{\mathcal G}}
\newcommand{\FireHash}{\mathsf{Fire.Hash}^{\mathcal O}}
\newcommand{\FireTag}{\mathsf{Fire.Tag}^{\mathcal O}}

\newcommand{\FireCheck}{\mathsf{Fire.Check}}

\newcommand{\aux}{\mathsf{aux}}

\newcommand{\ananta}[1]{}
\newcommand{\rohit}[1]{}
\newcommand{\supartha}[1]{}
\newcommand{\vir}[1]{}

\hypersetup{pdftitle={Quantum Fire with Delegated Cloning},pdfsubject={Technical manuscript}}
\title{Quantum Fire with Delegated Cloning}
\author{
Rohit Chatterjee \thanks{National University of Singapore.
Email: \texttt{rochat@nus.edu.sg}.}
\qquad
Ananta Mukherjee\thanks{Centre for Quantum Technologies, National University of Singapore.
Email: \texttt{a.mukherjee@u.nus.edu}.}
\qquad
Vir Pathak\thanks{Stony Brook University.
Email: \texttt{vpathak@cs.stonybrook.edu}.}
\qquad
Supartha Podder\thanks{Stony Brook University.
Email: \texttt{supartha@cs.stonybrook.edu}.}
}
\date{}
\begin{document}
\maketitle
\begin{abstract}
 Quantum fire is a recently introduced cryptographic primitive consisting of efficiently preparable quantum states, called \emph{flames}, that admit efficient cloning but resist efficient telegraphing, namely reconstruction via classical communication without preshared entanglement. In all prior constructions of quantum fire, cloning is a public operation that requires no separate key and every holder of a flame state can clone it. For applications to access control, however, an issuer may wish to delegate cloning to designated quantum servers while withholding this capability from other flame holders. To address this, we introduce \emph{delegatable quantum fire}, in which cloning requires a separate key.

We give two constructions in the classical-oracle model. Our first construction uses a classical secret key which enables cloning, and any user with the entire key may clone successfully. Our second construction, which we call \emph{torch-fire}, uses quantum cloning keys, called \emph{torches}, which enable cloning while remaining unaffected in the process but cannot otherwise be split or delegated to enable additional cloning. An efficient adversary given \(m\) torches cannot, except with negligible probability, enable more than \(m\) noncommunicating parties to each clone a fresh, independently issued challenge flame. The adversary may jointly process its resources before separating the parties and distribute arbitrarily entangled registers among them. Both constructions are in the oracle model, relying on public classical oracles that allow queries in quantum superposition.
\end{abstract}
\newpage
\tableofcontents
\newpage

\section{Introduction}\label{sec:intro}
 
The no-cloning principle~\cite{WZ82,Die82} asserts that an unknown quantum state cannot be copied in general. Since the early proposal of quantum money by Wiesner~\cite{Wie83}, this principle has been the basis for a large body of work on \emph{uncloneable cryptography}, which uses the impossibility of copying to obtain guarantees that have no classical analogue. Examples include quantum money and quantum lightning~\cite{AC12,Zha19}, copy-protection of programs~\cite{Aar09,CLLZ21,ALLZZ21}, secure software leasing~\cite{AL21, BJLPS21}, uncloneable encryption~\cite{BL20, AK2021}, quantum witnesses and advice~\cite{BKL23, CKP25, BGPS25} and one-shot signatures~\cite{AGKZ20,SZ25} (see also~\cite{AC23} for a discussion of the different no-cloning paradigms used in such constructions).
 
A closely related principle is that of \emph{no-telegraphing}~\cite{Wer01}: an unknown quantum state cannot be converted into classical information from which it can later be reconstructed. Information-theoretically, the two principles are equivalent. A classical description of a state can be copied freely, and so the ability to telegraph a state implies the ability to clone it. Conversely, given an unbounded number of copies of a state, one can learn a classical description of it via tomography. This second direction, however, is inherently inefficient, and it is natural to ask whether the two principles separate once we restrict ourselves to efficient computation.
 
\paragraph{Quantum fire.} Nehoran and Zhandry~\cite{NZ24} showed that this is indeed the case, relative to a quantum oracle. They introduced the notion of \emph{quantum fire}: a family of efficiently preparable states that can be efficiently cloned by anyone, but which cannot be efficiently telegraphed, namely transmitted using only classical communication (and no pre-shared entanglement). Bostanci, Nehoran, and Zhandry~\cite{BNZ25} subsequently gave a candidate construction based on group actions, and Çakan, Goyal, and Shmueli~\cite{CGS26} gave the first construction of public-key quantum fire relative to classical oracles, with unconditional security in that model. In the same work, they also introduced \emph{quantum key-fire}, where the flame state additionally serves as a functional key (such as a signing or decryption key) that can be cloned freely but cannot be exfiltrated into a classical string. Very recently, Casper, Nehoran, and Sattath~\cite{CNS26} gave the first construction of quantum fire in the standard model, from one-shot signatures together with indistinguishability obfuscation and LWE, and used it to obtain publicly certifiable min-entropy.
 
Suggested applications for quantum fire have included forms of access control that are enabled by quantum computers. Roughly speaking, a flame can be replicated freely across quantum devices, but cannot be leaked or transmitted over classical channels. This, however, provides an extremely total model of control. Every entity that holds a flame is also able to clone it, and the only line being drawn is between quantum and classical communication. In many settings one could reasonably desire a more nuanced form of control, where some designated parties are allowed to replicate flames, while others, even those with quantum computers and with flames in hand, are not. As an informal example, an issuing authority may want a few trusted servers to be able to replicate credentials as demand requires, without every holder of a credential gaining the same ability. This leads to the following question.
 
\begin{quote}
\textbf{Question 1:} Can quantum fire be augmented so that only parties holding an appropriate key can clone flames, while flames remain untelegraphable?
\end{quote}
 
We emphasize that for such a guarantee to be meaningful, uncloneability for parties without the key must be \emph{functional}: such a party should not be able to produce two states that each pass verification, irrespective of whether these states are copies of the original flame. In this sense, flames should behave like uncloneable states to those without the key, and like quantum fire to those with it.
 
\paragraph{The nature of the key.} The answer to Question 1 depends quite crucially on what kind of key is used to delegate cloning. The simplest choice is a classical secret key. However, a classical key can itself be copied and handed out arbitrarily, and so any party that is ever given cloning power can grant it to as many others as it likes. The issuing authority therefore has no control over the number of parties that are eventually able to clone. To retain such control, the key must itself be quantum, and must resist duplication. At the same time, a quantum key that is consumed upon use, as is the case with signing tokens~\cite{BS23} or the signing state of a one-shot signature~\cite{AGKZ20}, would only permit a bounded number of cloning operations. Ideally, we would like a key that can be used to clone repeatedly while remaining intact, but which still cannot be used to grant cloning power at more locations than there are keys. This raises a second question.
 
\begin{quote}
\textbf{Question 2:} Is there a reusable quantum cloning key, such that $m$ such keys cannot be used to enable cloning at more than $m$ separate locations?
\end{quote}
 
We note that formalizing Question 2 already requires some care. An adversary holding $m$ keys need not split them up in the obvious way. It may instead jointly process its keys and flames, and distribute arbitrarily entangled registers among a larger number of parties, in the hope that more than $m$ of them are subsequently able to clone. A meaningful definition must account for such strategies.
 
\subsection{Our Results}\label{sec:our-results}
 
In this work, we study the notion of quantum fire with \emph{delegatable cloning}, which augments quantum fire so that only authorized users holding the right kind of key are able to clone flames, while users without a key remain unable to clone even in the functional sense above. We give two constructions, both in the classical oracle model (i.e., relative to classical oracles that may be queried in superposition), achieving different flavors of this guarantee. Our first construction gives a positive answer to Question 1 using a classical secret key, and our second construction gives a positive answer to Question 2 using reusable quantum keys, which we call \emph{torches}.
 
\subsubsection{Delegation via a Classical Secret Key}\label{sec:intro-classical}
 
Our first construction considers the setting where a classical secret key is required to clone, and any user with the entire key may clone successfully. In this construction, flames are privately verifiable, and the same secret key is used both to verify and to clone flames. Flames can be cloned for a bounded number of generations, and each flame keeps track of how many further clones are possible.
 
\begin{theorem}[Informal]\label{thm:intro-classical}
Relative to classical oracles, there exists a privately verifiable quantum fire scheme with a classical secret cloning key, in which flames can be cloned up to $O(\log \lambda)$ generations by any user holding the key, are uncloneable for any efficient (quantum) adversary that does not hold the key, and are untelegraphable.
\end{theorem}
 
At a high level, our construction builds on one-shot signatures in the classical oracle model of Shmueli and Zhandry~\cite{SZ25}, and on the quantum fire construction of Casper, Nehoran, and Sattath~\cite{CNS26}, where a flame is a one-shot signing state accompanied by a chain of signatures certifying its provenance. Signing in the scheme of~\cite{SZ25} requires access to a classical dual oracle, which is not needed for key generation or verification, and we use this oracle as the secret cloning key. Our analysis relies on the delayed collapsing property of this scheme, recently established by Bartusek et al.~\cite{BGMS25}. We give an overview of the construction and its analysis in Section~\ref{sec:overview-classical}.
 
\subsubsection{Delegation via an Uncloneable Quantum Key: Torch-Fire}\label{sec:intro-torch}
 
As discussed above, a classical cloning key offers no control over how many parties can eventually clone. Our second construction, which we call \emph{Torch-Fire}, addresses this with quantum keys. A torch-fire scheme separates the publicly verifiable quantum states, which we call \emph{flames}, from reusable quantum cloning credentials, which we call \emph{torches}. The issuing authority prepares some number $R$ of torches during setup, each associated with a public identifier, and is the only party that can create fresh flames. Anyone can verify a flame using only public information. A party holding a torch can clone any flame, and receives its torch back after each cloning operation, in a state that allows it to continue cloning. More precisely, our correctness guarantee requires that a single torch can be used to clone a flame polynomially many times in succession, and that all the resulting flames pass verification (with overwhelming probability).
 
We ask for two security properties of a torch-fire scheme. The first is \emph{untelegraphability}, which we require to hold even against parties that hold torches. That is, two parties that do not share any entanglement (but may share classical randomness), each possibly holding its own set of torches, cannot use classical communication to transmit a flame from one to the other. The second property, which is our main focus, is \emph{global non-proliferation}. Here, an adversary receives a set of $m$ torches of its choice, together with polynomially many flames, and may process all of these jointly. It then distributes arbitrarily entangled registers among $N$ parties, after which the parties may no longer communicate. Each party now receives a fresh challenge flame, issued independently, and attempts to output two flames that both verify under the serial number of its challenge. Global non-proliferation requires that, except with negligible probability, at most $m$ of the parties succeed.
 
\begin{theorem}[Informal]\label{thm:intro-torch}
Relative to classical oracles, there exists a torch-fire scheme satisfying correctness, untelegraphability (even against parties holding torches), and global non-proliferation. That is, for any polynomial $N$ and any $m \leq R$, no efficient adversary given $m$ torches can enable more than $m$ non-communicating parties to each clone a fresh challenge flame, except with negligible probability.
\end{theorem}
 
Setting $m = 0$, global non-proliferation in particular implies that flames are uncloneable (in the functional sense) for parties that do not hold any torch. Further, it implies that torches cannot be cloned or split in any functional sense either: with $m$ torches, cloning remains possible from at most $m$ locations, regardless of how the torches are processed or entangled across these locations. We note that this guarantee is quite different from the usual uncloneability of the torch states themselves. Indeed, it is not about whether an adversary can produce copies of a torch, but whether it can distribute the \emph{capability} to clone among more parties than it was given torches.
 
We also remark that Theorem~\ref{thm:intro-torch} holds unconditionally relative to classical oracles. Following~\cite{CGS26}, we work in a model where each instance of the scheme is associated with its own sampled collection of classical oracles. For our case also we can lift it to the standard oracle model where there is one global oracle and the setup procedure samples public key, secret key to specify the corresponding oracles.
 
\paragraph{Our construction.} A torch in our scheme is simply the signing state of a (long-message) one-shot signature scheme~\cite{SZ25}. A flame consists of a flame of the quantum fire scheme of~\cite{CGS26}, together with an independently generated one-shot signing state that we call its \emph{anchor}, which is what makes flames uncloneable to parties without a torch. Cloning the anchor requires access to an oracle that answers only when queried with valid signatures from both the anchor and a torch, all computed coherently and then uncomputed, so that the torch is left essentially undisturbed and can be reused. Untelegraphability is inherited from the scheme of~\cite{CGS26}. We give an overview of the construction in Section~\ref{sec:overview-torch}.
 
\paragraph{Extraction from separated computations.} The main technical challenge lies in proving global non-proliferation. Roughly, a party that clones its challenge flame must have obtained answers from the anchor-cloning oracle for that flame, and so must have produced a valid torch signature on a message involving that flame. Extracting such signatures from $m+1$ successful parties, using only $m$ torches, would then contradict the security of the one-shot signature scheme. The difficulty is that the parties share an arbitrarily entangled state, so the events of successful cloning at different locations, as well as the outcomes of any local extraction procedures, may be arbitrarily correlated. We need that whenever at least $m+1$ parties succeed, local procedures that do not communicate with each other \emph{simultaneously} extract at at least $m+1$ locations.
 
To this end, we prove a general extraction lemma for $N$ separated oracle computations run on an arbitrary (possibly entangled) joint input. Consider modifying the oracle on some set of inputs, which we call witnesses. Our lemma states, roughly, that the probability that at least $d$ of the original computations accept is at most the probability that at least $d$ efficient local extractors each output a witness, plus an error term. This error term is governed by the probability that the \emph{modified} computations accept after their inputs have been passed through a suitable filter, along with quantities that can be made arbitrarily (inverse-polynomially) small through the choice of parameters. We believe that this lemma may be of independent interest, and anticipate that it will find more applications in settings involving several non-communicating parties that share entanglement. We describe the lemma and its application in Section~\ref{sec:overview-nonprolif}.
 
\paragraph{Open questions.} Our work leaves several natural questions open. Both of our constructions are in the classical oracle model, and it would be interesting to instantiate them in the standard model, for instance building on the approach of~\cite{CNS26}. It would also be interesting to achieve a publicly verifiable variant of our first construction, or to support an unbounded number of cloning generations with a classical key. For torch-fire, our untelegraphability guarantee is against parties that do not share entanglement and communicate classically; a stronger notion would allow arbitrary LOCC protocols between the parties. Finally, a very pertinent question is to study whether one can realize quantum {\em key-fire} with delegatable cloning properties. This seems to require nontrivial additions to our approach. 
 
\paragraph{Organization.} We give the overviews behind our technical results in \Cref{sec:overview}. \Cref{sec:prelims} describes several preliminaries of quantum information and computation used in this paper. In \Cref{sec:sources}, we describe the two foundational cryptographic building blocks: One-shot Signature due to \cite{SZ25} and Quantum Fire due to \cite{CGS26} and their underlying machinery used in our approach. In, \Cref{sec:classical-secret-key-Q-Fire}, we give the construction of the delegatable quantum fire in the classical secret key realm. In \Cref{sec:definition}, we formally define the torch-fire primitive ans especially its global non-proliferation property. In \Cref{sec:torch-fire}, we give a construction of the torch-fire object and prove its correctness and untelegraphability properties. In \Cref{sec:extraction}, we develop the extraction lemma and using it prove the global non proliferation property of the torch-fire object.

\subsection{Related Work}\label{sec:related}
 
\paragraph{Quantum fire and no-telegraphing.} The notion of quantum fire was introduced by Nehoran and Zhandry~\cite{NZ24}, who used it to separate no-cloning from no-telegraphing relative to a quantum oracle, and also separated $\mathsf{cloneableQMA}$ from $\mathsf{QCMA}$ in the same setting. Bostanci, Nehoran, and Zhandry~\cite{BNZ25} gave a general duality theorem for representations of groups, with applications to quantum money, lightning and fire, along with a candidate construction of quantum fire from group actions. The work of Çakan, Goyal, and Shmueli~\cite{CGS26} gave the first construction of public-key quantum fire relative to classical oracles, along with quantum key-fire, and our torch-fire construction uses their scheme as a building block. We note that key-fire is quite different from what we consider here: there, the flame itself is a functional key that anyone may clone, whereas in our setting, cloning is the capability being restricted. Casper, Nehoran, and Sattath~\cite{CNS26} gave the first standard-model construction of quantum fire, assuming one-shot signatures along with indistinguishability obfuscation and LWE, and our first construction follows their template. In a related direction, untelegraphable encryption, where an adversary is restricted to holding classical information about a ciphertext, was introduced by Champion, Kitagawa, Nishimaki, and Yamakawa~\cite{CKNY25}, and further studied by Broadbent, Culf, and Rochette~\cite{BCR25}. Security against classical leakage of quantum states has also been studied by Çakan, Goyal, Liu-Zhang, and Ribeiro~\cite{CGLR24}.
 
\paragraph{One-shot signatures.} One-shot signatures were introduced by Amos, Georgiou, Kiayias, and Zhandry~\cite{AGKZ20}. Shmueli and Zhandry~\cite{SZ25} gave a construction relative to classical oracles with unconditional security, which forms the basis of both of our constructions. A simpler and more efficient construction was subsequently given by Huang and Vaikuntanathan~\cite{HV25}, with further improvements by Huang, Shmueli, Vaikuntanathan, and Zhandry~\cite{HSVZ26}. Bartusek, Gupte, Mutreja, and Shmueli~\cite{BGMS25} showed that the scheme of~\cite{SZ25} satisfies a delayed collapsing property, which is central to the analysis of our first construction. In a sense, our torches can be seen as one-shot signing states that are used repeatedly, but only in superposition, and are never measured.
 
\paragraph{Delegating and revoking quantum capabilities.} The idea of granting a limited capability through a quantum state also appears in quantum tokens for digital signatures~\cite{BS23}, where a token permits a single signature, and in secure software leasing~\cite{AL21} and secure key leasing~\cite{KN22,AKNYY23,APV23}, where a program or key can be lent out and later verifiably returned. In these settings, the capability is either consumed upon use or meant to be revoked. In contrast, our torches can be used an unbounded polynomial number of times, and the guarantee we seek is that the capability they confer cannot be spread to more locations than there are torches.

\paragraph{Extraction against entangled parties.} Arguing about several non-communicating parties that share entanglement is a recurring challenge in uncloneable cryptography, for instance in the analysis of monogamy-of-entanglement games~\cite{TFKW13}. In the oracle setting, Broadbent and Lord \cite{BL20} proved a simultaneous version of the one-way-to-hiding lemma~\cite{Unr15,AHU19} for two entangled parties, whose loss was shown by Majenz, Schaffner, and Tahmasbi~\cite{MST21} to be inherent to some extent, and two-party simultaneous extraction from oracle queries was further developed in~\cite{CMP24,AKLLZ22}. Simultaneous reductions for two parties also appear in~\cite{AKY25,CLX26}, and collusion-resistant copy-protection and multi-copy uncloneable primitives~\cite{LLQZ22,CG24} handle $k+1$ entangled parties that must all succeed, using projective implementations~\cite{Zha20}. Simultaneous extraction from several extractors acting on a common state also arises in the context of post-quantum non-malleability~\cite{CLPY25}. Our extraction lemma differs from these in handling threshold events over any polynomial number of separated locations, and in extracting oracle queries using only local, efficient procedures. Its proof uses the alternating projector technique of Marriott and Watrous~\cite{MW05}, which has been used extensively for extraction from quantum adversaries~\cite{Zha20,CMSZ21,LMS22}, together with a hybrid argument along the lines of~\cite{BBBV97}. We give a more detailed comparison in Section~\ref{sec:overview-nonprolif}.

\paragraph{Statements on AI usage:}Generative AI assistants (Claude Fable 5.1 and ChatGPT Astra 6) were used extensively for writing and editorial changes. Additionally, ChatGPT 6 Astra Pro helped suggest the final torch fire construction (based on heuristics the authors already had) as well as an initial proof sketch, which was significantly cleaned up and redesigned. The problem motivation, framing and the broad approach are solely the authors' work. LLMs were also used for literature search. The authors verified the technical and editorials contents generated by AI tools and take full responsibility for them.

\section{Technical Overview}\label{sec:overview}

This section provides an overview of our techniques. We first recall the one-shot signature scheme of~\cite{SZ25}, which underlies both of our constructions, in Section~\ref{sec:overview-oss}. We then discuss our construction with a classical secret key in Section~\ref{sec:overview-classical}. After that, we describe our torch-fire construction in Section~\ref{sec:overview-torch}, and the proof of global non-proliferation, including our extraction lemma, in Section~\ref{sec:overview-nonprolif}.

\subsection{One-Shot Signatures in the Classical Oracle Model}\label{sec:overview-oss}

We briefly recall the one-bit one-shot signature scheme of Shmueli and Zhandry~\cite{SZ25}, ignoring the specific parameters. The scheme is defined relative to a uniformly random permutation $\pi$ on $n$-bit strings, whose output is split as $\pi(x) = (H(x), J(x))$, together with a uniformly random full-rank matrix $A_y$ and a shift $b_y$ for every $y$ in the range of $H$. These define, for every $y$, a coset $W_y = b_y + \mathrm{colspan}(A_y)$. The public oracles are as follows: a forward oracle $P(x) = (H(x), A_{H(x)} J(x) + b_{H(x)})$, its inverse $P^{-1}$, which in particular reveals whether a string $u$ lies in $W_y$, and a dual oracle $D(y, v)$, which tests whether $v$ lies in the dual of the subspace $\mathrm{colspan}(A_y)$. None of the matrices, shifts, or the permutation are disclosed beyond these oracles.

Applying $P$ to a uniform superposition of inputs and uncomputing the input using $P^{-1}$ gives a superposition $\sum_y \ket{y}\ket{W_y}$, where $\ket{W_y}$ is the uniform superposition over $W_y$. Measuring the first register then gives a uniformly verification key $y$, and the signing state $\ket{W_y}$. A signature on a bit $b$ is any element of $W_y$ whose first bit is $b$, which can be checked using $P^{-1}$. To sign, one measures the first qubit of $\ket{W_y}$, and if the `wrong' bit is measured, uses the dual oracle $D$ (in the Hadamard basis) to restore the state to a uniform superposition over $W_y$ before trying again. Multiple bits are signed using parallel instances \footnote{ For efficiency, long messages are first hashed using a random oracle. More efficient multi-bit schemes are known~\cite{HV25,HSVZ26}; we do not rely on them here.} Briefly, security says that no efficient adversary can produce valid signatures on two distinct messages (or, in the strong version, two distinct message-signature pairs) under the same verification key, even with access to all three oracles.

Two features of this scheme are central to our constructions. First, the dual oracle $D$ is needed for signing, but not for generating keys or verifying signatures. Second, the projector onto the signing state $\ket{W_y}$ itself can be measured publicly: one tests membership in $W_y$ using $P^{-1}$, then applies Hadamards and tests membership in the dual subspace using $D$. Two registers that both pass this test can be measured to obtain two independent uniformly random elements of $W_y$, which are distinct with overwhelming probability, and so give a strong one-shot signature forgery. This is the precise sense in which signing states are uncloneable.

\subsection{Delegation via a Classical Secret Key}\label{sec:overview-classical}

\paragraph{Withholding the dual oracle.} Our first construction relies on the first observation above, namely that if the dual oracle $D$ of~\cite{SZ25} is withheld from the public, then a party that holds a signing state $\ket{W_y}$ but only has access to $P$ and $P^{-1}$ can still verify signatures and check coset membership, but cannot sign arbitrary messages. In other words, signing becomes an operation that requires a specific authorization, namely (access to) $D$. We would like to use this property to make cloning a keyed operation for our purposes. However, a signing state itself cannot simply be copied, even with $D$, and so we look towards a notion of flames for which ``cloning'' means producing two new states that each pass verification, rather than two copies of a pristine signing state.

\paragraph{Flames as signature chains.} A quantum fire construction along these lines was recently obtained by Casper, Nehoran, and Sattath~\cite{CNS26}. Following their approach, we let a flame state consist of a current one-shot key pair together with a chain of signatures certifying its provenance. During spark, the authority samples the oracles $(P, P^{-1}, D)$ and a root key pair $(\mathsf{pk}_{\mathrm{root}}, \ket{\mathsf{sk}})$, and sets the serial number to consist of $\mathsf{pk}_{\mathrm{root}}$ (along with the public oracles). The initial flame is $(\emptyset, \mathsf{pk}_{\mathrm{root}}, \ket{\mathsf{sk}})$, and the secret key is $D$. To clone a flame $(\mathsf{hist}, \mathsf{pk}_{\mathrm{cur}}, \ket{\mathsf{sk}_{\mathrm{cur}}})$, one generates two fresh key pairs $(\mathsf{pk}_0, \ket{\mathsf{sk}_0})$ and $(\mathsf{pk}_1, \ket{\mathsf{sk}_1})$ using only the public oracles, and then uses $D$ and the current signing state to sign the message $\mathsf{pk}_0 \| \mathsf{pk}_1$, obtaining a signature $\sigma$. The two children are
\[
\big(\mathsf{hist} \| (\mathsf{pk}_0, \mathsf{pk}_1, \sigma),\ \mathsf{pk}_0,\ \ket{\mathsf{sk}_0}\big) \quad \text{and} \quad \big(\mathsf{hist} \| (\mathsf{pk}_0, \mathsf{pk}_1, \sigma),\ \mathsf{pk}_1,\ \ket{\mathsf{sk}_1}\big).
\]
The original signing state is consumed in this process, which is consistent with the requirement that the parent flame produces exactly two clones - in this case it is replaced by its children. To verify a flame, one checks that every entry in its history is a valid signature on the next pair of keys, under one of the two keys of the previous entry (starting from $\mathsf{pk}_{\mathrm{root}}$), that the current key is one of the two keys in the last entry, and that the current signing state passes the projective test for $\ket{W_{\mathsf{pk}_{\mathrm{cur}}}}$ described above. Since this last test uses $D$, verification is also private. Verification additionally rejects histories longer than a fixed bound $c(\lambda)$, and reports the number of remaining clones $c(\lambda) - T$ for a flame with history of length $T$. Correctness of verification and cloning follows from the correctness of the one-shot signature scheme. As in \cite{CNS26}, we also require the depth of the flame state to be bounded by $c(\lambda) = O(\log \lambda)$ in order to obtain an efficient reduction from OSS security to untelegraphability. The reason for this is because the reduction must compute all possible continuations of the flame state.

\paragraph{Uncloneability without the key.} Consider an adversary that is given a flame (possibly after some honest cloning operations) and access only to $P$ and $P^{-1}$, and outputs two registers that both pass verification. We first observe that it suffices to bound the adversary's success when the projective test on the current signing state is replaced by a weaker \emph{coset test}, which simply measures the register and checks membership in $W_{\mathsf{pk}_{\mathrm{cur}}}$ using $P^{-1}$. Any state that passes the original test passes the coset test with at least the same probability, so this replacement can only help the adversary. The advantage is that the entire remaining experiment, namely the adversary and the modified verification, now only uses $P$ and $P^{-1}$.

This allows us to use the \emph{delayed collapsing} property of the scheme of~\cite{SZ25}, recently established by Bartusek, Gupte, Mutreja, and Shmueli~\cite{BGMS25}. Roughly, this says that a register that (when measured) yields a valid signature cannot be distinguished from its measured version by any efficient procedure that only has access to $P$ and $P^{-1}$, even if the register itself was prepared using all three oracles. This is precisely our setting: the adversary's flame may be the result of honest cloning operations performed using $D$, but everything after that point avoids $D$. We may therefore move, through a sequence of hybrids, to an experiment where each single-bit signing state held by the adversary has been measured in the computational basis. We note that ordinary collapsing would not suffice here, since the honest preparation of the flame requires $D$.

In the fully collapsed experiment, we can reduce to the security of the one-shot signature scheme. The key observation is that a measured signing state, i.e., a uniformly random element of $W_y$, is distributed identically to a signature on a uniformly random bit. Since the concatenation of two freshly generated verification keys is itself uniformly random, for each (measured) signing state given to the adversary, the reduction can sample the next generation of keys $(\mathsf{pk}_0, \mathsf{pk}_1)$ in advance, retaining their signing states, sign $\mathsf{pk}_0 \| \mathsf{pk}_1$ using the signing state (and the dual oracle, which it has in the one-shot signature security game), and give the resulting signature to the adversary in place of the measured signing state. The adversary's view is identical to that in the collapsed experiment. Now suppose the adversary outputs an accepted flame that extends the chain by a new entry $(\widehat{\mathsf{pk}}_0, \widehat{\mathsf{pk}}_1, \hat\sigma)$. If $\widehat{\mathsf{pk}}_0 \| \widehat{\mathsf{pk}}_1 \neq \mathsf{pk}_0 \| \mathsf{pk}_1$, then the two messages differ in some bit, and the corresponding components of $\hat\sigma$ and the reduction's signature give valid signatures on two distinct bits under the same single-bit key. Otherwise, the adversary's flame contains a state that passes the coset test for one of the keys $\mathsf{pk}_a$ that the reduction generated, and whose signing state the reduction still holds. Measuring a component of the adversary's state gives a signature on some bit $b$, and the reduction can use its own signing state to sign $1-b$. In both cases, this violates the security of the one-shot signature scheme.

Note that in the case where the adversary outputs a clone where the new $pk_0 = pk_1 = y$, the reduction is significantly simpler. In this case, the corresponding state in the signing key register of both clones is the OSS signing key corresponding to $y$. Therefore, the adversary can use the two copies of the keys to sign different bits and violate OSS security. 

\paragraph{Untelegraphability.} The argument for untelegraphability follows the reduction to OSS security presented in \cite{CNS26} almost verbatim. We briefly outline it here. Suppose $\QPT$ adversaries $\mathcal{A}A$ and $\mathcal{B}B$ share a classical channel and $A$ is able to telegraph a flame state to $B$. Recall that the depth of the flame state is bounded by $c\log(\lambda)$. Therefore, the reduction makes $\mathcal{A}A$ telegraph the flame state to $\mathcal{B}$ more than $2^{c\log(\lambda)}$ many times. By the pigeonhole principle, we must have that either $2$ of the clones must terminate with the same public key, in which case the signing keys at the end of the chain in both clones can be used to violate OSS security. The other case is when a terminal public key in one clone coincides with a public key in the internal history of the other clone. Then we can again use the terminal signing key in the first clone to obtain an OSS forgery.

\subsection{Torch-Fire}\label{sec:overview-torch}

We now turn to torch-fire. Recall that we would like flames to be publicly verifiable, as well as cloneable by any party holding one of the $R$ torches issued at setup but uncloneable (functionally) for everyone else, and untelegraphable even for parties holding torches. Moreover, we would like global non-proliferation: $m$ torches should enable cloning at no more than $m$ separate locations.

\paragraph{Starting point: CGS fire.} Our starting point is the public-key quantum fire scheme of Çakan, Goyal, and Shmueli~\cite{CGS26}, which we refer to as CGS fire. A CGS flame is the state $\ket{K} = \sum_y \ket{y}\ket{W_y}$, i.e., the superposition over all verification keys and signing states of a one-shot signature instance, before measuring the key. The authority holds private access to a generation oracle $G$ that bundles the forward, inverse and dual oracles of this instance. Public access to $G$ is provided through an oracle that answers a query $z$ only when it is accompanied by a valid \emph{tag}, consisting of a key $y$ and two hash values of $(y, z)$. These hash values are in turn only released in exchange for valid one-shot signatures on $0$ and on $1$ respectively under $y$. A party holding $\ket{K}$ can compute such a tag coherently, controlled on the key register, by signing both bits, querying the hash oracles, and uncomputing the signatures. Since honest signatures are accepted with overwhelming probability, this disturbs $\ket{K}$ only negligibly. Thus a holder of a CGS flame can query $G$, and in particular can run the preparation circuit for $\ket{K}$ on fresh registers, which gives public cloning. \cite{CGS26} also show that the resulting scheme is untelegraphable.

\paragraph{Anchors.} To restrict cloning, one could try to modify CGS fire directly, for instance by requiring additional credentials in its public oracles. We instead leave the CGS fire state intact (so that we can use its untelegraphability as is) and attach to each flame a second component whose uncloneability is easy to reason about. We call this the \emph{anchor}: an independently generated signing state $\ket{A_a} = \bigotimes_{c} \ket{W_{a_c}}$ of an ($L$-bit) one-shot signature scheme, with verification key $a = (a_1, \ldots, a_L)$. The flame state is thus $\ket{A_a} \otimes \ket{K}$, and it is verified by applying the projective test for $\ket{A_a}$ from Section~\ref{sec:overview-oss} to the anchor (componentwise) and running CGS verification on the second part. Since the accepting projector for the anchor is the rank-one projector $\ket{A_a}\!\bra{A_a}$, two registers that both pass anchor verification yield a strong one-shot signature forgery, exactly as described in Section~\ref{sec:overview-oss}.

The anchors of different flames do not all use the same one-shot signature oracles. Instead, each flame has a uniformly random \emph{label} $\ell \in \{0,1\}^\lambda$, and its anchor is generated relative to an independent oracle collection $\mathcal{O}_{A,\ell}$ associated with this label. The serial number of the flame is $s = (\ell, a)$. This choice plays an important role in the proof of global non-proliferation, where we will need to disable access related to a single flame without affecting any other flame.

\paragraph{Torches and cloning anchors.} A torch is also a signing state $\ket{T_t}$ of a multi-bit one-shot signature scheme (relative to a single oracle collection $\mathcal{O}_T$ shared by all torches), with the verification keys $t_1, \ldots, t_R$ of all issued torches published through a public oracle. Since we need torches for specific published keys, the authority receives the matrices and shifts defining the corresponding cosets, which allows it to prepare $\ket{T_{t_j}}$ directly.

To clone a flame with serial $s = (\ell, a)$, one clones its CGS part using the public cloning procedure of~\cite{CGS26}, and must additionally prepare a fresh copy of $\ket{W_{a_c}}$ for every component $c$ of the anchor. Unlike in the key generation procedure of~\cite{SZ25}, the key $a_c$ here is fixed in advance, and so this cannot be done using the public oracles of $\mathcal{O}_{A,\ell}$. Instead, we observe that $\ket{W_y}$ can be prepared from the two affine maps
\[
\mathsf{affine}_{y,0}(w) = A_y w + b_y, \qquad \mathsf{affine}_{y,1}(u) = \begin{cases} w & \text{if } u = A_y w + b_y, \\ \bot & \text{if } u \notin W_y, \end{cases}
\]
by preparing a uniform superposition over $w$, computing $A_y w + b_y$ into a second register, and uncomputing $w$ using the inverse map. These maps are never made available directly. Access to them is provided only through a dedicated public oracle $\Gamma$, which takes as input a query $q = (s, t, c, e, u)$ asking for $\mathsf{affine}_{a_c, e}(u)$ in component $c$ of the anchor source $\mathcal{O}_{A,\ell}$, and answers only if all of the following hold:
\begin{itemize}
\item $t$ is the published key of an issued torch;
\item the query carries a valid anchor signature, under $a$, on the message $0 \| \langle s, t\rangle$;
\item the query carries a valid torch signature, under $t$, on the message $1 \| \langle s, t \rangle$;
\item the query carries a valid CGS tag for (an encoding of) the query $q$ itself.
\end{itemize}
Here $\langle s, t\rangle = \ell \| a \| t$, and the leading bit separates anchor and torch messages. In particular, every answered query is bound to the label $\ell$ of the flame being cloned and to the specific torch being used.
\rohit{It would help to state explicitly where the CGS tag check is used in the analysis. In the draft, the non-proliferation witness predicate ignores it, and untelegraphability simulates it; is it needed only to bind anchor-cloning to possession of a full flame?}
\ananta{yes, that was the thought}

\paragraph{Reusing the torch.} The cloning procedure makes $2L$ queries to $\Gamma$, and each of them uses the same anchor, torch, and CGS flame. For each query, the party coherently computes the request, the CGS tag (using the CGS part of the flame, as in~\cite{CGS26}), and the two signatures (using the coherent signing circuits for the anchor and the torch), queries $\Gamma$, and then uncomputes everything except the answer. No signature, tag, or answer is ever measured. Since honestly computed signatures and tags pass the checks of $\Gamma$ except with negligible probability, and $\Gamma$ agrees with the ideal affine map whenever its checks pass, each such query acts as the ideal affine query on the honest resources, up to an error of $2\sqrt{d}$ in Euclidean norm, where $d$ is the total signing and tagging error. In particular, the anchor, the torch and the CGS flame are returned essentially undisturbed. A hybrid argument over the $2L$ queries shows that a single cloning call produces a state that is negligibly close (in trace distance) to the torch together with two copies of the honest flame, and a further hybrid over polynomially many calls shows that a single torch can be reused to clone polynomially many times. This is what distinguishes torches from ordinary one-shot signing states: a torch is only ever used to sign in superposition, and is never measured.

\paragraph{Untelegraphability.} Untelegraphability reduces to that of CGS fire. Since the accepting effect for a flame is $\ket{A_a}\!\bra{A_a} \otimes B$, where $B$ is the accepting effect of CGS verification, it is dominated by $I \otimes B$, and so it suffices to show that the adversary cannot telegraph the CGS part of a flame. Given a pair of telegraphing adversaries for torch-fire, we construct a pair for CGS fire as follows. The two reductions use shared classical randomness (which is allowed in the definition) to sample all torch and anchor oracles consistently, prepare the torches requested by each party locally, and forward all queries to the CGS oracles to the CGS challenger. Fresh flames are obtained by requesting a CGS flame from the challenger and attaching a freshly generated anchor. The only oracle that combines CGS data with the sampled data is $\Gamma$, through its tag check. Since $\Gamma$ only checks tags on forward queries, for which the generation oracle never returns $\bot$, a tag is valid exactly when the public CGS oracle answers the corresponding query. Thus each query to $\Gamma$ can be simulated exactly using two queries to the public CGS interface, and the receiving reduction simply discards the anchor of its output flame.

\subsection{Global Non-Proliferation}\label{sec:overview-nonprolif}

We now describe the proof of global non-proliferation. Recall that the adversary jointly processes $m$ torches and polynomially many flames, splits into $N$ non-communicating parties sharing an arbitrary entangled state, and each party $i$ receives a fresh challenge flame with serial $s_i = (\ell_i, a_i)$ and attempts to output two flames that both verify under $s_i$. We first move to an experiment where all labels issued during the experiment are distinct, which only changes the success probability negligibly.

\paragraph{The high-level idea.} Fix a location $i$. If party $i$ outputs two accepted flames, it must have produced two registers passing the anchor test for $a_i$. Without any affine access to the anchor source $\mathcal{O}_{A,\ell_i}$, this is impossible, as it would yield a one-shot signature forgery. So party $i$ must have made queries to $\Gamma$ on label $\ell_i$, and any such query that $\Gamma$ answers carries a valid torch signature on a message beginning with $1 \| \ell_i$. If we could extract such a torch signature from each of $m+1$ successful parties, we would obtain $m+1$ valid torch signatures on distinct messages (since the labels are distinct), while only $m$ torches were issued. Then either two of these signatures are under the same issued torch, or one of them is under a torch that was never issued, whose signing state the reduction still holds and can use to sign a different message (say, one beginning with $0$). Either way, this gives a forgery for the long-message one-shot signature scheme underlying the torches.

The difficulty lies in the extraction step. The parties share an arbitrarily entangled state, and so the events of successful cloning at different locations may be arbitrarily correlated, as may the outcomes of any extraction procedures we run locally. It does not suffice to show that each successful party can individually be made to yield a torch signature. We need that whenever at least $m+1$ parties succeed, local extraction procedures, which do not communicate with each other, \emph{simultaneously} succeed at at least $m+1$ locations. Moreover, we need this extraction to be efficient, since it is used inside a reduction to one-shot signature security.

\paragraph{Original and modified computations.} We formalize the above as follows. At location $i$, let $U^1_i$ denote the (purified) computation that runs party $i$ on its input register and then verifies both of its output flames under $s_i$, and let $\Pi_i$ project onto acceptance. We define a modified computation $U^0_i$ that is identical, except that it uses a modified oracle $\Gamma^0_i$, which returns $\bot$ on every query carrying the label $\ell_i$. The two oracles only differ on queries that carry the label $\ell_i$ and are answered by $\Gamma$, and every such query carries a valid torch signature on a message involving $\ell_i$. We therefore let the \emph{witnesses} at location $i$ be queries to $\Gamma$ that carry label $\ell_i$, an issued torch key $t$, and a valid torch signature under $t$. Notably, this predicate only checks the torch signature, and not the anchor signature or the CGS tag, so it can be tested using public oracles alone. The anchor key within the query also need not equal $a_i$: only the label matters.

\paragraph{An extraction lemma for separated computations.} We abstract this situation into a general lemma. Consider $N$ locations, each running a $q$-query oracle computation on its share of a joint (possibly entangled) input state $\rho$, in an original and a modified version that differ only on witness queries. Let $A_i$ and $C_i$ denote the accepting maps of the original and modified computations at location $i$. Let $E_i$ be the effect corresponding to the following \emph{witness sampler}: choose $j \in [q]$ uniformly, run the modified computation up to its $j$-th query, and test whether the query is a witness. A standard hybrid argument in the style of~\cite{BBBV97}, as in the proof of the one-way-to-hiding lemma~\cite{Unr15,AHU19}, gives the operator inequality
\[
(A_i - C_i)^\dagger (A_i - C_i) \preceq 4q^2 E_i,
\]
i.e., the original and modified computations can only differ substantially on components of the input on which the witness sampler succeeds with noticeable probability. We handle the two kinds of components separately.

On components where the witness sampler rarely succeeds, we would like to replace the original computation by the modified one. To this end, we use the \emph{filter} $L_i = (I - E_i)^\beta$, for a parameter $\beta$. This can be implemented efficiently as follows. We prepare a uniform superposition over $j \in [q]$, coherently run the modified computation up to query $j$, perform the witness test and continue only on rejection, uncompute, and continue only if all ancillary registers are returned to zero. A successful run of this cycle applies exactly $I - E_i$ to the input, and $\beta$ successful runs apply $L_i$. On the filtered input, the comparison above improves to
\[
L_i (A_i - C_i)^\dagger (A_i - C_i) L_i \preceq 4 q^2 E_i (I - E_i)^{2\beta} \preceq \frac{4 q^2}{2\beta + 1},
\]
since the function $u(1-u)^{2\beta}$ is at most $1/(2\beta+1)$ on $[0,1]$. Thus, after filtering, original acceptance is bounded by (twice) modified acceptance, up to an error that decreases with $\beta$.

On components where the witness sampler often succeeds, we would like to actually extract a witness. Here we use an alternating projector procedure following Marriott and Watrous~\cite{MW05}. We run the coherent witness sampler once and test for a witness; if this fails, we alternate between measuring the projector onto the image of the (isometric) coherent sampler, and testing for a witness again, for $t = \lceil 2\beta/\delta^2 \rceil$ rounds. On an eigenvector of $E_i$ with eigenvalue $u$, the two projectors act on a two-dimensional invariant subspace, and each additional round succeeds with probability exactly $2u(1-u)$, independently of previous rounds. A short calculation then shows that the success effect $W_i$ of this extractor satisfies
\[
W_i \succeq (I - L_i)^2 - \delta I,
\]
by considering separately eigenvalues $u \geq \delta/(2\beta)$, where the extractor succeeds with probability at least $1 - \delta$, and $u < \delta/(2\beta)$, where $(1 - (1-u)^\beta)^2 \leq \delta$. Combining the two cases using the decomposition $A_i = A_i (I - L_i) + A_i L_i$ and a weighted Cauchy--Schwarz inequality, we obtain a local comparison of the form
\[
A_i^\dagger A_i - W_i \preceq \delta(2+\delta) I + (1 + \delta^{-1})\, L_i A_i^\dagger A_i L_i,
\]
where the last term corresponds to successful filtering followed by original acceptance, and is bounded in terms of modified acceptance as above.

It remains to lift these local comparisons to a comparison between the events ``at least $d$ locations accept'' and ``at least $d$ locations extract a witness''. Each of these events corresponds to a threshold effect, obtained by summing tensor products of local success and failure effects over all outcomes with at least $d$ successes. We replace the acceptance effect by the extraction effect one location at a time. Changing the outcome at location $i$ only affects the threshold event when exactly $d-1$ of the other locations succeed, and so the difference between consecutive hybrids is exactly $(A_i^\dagger A_i - W_i) \otimes G_i$, where $G_i$ is the (positive semidefinite) effect of exactly $d-1$ successes elsewhere. Since tensoring with a positive semidefinite operator preserves operator inequalities, the local comparison applies directly. This is what allows us to handle arbitrary entanglement between the locations without assuming any independence. Altogether, we obtain:

\begin{lemma}[Extraction from separated computations, informal]\label{lem:overview-extraction}
For any $\beta, \delta$ and threshold $d$, there are efficient local extractors such that
\[
\Pr\Big[\textstyle\sum_i X_i \geq d\Big] \leq \Pr\Big[\textstyle\sum_i Y_i \geq d\Big] + N\delta(2 + \delta) + (1 + \delta^{-1})\left(2\varepsilon_\beta + \frac{8Nq^2}{2\beta+1}\right),
\]
where $X_i$ and $Y_i$ indicate acceptance by the original computation and successful witness extraction at location $i$, and $\varepsilon_\beta$ is the total probability, summed over locations, that the filter $L_i$ succeeds and the modified computation then accepts.
\end{lemma}

\paragraph{Comparison with prior techniques.} The individual ingredients of Lemma~\ref{lem:overview-extraction} are standard. The operator inequality comparing the two computations is the operator form of the hybrid argument underlying the one-way-to-hiding lemma~\cite{BBBV97,Unr15,AHU19}, and our extractor is an instance of the alternating projector technique of Marriott and Watrous~\cite{MW05}, which has been used extensively for extraction and estimation against quantum adversaries~\cite{Zha20,CMSZ21,LMS22}. The filter $L_i$ is a polynomial in the witness-sampling effect $E_i$, implemented through post-selection; we note that this is different from the measure-rewind-measure technique of~\cite{KSSSS20}, which rewinds the final measurement of a distinguisher to extract from it.\rohit{Polynomials in $E_i$ could alternatively be implemented using quantum singular value transformation~\cite{GSLW19}; consider whether this improves parameters.}

What we are not aware of is a statement of this form for many entangled locations. The closest analogue is the \emph{simultaneous} one-way-to-hiding lemma of Broadbent and Lord~\cite{BL20}, which considers two non-communicating parties with shared entanglement that both query a random oracle, and bounds the probability that both of them guess the value of the oracle at a hidden point. The bound is in terms of the probability that measuring a uniformly random query of each party simultaneously reveals the hidden point, and incurs a constant multiplicative loss that Majenz, Schaffner, and Tahmasbi~\cite{MST21} showed cannot be reduced below $9/8$. Coladangelo, Majenz, and Poremba~\cite{CMP24} gave a (lossy) two-party search-to-decision reduction of this kind for copy-protection in the quantum random oracle model, and Ananth, Kaleoglu, Li, Liu, and Zhandry~\cite{AKLLZ22} extracted from two entangled parties by combining Jordan's lemma, a threshold measurement on each party, and measuring a random query. Related simultaneous statements for two parties, which do not involve oracle queries, include the simultaneous quantum Goldreich-Levin reductions of~\cite{AKY25,CLX26}. On the other hand, works on collusion-resistant copy-protection and multi-copy uncloneable cryptography~\cite{LLQZ22,CG24} handle $k+1$ entangled parties, and show that they cannot all succeed when given $k$ copies, by applying projective implementations~\cite{Zha20} at each location and using the pigeonhole principle; these do not involve extracting oracle queries. Our own earlier simultaneous extraction lemma~\cite{CLPY25} considers several extractors acting on the same quantum state, each with a unique output, rather than separated computations.

In comparison, Lemma~\ref{lem:overview-extraction} applies to any polynomial number $N$ of separated locations, and directly compares the threshold events ``at least $d$ of $N$ locations accept'' and ``at least $d$ of $N$ locations extract''. This is important for our application, since the adversary may distribute its resources among $N \gg m+1$ parties, and we do not know in advance which $m+1$ of them succeed. A naive approach that guesses a set of $m+1$ locations and applies a simultaneous extraction argument to it would lose a factor of $\binom{N}{m+1}$, which is superpolynomial in general (for instance, when $N = 2m$ and $m$ grows with $\lambda$). Moreover, the error in our lemma is not bounded in terms of a finding probability; it is instead expressed through the acceptance probability of the modified computations after filtering, which we bound through a separate reduction (in our case, to the security of the anchors).\rohit{To the best of our knowledge there is no existing $N$-party threshold version, but please double check against recent work on collusion-resistant uncloneable primitives (e.g.\ secure key leasing and copy-protection with many copies) before submission.}

\paragraph{Bounding filtered modified acceptance.} To apply Lemma~\ref{lem:overview-extraction}, we first show that $\varepsilon_\beta$ is negligible for any polynomial $\beta$, using the security of the one-shot signature scheme underlying the anchors. The idea is that in the modified computation at location $i$, no query on label $\ell_i$ is ever answered, so the reduction can embed a one-shot signature challenge into the anchor source $\mathcal{O}_{A,\ell_i}$ (specifically, into its first component) without ever needing the corresponding affine maps. The reduction guesses the location $i$ uniformly at random, samples all other oracles itself, and generates the first component of the challenge anchor $a_i$ using the public key generation procedure of the challenge oracles. It then runs the filter followed by the modified computation at location $i$, and on acceptance, measures the first component of both accepted anchors. Since both anchors were projected onto $\ket{A_{a_i}}$, this gives two independent uniform elements of the challenge coset, and hence two distinct signatures under the challenge key, with overwhelming probability. Note that the reduction can implement the filter, since witnesses are testable using public torch verification.

There is one subtlety. The joint input $\rho$ was prepared before the split using the original oracle $\Gamma$, and so may already contain information obtained from affine queries on label $\ell_i$, which the reduction cannot answer. We therefore also disable these queries before the split. This only changes the input state negligibly, since the challenge label $\ell_i$ is uniformly random and unknown to the adversary at that point. Concretely, the total query mass on a uniformly random label is at most $Q_{\mathrm{pre}} 2^{-\lambda}$ in expectation for $Q_{\mathrm{pre}}$ pre-split queries, and the hybrid argument of~\cite{BBBV97} bounds the resulting change by $2Q_{\mathrm{pre}} 2^{-\lambda/2}$.

\paragraph{Bounding extraction.} Next, we show that witnesses for $m+1$ distinct labels can only be obtained with negligible probability, using the security of the long-message one-shot signature scheme underlying the torches. The reduction embeds the challenge oracles as the torch source $\mathcal{O}_T$, and generates all $R$ torches itself using the public key generation procedure. This gives the correct distribution of torches without the authority's private data. It hands the $m$ requested torches to the adversary and retains the rest, and samples the anchor and CGS sources itself, so that it can answer all affine queries and tag checks. The challenge oracles are only needed for torch verification within $\Gamma$. As discussed above, $m+1$ witnesses on distinct labels yield either two signatures on distinct messages under the same issued torch, or a signature under an unissued torch, which the reduction can combine with its own signature (computed using the retained torch) on a message beginning with $0$.

\paragraph{Putting things together.} Combining the two bounds with Lemma~\ref{lem:overview-extraction} at threshold $d = m+1$, we obtain that the probability that at least $m+1$ parties succeed is at most
\[
\mathsf{negl}(\lambda) + N\delta(2+\delta) + (1 + \delta^{-1})\frac{8Nq^2}{2\beta+1}.
\]
For any constant $c$, choosing $\delta = 1/(2N\lambda^c)$ and $\beta = \lceil N^2 q^2 \lambda^{2c}\rceil$ keeps the extractors and the filter efficient, and makes this bound $O(\lambda^{-c})$. Since the success probability does not depend on these parameters, and $c$ is arbitrary, it must be negligible.
\section{Preliminaries}
\label{sec:prelims}

\paragraph{Basic Notations.}
We write 
$\mathbb N_0=\{0,1,2,\ldots\}$.
For $n\in\mathbb N$, let $[n]=\{1,\ldots,n\}$.
Given a finite set $S$, we use $|S|$ to denote its cardinality
and $x\leftarrow S$ to denote uniform sampling from $S$.
For a statement $\mathcal E$, its indicator
$\ind[\mathcal E]$ is $1$ if $\mathcal E$ holds and $0$ otherwise.
We use $\lambda$ for the security parameter.
A function $\epsilon(\lambda)$ is negligible if, for every
polynomial $p(\lambda)$, we have
$\epsilon(\lambda)<1/p(\lambda)$ for all sufficiently large
$\lambda$.
We write $\negl(\lambda)$ for an unspecified negligible function.

\paragraph{Notations for Quantum Computing.}
We use the standard finite-dimensional quantum formalism as in ~\cite{NC00}.
For a vector $\ket{\psi}$ in a Hilbert space $\mathcal H$,
$\bra{\psi}$ denotes its conjugate transpose,
$\braket{\phi|\psi}$ denotes the inner product, and
$\ket{\phi}\bra{\psi}$ denotes the outer product.
The norm $\norm{\ket{\psi}}$ is the standard $\ell_2$ norm.
For a $d$-dimensional Hilbert space, the computational basis is
$\{\ket{0},\ldots,\ket{d-1}\}$.
A quantum register $F$ has an associated Hilbert space
$\mathcal H_F$, and the joint space of registers $F,G$ is
$\mathcal H_F\otimes\mathcal H_G$.
We write $I_F$ for the identity operator on $\mathcal H_F$.
Register subscripts and identity operators on untouched
registers are omitted when clear from context.

\paragraph{Positive Semidefinite Operators and Operator Order.}
For a linear operator $A$, let $A^\dagger$ denote its adjoint.
A Hermitian operator $A$ is positive semidefinite (PSD), written
$A\succeq0$, if
$\bra{v}A\ket{v}\ge0$ for every vector $\ket{v}$.
Equivalently, all eigenvalues of $A$ are nonnegative.
For Hermitian operators $A,B$ on the same space, we write
$A\preceq B$, or $B\succeq A$, if $B-A\succeq0$.
For every compatible linear operator $C$,
\[
 A\preceq B
 \quad\Longrightarrow\quad
 C^\dagger A C\preceq C^\dagger B C.
\]

We will also use that, for $T\succeq0$ and $\rho\succeq0$,
\[
 A\preceq B
 \quad\Longrightarrow\quad
 A\otimes T\preceq B\otimes T,
 \qquad
 \Tr(A\rho)\le\Tr(B\rho).
\]

\paragraph{Pure and Mixed States.}
A pure state is represented by a unit vector $\ket{\psi}$,
or equivalently by the density operator $\proj{\psi}$.
A (possibly) mixed state of a register $F$ is denoted by
$\rho_F$ and is a unit trace PSD operator on $\mathcal H_F$.
We also allow subnormalized states for the sake of analysis:
PSD operators $\rho$ satisfying $\Tr(\rho)\le1$.

For a joint state $\rho_{FG}$, its reduced state on $F$ is
$\rho_F=\Tr_G(\rho_{FG})$, where $\Tr_G$ denotes the partial
trace over $G$.

\paragraph{Projectors and Measurements.}
An orthogonal projector $\Pi$ on $\mathcal H$ satisfies
$\Pi^\dagger=\Pi$ and $\Pi^2=\Pi$.
Its support/range is
\(
 \Ran(\Pi)
 =\{\ket{v}\in\mathcal H:\Pi\ket{v}=\ket{v}\}.
\)
In particular, $0\preceq\Pi\preceq I$.
A projective measurement is a collection of pairwise
orthogonal projectors $\{\Pi_i\}_i$ satisfying
$\sum_i\Pi_i=I$.
On input $\rho$, the probability of outcome $i$ is
\(
 p_i=\Tr(\Pi_i\rho).
\)
For a pure input $\proj{\psi}$, this equals
$\norm{\Pi_i\ket{\psi}}^2$.
The unnormalized state after outcome $i$ is
$\Pi_i\rho\Pi_i$; if $p_i>0$, its normalized state is
$\Pi_i\rho\Pi_i/p_i$.
A computational-basis diagonal projector has the form
\(
 \Pi_S=\sum_{x\in S}\proj{x}
\)
for a subset $S$ of the computational-basis states.
Measuring $\{\Pi_S,I-\Pi_S\}$ tests membership in $S$.

\paragraph{Quantum Effects.}
A quantum effect is an operator $E$ satisfying
$0\preceq E\preceq I$.
It describes the acceptance probability of a quantum test:
on input $\rho$, the test accepts with probability
$\Tr(E\rho)$.
For example, let $V$ be an isometry, so that $V^\dagger V=I$,
and let $\Pi$ be an acceptance projector.
Applying $V$ and then testing $\Pi$ has acceptance effect
\(
 E=V^\dagger\Pi V,
 \:\:
 0\preceq E\preceq V^\dagger V=I,
\)
and hence, by cyclicity of trace,
\(
 \Pr[\mathrm{accept}] =\Tr(E\rho)
 =\Tr(\Pi V\rho V^\dagger)
\)

\paragraph{Operator Norms and Trace Distance}
For a linear operator $A:\mathcal H\to\mathcal K$, its
spectral norm is
\[
 \norm{A}_{\mathrm{op}}
 =\sup_{\ket{v}\in\mathcal H:\,\norm{\ket{v}}=1}
       \norm{A\ket{v}}.
\]
We also write $\norm{A}$ when the operator norm is clear
from context.

An operator $A$ is a contraction if
$\norm{A}_{\mathrm{op}}\le1$, equivalently
$A^\dagger A\preceq I$.
The trace norm is defined as:
\(
 \norm{A}_1=\Tr\!\left(\sqrt{A^\dagger A}\right).
\)
For density operators $\rho,\sigma$ on the same space,
their trace distance is
\(
 D(\rho,\sigma)=\frac12\norm{\rho-\sigma}_1.
\)

\paragraph{Classical Oracles Queried in Superposition.}
For a classical function
$f:\{0,1\}^a\to\{0,1\}^b$, an oracle for that function query-able in superposition means
access to the unitary
\[
 \mathcal O_f\ket{x,z,w}
 =\ket{x,z\oplus f(x),w},
 \qquad
 x\in\{0,1\}^a,\quad z\in\{0,1\}^b,
\]
where $\oplus$ denotes bitwise XOR and $w$ labels the
remaining workspace.
The action on arbitrary superposition of states is determined by linearity.
In particular,
\[
 \mathcal O_f
 \left(\sum_x\alpha_x\ket{x}\ket{0^b}\right)
 =\sum_x\alpha_x\ket{x}\ket{f(x)}.
\]

The notation $\mathcal A^{\mathcal O}$ indicates that
$\mathcal A$ has quantum query access to the specified
collection $\mathcal O$ of classical functions.

\paragraph{Query-Bounded Quantum Algorithms.}
A $T$-query quantum algorithm acts on a query register $Q$,
an answer register $Z$, and a working register $W$.

Its computation can be written as
\[
 U^{\mathcal O}
 =G_T \mathcal O_T G_{T-1}\cdots G_1 \mathcal O_1 G_0,
\]
where $G_0,\ldots,G_T$ are unitaries between queries and
$O_j$ is the $j$-th query to the specified oracle collection.

\paragraph{Quantum Polynomial-Time Algorithms.}
A quantum polynomial-time ($\QPT$) algorithm is described
by a uniformly generated family of quantum circuits whose
size and number of qubits are polynomially bounded in the
input length.
For an oracle-aided algorithm, each oracle query counts as
one gate; in particular, its total number of oracle queries
is polynomially bounded.
The oracle functions themselves need not be efficiently
sampleable or evaluable.

\begin{fact}
\label{fact:coherent-implementation}
\normalfont
Let $\mathcal A^{\mathcal O}$ be an oracle aided $\QPT$ algorithm with classical input $x\in\{0,1\}^n$,
quantum input register $R$, and classical output.
The input and output register sizes are fixed polynomially
bounded functions of $n$. The algorithm can perform any intermediate measurement and use classical randomness. 
Let the oracle $\mathcal O$ provides access to a classical function
$f\colon\{0,1\}^a\to\{0,1\}^b$ through the unitary
\[
 \mathcal{O}_f\ket{q}\ket{z}
 =\ket{q}\ket{z\oplus f(q)},
 \qquad q\in\{0,1\}^a,\quad z\in\{0,1\}^b.
\]

Then, there are polynomial-size registers $E,C$, and oracle aided $\QPT$ implementable unitaries $U_x^{\mathcal O}$ acting on the registers $REC$ satisfying the
following property. The register $C$ holds the output,
and $E$ contains all auxiliary workspaces. 
For every density operator $\rho_R$ and every output $c$,
 \begin{equation}\label{eq:coherent-implementation}
 \Pr\!\left[\mathcal A^{\mathcal O}(x,\rho_R)=c\right]
 =
 \operatorname{Tr}\!\left[
  (I_{RE}\otimes\proj{c}_C)\,
  U_x^{\mathcal O}
  (\rho_R\otimes\proj{0}_{EC})
  (U_x^{\mathcal O})^\dagger
 \right].
 \end{equation}

This follows from reversible computation and the principle of deferred measurement; see \cite{NC00}.
\end{fact}

\begin{lemma}[BBBV hybrid bound {\cite{BBBV}}]
\label{lem:bbbv-hybrid}
\normalfont
Consider two executions \(b\in\{0,1\}\) of a quantum algorithm
with the same initial state and the same unitary operations between
queries. The query input domain is $\mathcal D$. The difference between the two executions is: at query \(j\in[q]\), execution \(b\) applies a
query unitary \(O_j^b\).
For each \(j\), let \(P_j\) be a projector on the query registers
such that
\[
 (O_j^1-O_j^0)(I-P_j)=0.
\]
i.e., for every input $\ket{x}\in \mathcal D$ lying outside the support of the projector $P_j$, both the oracle outputs are identical.
Let \(\rho_j\) be the state of the query input register immediately before
query \(j\) in execution \(1\), and let \(\sigma^b\) be the final
state of execution \(b\). Define
\[
 M:=\sum_{j=1}^{q}\Tr(P_j\rho_j).
\]
Then
\[
 D(\sigma^1,\sigma^0)
 :=\frac12\norm{\sigma^1-\sigma^0}_1
 \le 2\sqrt{qM}.
\]

\end{lemma}
\section{Cryptographic Building Blocks}\label{sec:sources}

\subsection{One-shot signatures (OSS)}\label{subsec:sz}
We first start with the definition of One-shot Signature in a classical oracle model as introduced in \cite{AGKZ20}.
\begin{definition}
\label{def:oss}
\label{def:oracle-oss}
\normalfont
Let $\{\mathcal M_\lambda\}_{\lambda\in\mathbb N}$ be a family
of classical message spaces of polynomially bounded message length.
A one-shot signature scheme  consists of an oracle-generation procedure and three $\QPT$ algorithms: 
\[
 \mathsf{OSS}
 =(\OSSGenOracles,\OSSGenQKey^{\mathcal O},\OSSSign^{\mathcal O},\OSSVer^{\mathcal O})
\]
\begin{enumerate}[label=\arabic*.,leftmargin=*,itemsep=5pt]
 \item $\mathcal O\leftarrow\OSSGenOracles(1^\lambda)$:
 an oracle-generation procedure that samples a collection
 $\mathcal O$ of classical oracles.

 \item $(y,R)\leftarrow\OSSGenQKey^{\mathcal O}(1^\lambda)$:
 a $\QPT$ algorithm that outputs a
 classical verification key $y$ and a quantum signing
 register $R$.

 \item $\sigma\leftarrow\OSSSign^{\mathcal O}(y,R,m)$:
 a $\QPT$ algorithm that, on input
 a verification key $y$, a signing register $R$, and a message
 $m\in\mathcal M_\lambda$, outputs a classical signature
 $\sigma$.

 \item $b\leftarrow\OSSVer^{\mathcal O}(y,m,\sigma)$:
 a deterministic polynomial-time classical algorithm
 that outputs $b\in\{0,1\}$. 
\end{enumerate}
All the algorithms and adversaries have query access
to the oracle $\mathcal O$ query-able in superposition. The scheme must satisfy the following properties.

\paragraph{Correctness.}
For every
$m\in\mathcal M_\lambda$,
\begin{equation}\label{eq:oss-correctness-definition}
 \Pr\!\left[
 b=1:
 \begin{array}{l}
  \mathcal O\leftarrow\OSSGenOracles(1^\lambda)\\
  (y,R)\leftarrow\OSSGenQKey^{\mathcal O}(1^\lambda)\\
  \sigma\leftarrow\OSSSign^{\mathcal O}(y,R,m)\\
  b\leftarrow\OSSVer^{\mathcal O}(y,m,\sigma)
 \end{array}
 \right]
 \ge1-\negl(\lambda).
\end{equation}

\paragraph{One-shot security.}
For every quantum adversary $\mathcal A$ making polynomially
many oracle queries,
\begin{equation}\label{eq:oss-security-definition}
 \Pr\!\left[
 \begin{array}{c}
  m_0\ne m_1\ \wedge\\
  b_0=b_1=1
 \end{array}
 :
 \begin{array}{l}
  \mathcal O\leftarrow\OSSGenOracles(1^\lambda)\\
  (y,m_0,\sigma_0,m_1,\sigma_1)
       \leftarrow\mathcal A^{\mathcal O}(1^\lambda)\\
  b_i\leftarrow\OSSVer^{\mathcal O}(y,m_i,\sigma_i),
       \quad i\in\{0,1\}
 \end{array}
 \right]
 \le\negl(\lambda).
\end{equation}


\paragraph{Strong one-shot security.}
The scheme satisfies strong one-shot security if
\eqref{eq:oss-security-definition} holds with the condition
$m_0\ne m_1$ replaced by
\(
 (m_0,\sigma_0)\ne(m_1,\sigma_1).
\)
i.e., the underlying message may be the same but the produced signatures need to be distinct.

All probabilities are taken over oracle sampling, the randomness
of the algorithms, and measurement outcomes.
\end{definition}
Now, we shall describe the technical tools needed for the classical oracle OSS construction due to \cite{SZ25} followed by the construction with the modifications done in \cite{CGS26}.

\begin{construction}[Oracle distribution
{\cite[Construction~28]{SZ25}}]\label{con:sz-source}
\normalfont
For $\lambda\ge2$, set
\begin{equation}\label{eq:sz-parameters}
 r=16\lambda(\lambda-1),\qquad n=r+24\lambda,\qquad k=n.
\end{equation}
Sample a uniform permutation
$\pi:\{0,1\}^n\to\{0,1\}^n$ and write
\[
 \pi(x)=(H(x),J(x)),
 \qquad |H(x)|=r,\quad |J(x)|=n-r.
\]
Independently for each $y\in\{0,1\}^r$, sample a uniform
full-rank matrix
$A_y\in\mathbb Z_2^{k\times(n-r)}$ and a uniform vector
$b_y\in\mathbb Z_2^k$. Define the subspace $S_y$ and the coset $W_y$ as follows:
\begin{equation}\label{eq:sz-cosets}
 S_y=\{A_yw:w\in\mathbb Z_2^{n-r}\},
 \qquad W_y=b_y+S_y.
\end{equation}
The public oracle collection is
$\mathcal O=(\mathcal P,\mathcal P^{-1},\mathcal D)$, where
\begin{equation}\label{eq:SZ}
\begin{aligned}
 \mathcal P(x)&=(H(x),A_{H(x)}J(x)+b_{H(x)}),\\
 \mathcal P^{-1}(y,u)&=
 \begin{cases}
  \pi^{-1}(y,w),&u=A_yw+b_y
      \text{ for some }w\in\mathbb Z_2^{n-r},\\
  \perp,&\text{otherwise},
 \end{cases}\\
 \mathcal D(y,v)&=\ind[v^{\mathsf T}A_y=0^{n-r}].
\end{aligned}
\end{equation}
\begin{fact}
\normalfont
    We note that, beyond the access to these oracles, none of the matrices, shifts, and permutation are publicly disclosed.
\end{fact}
\begin{fact}
\normalfont
    Full column rank makes the inverse coordinate $w$ unique.
Consequently,
\[
 \mathcal P^{-1}(\mathcal P(x))=x,\qquad
 \mathcal P^{-1}(y,u)\ne\perp
 \ \Longleftrightarrow\ u\in W_y.
\]
\end{fact}
\end{construction}
\begin{definition}
\normalfont
    For each $y$, the corresponding coset superposition state is:
\begin{equation}\label{eq:sz-coset-state}
 \ket{W_y}
 =2^{-(n-r)/2}\sum_{w\in\mathbb Z_2^{n-r}}\ket{A_yw+b_y}
 =2^{-(n-r)/2}\sum_{u\in W_y}\ket u.
\end{equation}
\end{definition}

The following procedure describes, how using the oracles $\mathcal P$ and $\mathcal P^{-1}$, one can prepare the coset superposition states. 
\begin{construction}[Coset Superposition State Preparation,
{\cite[Construction~31]{SZ25}}]\label{fact:sz-generation}
\normalfont
For registers $X,Y,U$ of $n,r,k$ qubits, respectively, we apply the following sequence of procedures in the the specified registers,
\begin{align*}
 &\ket{0^n}_X\ket{0^r}_Y\ket{0^k}_U\\
 &\xrightarrow{H_X^{\otimes n}}
 2^{-n/2}\sum_x\ket x_X\ket{0^r}_Y\ket{0^k}_U\\
 &\xrightarrow{\text{Query }{\mathcal P} \text{ from } X \text{ to } YU}
 2^{-n/2}\sum_x\ket x_X\ket{H(x)}_Y
                       \ket{A_{H(x)}J(x)+b_{H(x)}}_U\\
 &\xrightarrow{\text{Query }{\mathcal{P}^{-1}} \text{ from } YU \text{ to } X}
 2^{-n/2}\sum_x\ket{0^n}_X\ket{H(x)}_Y
                       \ket{A_{H(x)}J(x)+b_{H(x)}}_U\\
 &=\ket{0^n}_X\otimes
   2^{-n/2}\sum_{y\in\{0,1\}^r}
                  \sum_{w\in\mathbb Z_2^{n-r}}
                  \ket y_Y\ket{A_yw+b_y}_U\\
 &=\ket{0^n}_X\otimes
   2^{-r/2}\sum_{y\in\{0,1\}^r}\ket y_Y\ket{W_y}_U.
\end{align*}
Measuring $Y$ gives
\begin{equation}\label{eq:sz-generation-output}
 \Pr[Y=y]=2^{-r},
 \qquad \text{ and the state } \rho_{U\mid Y=y}=\proj{W_y}.
\end{equation}
\end{construction}

\begin{construction}[One-bit signatures]\label{con:one-bit-oss}
\normalfont
Let the message space is $\{0,1\}$.
We use the signing algorithm to be the repeated trial algorithm of \cite[Section~6.2]{CGS26}.
\begin{enumerate}[label=(\roman*),leftmargin=*]
 \item $\OSSGenOracles(1^\lambda)$:
 sample and return
 $\mathcal O=(\mathcal P,\mathcal P^{-1},\mathcal D)$
 according to \Cref{con:sz-source}.

 \item $\OSSGenQKey^{\mathcal O}(1^\lambda)$:
 execute \Cref{fact:sz-generation} and measure the register $Y$.
 Output $(y,R)$, where the signing key $R$ is taken to be the register $U$.
 Thus the verification key $y$ is uniform on $\{0,1\}^r$ and,
 corresponding to $y$, the signing key $R$ is $\proj{W_y}$.

 \item $\OSSSign^{\mathcal O}(y,R,b)$:
 for $j=1,\ldots,\lambda$, execute:
 \begin{enumerate}[label=(\alph*),leftmargin=*]
  \item Measure the first qubit of $R$ in the computational
  basis, obtaining $c$. If $c=b$, measure the remaining
  qubits, obtaining $u'\in\{0,1\}^{k-1}$, and return
  $u=b\Vert u'$.

  \item If $c\ne b$, apply $H^{\otimes k}$ to $R$ and
  evaluate $\mathcal D(y,\cdot)$ into a fresh qubit:
  \[
   \ket v_R\ket0_B
   \longmapsto\ket v_R\ket{\mathcal D(y,v)}_B.
  \]
  Measure the register $B$, apply $H^{\otimes k}$ to $R$, and continue to the next iteration.
 \end{enumerate}
 Return $\perp$ if all $\lambda$ attempts fail.

 \item $\OSSVer^{\mathcal O}(y,b,u)$:
 Reject unless
 $(y,b,u)\in\{0,1\}^r\times\{0,1\}\times\{0,1\}^k$.
 Otherwise output
 \begin{equation}\label{eq:one-bit-verification}
  \OSSVer^{\mathcal O}(y,b,u)
  =\ind[\mathcal P^{-1}(y,u)\ne\perp]\wedge\ind[u_1=b].
 \end{equation}
\end{enumerate}
\end{construction}

The following is the affine-coset version of the standard
hidden-subspace verification test of ~\cite{AC12}. For, sake of completeness, We include a short proof
in our notation.
\begin{lemma}\label{lem:token-test}
Fix the oracles in \Cref{con:sz-source} and a verification key $y$. 
Consider the membership testing projectors
\(
 \Theta_y=\sum_{u\in W_y}\proj u,\:\:
 Q_y=\sum_{v\in S_y^\perp}\proj v,
\)
and the procedure
\(Q'_y=H^{\otimes k}Q_yH^{\otimes k}.\)
Then, for every $u\in\mathbb Z_2^k$,
\begin{equation}\label{eq:coset-averaging}
 Q'_y\ket u
 =2^{-(n-r)}\sum_{w\in u+S_y}\ket w.
\end{equation}
Moreover,
\begin{equation}\label{eq:coset-projector}
 \Pi_y:=Q'_y\Theta_y=\proj{W_y}.
\end{equation}

\end{lemma}
\begin{proof}
The amplitude of $\ket w$ in $Q'_y\ket u$ is
\[
 2^{-k}\sum_{v\in S_y^\perp}(-1)^{v^{\mathsf T}(u+w)}
 =
 \begin{cases}
  2^{-(n-r)},&u+w\in S_y,\\
  0,&u+w\notin S_y.
 \end{cases}
\]
This proves \eqref{eq:coset-averaging}. For $u\in W_y$,
$u+S_y=W_y$, and hence
\[
 Q'_y\Theta_y
 =2^{-(n-r)}\sum_{u,w\in W_y}\ket w\!\bra u
 =\proj{W_y}.
\]
\end{proof}
\begin{remark}
\normalfont
    Thus, there is an efficient public test that accepts the coset superposition state $\ket{W_y}$:
measure $\{\Theta_y,I-\Theta_y\}$ using the predicate
$\ind[\mathcal P^{-1}(y,u)\ne\perp]$. On acceptance, apply $H^{\otimes k}$, measure
$\{Q_y,I-Q_y\}$ using $\mathcal D(y,\cdot)$,
Accept
if and only if both outcomes are $1$.
\end{remark}

We now state the correctness and security properties of 
\Cref{con:one-bit-oss}.

\begin{theorem}[\Cref{con:one-bit-oss} Correctness, adapted from {\cite[Theorem~11]{CGS26}}]\label{thm:one-bit-correctness}
\normalfont
For fixed oracles and $y\in\{0,1\}^r$, let
\[
 \epsilon_{y,b}
 =\Pr\!\left[
  \OSSVer^{\mathcal O}(y,b,u)=0:
  u\leftarrow\OSSSign^{\mathcal O}(y,\proj{W_y},b)
 \right].
\]
Let $e_1=(1,0,\ldots,0)\in\mathbb Z_2^k$. Then
\begin{equation}\label{eq:one-bit-pointwise-error}
 \epsilon_{y,b}=
 \begin{cases}
  2^{-\lambda},&e_1^{\mathsf T}A_y\ne0^{n-r},\\
  \ind[(b_y)_1\ne b],&e_1^{\mathsf T}A_y=0^{n-r}.
 \end{cases}
\end{equation}
In particular, \Cref{con:one-bit-oss} satisfies correctness with probability at least $1-e(\lambda)$ where $e(\lambda):=2^{-\lambda}+2^{-(n-r)}$.
\end{theorem}

\begin{theorem}[\Cref{con:one-bit-oss} security, {\cite[Theorem~30]{SZ25}}]\label{fact:sz}
\normalfont
For every quantum algorithm $\mathcal A$ making at most $q$
queries to the oracles of \Cref{con:sz-source},
\[
 \Pr\!\left[
 \begin{array}{c}
  (b_0,u_0)\ne(b_1,u_1)\ \wedge\\
  \OSSVer^{\mathcal O}(y,b_0,u_0)
  =\OSSVer^{\mathcal O}(y,b_1,u_1)=1
 \end{array}
 :
 \begin{array}{l}
  \mathcal O\leftarrow\OSSGenOracles(1^\lambda)\\
  (y,b_0,u_0,b_1,u_1)\leftarrow
       \mathcal A^{\mathcal O}(1^\lambda)
 \end{array}
 \right]
 \le O\!\left(\frac{\lambda^3k^3(q+2)^3}{2^\lambda}\right).
\]
If $q$ is polynomial in $\lambda$, we get strong OSS security as per \Cref{def:oracle-oss}.
\end{theorem}

\begin{remark}[Signing Unitary]\label{rem:coherent-signing}
\normalfont
Fix the oracles $\mathcal O$, a verification key $y$, and
a bit $b$. The signing algorithm admits a polynomial-time
unitary oracle implementation $U^{\mathcal O}_{y,b}$ on registers $RE\Sigma$,
where $R$ holds the signing state, $E$ contains all auxiliary
registers and measurement records, and $\Sigma$ holds the
signature. 
The basis state $\ket{\perp}_\Sigma$ denotes failure. 
The circuit $U^{\mathcal O}_{y,b}$ is obtained from
$\OSSSign^{\mathcal O}(y,\cdot,b)$ by \Cref{fact:coherent-implementation}. 
For density operator $\rho_R$ corresponding to the signing key and the candidate signature 
$u\in\{0,1\}^k\cup\{\perp\}$,
\begin{equation}\label{eq:coherent-signing-output}
 \Pr[\OSSSign^{\mathcal O}(y,\rho_R,b)=u]
 =
 \operatorname{Tr}\!\left[
  (I_{RE}\otimes\proj{u}_\Sigma)
  U^{\mathcal O}_{y,b}(\rho_R\otimes\proj{0}_{E\Sigma})
  {U^{\mathcal O}_{y,b}}^{\dagger}
 \right].
\end{equation}
\end{remark}
We now consider the extension of this one bit signing scheme to a long message signing scheme by signing bit by bit. We note, \cite{HV25, HSVZ26} gave more efficient and direct ways to achieve that. But, in this paper, we stick to the bit by bit signing paradigm.
\begin{construction}[Signatures for long messages]
\label{con:long-oss}
\normalfont
Set the message length $L=\poly(\lambda)$. Define
\[
 \mathsf{OSS}
 =(\OSSGenOracles,\OSSGenQKey^{\mathcal O},\OSSSign^{\mathcal O},\OSSVer^{\mathcal O})
\]
as follows. From \Cref{con:one-bit-oss}, we use the corresponding oracles for one bit scheme and compose them component wise.
\begin{enumerate}[label=(\roman*),leftmargin=*]
 \item $\OSSGenOracles(1^\lambda)$:
 sample independent instances
 $\mathcal O_c\leftarrow\OSSGenOracles(1^\lambda)$,
 $c\in[L]$, of \Cref{con:one-bit-oss}, and an independent
 random function
 \begin{equation}\label{eq:Digest-Oracle-for-long-message-OSS}
  \Dig:\{0,1\}^*\longrightarrow\{0,1\}^{L}.
 \end{equation}
 Return the joint oracle
 $\mathcal O=(\mathcal O_1,\ldots,\mathcal O_L,\Dig)$.

 \item $\OSSGenQKey^{\mathcal O}(1^\lambda)$:
 independently generate
 \[
  (a_c,R_c)\leftarrow
       \OSSGenQKey^{\mathcal O_c}(1^\lambda),
       \qquad c\in[L],
 \]
 and output $(a,R_1\cdots R_L)$, where
 $a=(a_1,\ldots,a_L)$.

 \item $\OSSSign^{\mathcal O}(a,R_1\cdots R_L,M)$:
 compute $\Dig(M)=(b_1,\ldots,b_L)$, run
 \[
  \sigma_c\leftarrow
       \OSSSign^{\mathcal O_c}(a_c,R_c,b_c),
       \qquad c\in[L],
 \]
 and output $\sigma=(\sigma_1,\ldots,\sigma_L)$.

 \item $\OSSVer^{\mathcal O}(a,M,\sigma)$:
 reject malformed inputs. Otherwise, compute
 $\Dig(M)=(b_1,\ldots,b_L)$ and output
 \[
  \prod_{c=1}^L
       \OSSVer^{\mathcal O_c}(a_c,b_c,\sigma_c).
 \]
\end{enumerate}
Thus, for a fixed joint oracle $\mathcal{O}$ and message $M$ with $\Dig(M)=(b_1,\ldots,b_L)$, we have
\begin{equation}\label{eq:long}
\begin{aligned}
 &\text{the verification key }a=(a_1,\ldots,a_L),
 \qquad \text{the signing state }\ket{A_a}=\bigotimes_{c=1}^L\ket{W_{a_c}}, \\
 &\text{the signing unitary }U^{\mathcal{O}}_{a,M}=\bigotimes_{c=1}^L U^{\mathcal{O}_c}_{a_c,b_c},
 \qquad \text{ the joint projector }\Pi_a=\bigotimes_{c=1}^L\Pi_{a_c}=\proj{A_a},\\
 &\text{the signature }\sigma=(\sigma_1,\ldots,\sigma_L),
 \qquad\text{the verification procedure }\OSSVer^{\mathcal{O}}(a,M,\sigma)=\prod_{c=1}^L\OSSVer^{\mathcal{O}_c}(a_c,b_c,\sigma_c).
\end{aligned}
\end{equation}
\end{construction}

Combining the one-bit component guarantees as  given in
\Cref{thm:one-bit-correctness,fact:sz}, we get

\begin{theorem}[Long-message correctness and strong security,
adapted from {\cite[Theorem~30]{SZ25}} and
{\cite[Lemma~71]{HSVZ}}]
\label{lem:oss}
\normalfont
Let $p:\mathbb N\to\mathbb N$ be a polynomial in $\lambda$ and set the message space to be 
\(
 \mathcal M_\lambda=\{0,1\}^{\le p(\lambda)}.
\). Then, the scheme of \Cref{con:long-oss} satisfies
correctness and strong one-shot security as in
\Cref{def:oracle-oss} against every quantum adversary
making polynomially many total oracle queries for message space $\mathcal M_\lambda$.



\end{theorem}

Finally, we state the \emph{delayed collapsing} property enjoyed by \Cref{con:one-bit-oss}.

\begin{definition}[Delayed Collapsing Poperty of \cref{con:one-bit-oss}, adapted from \cite{bartusek2025classical}, Theorem 6]
\normalfont
Consider the oracle distribution of \Cref{con:sz-source}. For a two-stage QPT adversary
$\mathcal A=(\mathcal A_0,\mathcal A_1)$ and
$\beta\in\{0,1\}$, define the experiment
$\mathsf{DCollapse}^{\mathcal A}_{\beta}(1^\lambda)$ as follows.

\begin{enumerate}[label=(\roman*),leftmargin=*]
    \item Sample
    \[
        \mathcal O=(\mathcal P,\mathcal P^{-1},\mathcal D)
        \leftarrow \OSSGenOracles(1^\lambda).
    \]

    \item Give the first stage access to all three oracles and run
    \[
        (y,b,U,Z)
        \leftarrow
        \mathcal A_0^{\mathcal P,\mathcal P^{-1},\mathcal D}
        (1^\lambda),
    \]
    where $y\in\{0,1\}^r$ and $b\in\{0,1\}$ are classical,
    $U$ is a $k$-qubit register, and $Z$ denotes arbitrary residual
    quantum side information.

    We require $\mathcal A$ to be valid, meaning that  measuring $U$ in the computational basis produces a string $u$ satisfying
    \[
        \OSSVer^{\mathcal O}(y,b,u)=1.
    \]

    \item If $\beta = 0$, do nothing. Otherwise, measure the register $U$ in the computational basis. Let the resulting state on the registers $U,Z$ be denoted by $\rho_{UZ}^{(\beta)}$.

    \item Give $(y,b,U,Z)$ in state $\rho^{(\beta)}_{UZ}$ to $\mathcal{A}_1$, which has access only to the public inversion
    oracles:
    \[
        d\leftarrow
        \mathcal A_1^{\mathcal P,\mathcal P^{-1}}
        (y,b,U,Z).
    \]
    Output $d$.
\end{enumerate}

The delayed-collapsing advantage of $\mathcal A$ is
\[
    \mathsf{Adv}_{\mathsf{DCollapse}}^{\mathcal A}(\lambda)
    :=
    \left|
        \Pr[
            \mathsf{DCollapse}^{\mathcal A}_{0}(1^\lambda)=1
        ]
        -
        \Pr[
            \mathsf{DCollapse}^{\mathcal A}_{1}(1^\lambda)=1
        ]
    \right|.
\]
We say that the one-bit OSS scheme is \emph{delayed collapsing}
if, for every valid two-stage QPT adversary
$\mathcal A=(\mathcal A_0,\mathcal A_1)$,
\[
    \mathsf{Adv}_{\mathsf{DCollapse}}^{\mathcal A}(\lambda)
    \leq \negl(\lambda).
\]
\end{definition}

Bartusek et al. \cite{bartusek2025classical} prove \cref{con:one-bit-oss} satisfies delayed collapsing.

\subsection{Quantum fire}\label{subsec:cgs}
We first start with the definition of Quantum fire in the classical oracle model as done in \cite{CGS26}.
\begin{definition}
\label{def:native-fire}
\normalfont
Let $\mu\in\mathbb N$ be the security parameter.
Quantum inputs and outputs are specified by registers
written in capital letters.
The states corresponding to these registers may be mixed.
A \emph{public-key quantum fire scheme} consists of an
oracle-generation procedure and four $\QPT$ algorithms
\[
 \mathsf{Fire}
 =(\FireGenOracles,\FireSetup,\FireSpark,\FireVer,\FireClone).
\]

\begin{enumerate}[label=\arabic*.,leftmargin=*,itemsep=5pt]

 \item $\mathcal O\leftarrow\FireGenOracles(1^\mu)$:
 an oracle-generation procedure that samples a collection
 $\mathcal O$ of classical oracles queryable in superposition.
 All algorithms and adversaries have query access to them.
 The sampled oracles remain fixed throughout each experiment.

 \item $(pk,sk)\leftarrow\FireSetup(1^\mu)$:
 this algorithm is run by the issuing authority.
 It takes the security parameter and outputs a classical
 public verification key $pk$ and a classical secret key $sk$.

\item $(s,F)\leftarrow\FireSpark(sk)$:
 this algorithm is run by the issuing authority.
 It takes the secret key and outputs a classical serial
 number $s$ and a quantum flame register $F$.
 Each invocation uses independent randomness.

 \item $b\leftarrow\FireVer(pk,s,F),\quad b\in\{0,1\}$:
 this public algorithm verifies an alleged flame $F$
 under serial number $s$ and public key $pk$.

 \item $(F_0,F_1)\leftarrow\FireClone(pk,s,F)$:
 this public algorithm returns two flame registers
 $F_0,F_1$, associated with the serial number $s$.

\end{enumerate}
We require the following two correctness guarantees
from the scheme.

\medskip
\noindent\textbf{Verification correctness.}
We require
\begin{equation}
\label{eq:cgs-fire-verification-definition}
 \Pr\biggl[
  \FireVer(pk,s,F)=1:
  \substack{
   \mathcal O\leftarrow\FireGenOracles(1^\mu)\\
   (pk,sk)\leftarrow\FireSetup(1^\mu)\\
   (s,F)\leftarrow\FireSpark(sk)
  }
 \biggr]
 \geq 1-\negl(\mu).
\end{equation}

\medskip
\noindent\textbf{Cloning correctness.}
For every polynomially bounded $q=q(\mu)$, sample
\[
 \begin{aligned}
  \mathcal O&\leftarrow\FireGenOracles(1^\mu),\\
  (pk,sk)&\leftarrow\FireSetup(1^\mu),\\
  (s,F_0)&\leftarrow\FireSpark(sk).
 \end{aligned}
\]
For $i=1,\ldots,q(\mu)$, execute
\[
 (F_{i-1},F_i)
 \leftarrow
 \FireClone(pk,s,F_{i-1}).
\]
The assignment to $F_{i-1}$ replaces the input register
by the first returned flame register.
After all cloning operations, check
\[
 b_i\leftarrow\FireVer(pk,s,F_i),
 \qquad i=0,\ldots,q(\mu).
\]
We require
\begin{equation}
\label{eq:cgs-fire-cloning-definition}
 \Pr\left[
  \bigwedge_{i=0}^{q(\mu)}(b_i=1)
 \right]
 \geq 1-\negl(\mu).
\end{equation}
\end{definition}
All probabilities include oracle generation randomness,
algorithmic randomness and measurement outcomes.
\\~\\
We also require the following untelegraphability property.
\begin{definition}
\label{def:native-fire-untelegraphability}
\normalfont
Let
\(
 \mathcal A^{\mathcal O}
 =(\mathcal A^{\mathcal O}_1,\mathcal A^{\mathcal O}_2)
\)
be a pair of quantum adversaries making polynomially many
queries to $\mathcal O$ in total, across both parties and
the entire experiment.
Their initial joint state is separable across the two
parties and independent of the sampled oracles and the
challenger's setup randomness.
Shared classical randomness is permitted.

Define the game
\[
 \mathsf{TELEGRAPH}^{\mathsf{Fire}}_{\mathcal A^{\mathcal O}}
 (1^\mu)
\]
between a challenger and an adversary as follows.

\begin{enumerate}[label=\arabic*.,leftmargin=*,itemsep=6pt]

 \item The challenger samples
 \(
  \mathcal O\leftarrow\FireGenOracles(1^\mu),
  \:\:
  (pk,sk)\leftarrow\FireSetup(1^\mu),
 \)
 and gives $pk$ to both parties.

 \item For polynomially many rounds,
 $\mathcal A^{\mathcal O}_1$ may submit a spark query.
 For each adaptive query, the challenger independently samples
 \(
  (s,F)\leftarrow\FireSpark(sk)
 \)
 and gives $(s,F)$ to $\mathcal A^{\mathcal O}_1$.
 No flame register is supplied to $\mathcal A^{\mathcal O}_2$.

 \item $\mathcal A^{\mathcal O}_2$ sends a classical message
 $\mathcal M$ to $\mathcal A^{\mathcal O}_1$.

 \item Using $\mathcal M$ and its private state,
 $\mathcal A^{\mathcal O}_1$ sends a classical string $L$
 to $\mathcal A^{\mathcal O}_2$.

 \item Using $L$ and its retained state,
 $\mathcal A^{\mathcal O}_2$ outputs a classical serial
 number $s^*$ and a quantum register $F^*$.

 \item The challenger verifies
 \[
  b\leftarrow\FireVer(pk,s^*,F^*)
 \]
 and outputs $b$.
 The serial number $s^*$ need not have appeared during
 the spark-query phase.

\end{enumerate}

The scheme satisfies \emph{untelegraphability} if, for every
such adversary $\mathcal A^{\mathcal O}$ making polynomially
many oracle queries,
\begin{equation}
\label{eq:cgs-fire-security-definition}
 \Pr\bigl[
  \mathsf{TELEGRAPH}^{\mathsf{Fire}}_{\mathcal A^{\mathcal O}}
  (1^\mu)=1
 \bigr]
 \leq\negl(\mu).
\end{equation}

\end{definition}
All probabilities include oracle generation randomness,
algorithmic randomness and measurement outcomes.
\\
\\
Now, we describe the oracle structure and other tools required for the Quantum Fire construction given in \cite{CGS26} followed by the original construction.

\begin{construction}[Oracle collections, adapted from  
{\cite[Sections~6.2]{CGS26}}]
\label{con:cgs-source}
\label{con:cgs-oracles}
\label{fact:cgs-forward-check}
\normalfont
\textbf{Underlying OSS oracles.}
Sample an independent instance of oracle tuples as in \Cref{con:sz-source}
with security parameter $\mu$, i.e.,
\begin{equation}\label{eq:native-parameters}
 r=16\mu(\mu-1),\qquad
 n=r+24\mu,\qquad
 k=n.
\end{equation}
Consider, the predicate below derived
from these source oracles. For $y\in\{0,1\}^r$, $b\in\{0,1\}$,
and $u\in\{0,1\}^k$, set
\begin{equation}\label{eq:native-split-oracles}
\begin{aligned}
 \mathcal D_0(y,b,u)
   &=\ind[u\in W_y]\wedge\ind[u_1=b].
\end{aligned}
\end{equation}
Thus, $\mathcal D_0$ is the
one-bit verification predicate of
\eqref{eq:one-bit-verification}.

\par\smallskip\noindent
\textbf{A Common Oracle Interface: }
We combine the four oracles
$\mathcal P,\mathcal P^{-1},\mathcal D,\mathcal D_0$
into a single classical oracle $\mathcal G$.
An input to $\mathcal G$ specifies the operation and its
arguments. Accordingly, we define the query set
\[
\begin{aligned}
 \mathcal Q
 =&
 \bigl(\{\texttt{forward}\}\times\{0,1\}^{n}\bigr)\\
 &\sqcup
 \bigl(\{\texttt{inverse},\texttt{dual}\}
       \times\{0,1\}^{r}\times\{0,1\}^{k}\bigr)\\
 &\sqcup
 \bigl(\{\texttt{verify}\}
       \times\{0,1\}^{r}\times\{0,1\}
       \times\{0,1\}^{k}\bigr).
\end{aligned}
\]
To represent every query by a string of the same length,
choose a polynomially bounded $\ell\ge r+k+3$, put
$\mathcal Z=\{0,1\}^{\ell}$, and fix an injective encoding
\begin{equation}\label{eq:native-encoding}
 \enc:\mathcal Q\longrightarrow\mathcal Z.
\end{equation}
The encoding prepends two type bits to the arguments and
pads with zeros to length $\ell$.
The processes of encoding, decoding and checking if an encoding is valid 
are efficient.
\\
The oracle $\mathcal G$ evaluates the query operation as follows:
\begin{equation}\label{eq:native-dispatcher}
\begin{aligned}
 \mathcal G(\enc(\texttt{forward},x))
   &=\mathcal P(x),\\
 \mathcal G(\enc(\texttt{inverse},y,u))
   &=\mathcal P^{-1}(y,u),\\
 \mathcal G(\enc(\texttt{dual},y,u))
   &=\mathcal D(y,u),\\
 \mathcal G(\enc(\texttt{verify},y,b,u))
   &=\mathcal D_0(y,b,u).
\end{aligned}
\end{equation}
Inputs outside $\enc(\mathcal Q)$ return $\perp$.
\begin{fact}
\normalfont
    Direct access to $\mathcal G$ is reserved for the authority;
public access will be provided through the oracle $O_2$
defined below.
\end{fact}

\par\smallskip\noindent
\textbf{Flame preparation.}
Recall that \Cref{fact:sz-generation} prepares a uniform
superposition of verification keys and their coset states.
Let \[\FireGen\] denote this circuit with $Y$ left unmeasured
and its queries to $\mathcal P,\mathcal P^{-1}$ implemented
by the corresponding queries to $\mathcal G$.
Its output on the register $YU$ is
\begin{equation}\label{eq:K}
 \ket K_{YU}
 =2^{-r/2}\sum_{y\in\{0,1\}^{r}}
       \ket y_Y\ket{W_y}_U,
\end{equation}
and its remaining workspace is returned to zero.
This state is going to be the flame state of the Fire scheme.

\par\smallskip\noindent
\textbf{Public Oracles: }
A public query to $\mathcal G$ must be accompanied by a
tag. To define these tags, independently sample random
functions
\[
 H_0,H_1:\{0,1\}^{r+\ell}\longrightarrow\{0,1\}^{\mu},
 \qquad
 h:\mathcal Z\longrightarrow\{0,1\}^{\mu}.
\]
For a query $z\in\mathcal Z$, a tag consists of a
verification key $y$ and two proposed hash values:
\[
 p=(y,\tau_0,\tau_1)
 \in\{0,1\}^{r}\times\{0,1\}^{\mu}
                  \times\{0,1\}^{\mu}.
\]
It is valid precisely when both hash values are correct.
Explicitly, we write $\FireCheck$ for this predicate:
\begin{equation}\label{eq:native-tag-check}
 \FireCheck(z,p)
 =\ind[\tau_0=H_0(y\Vert z)]
  \wedge
  \ind[\tau_1=H_1(y\Vert z)].
\end{equation}
The predicate is zero on malformed queries or tags.
It is used inside the public oracles and is not itself
a public oracle.

Together, we consider the following oracle collection:
\begin{equation}\label{eq:CGS-oracles}
\begin{aligned}
 O_b(y,u,z)
 &=
 \begin{cases}
  H_b(y\Vert z),&\mathcal D_0(y,b,u)=1,\\
  \perp,&\text{otherwise},
 \end{cases}
 &&b\in\{0,1\},\\[2pt]
 O_2(z,p)
 &=
 \begin{cases}
  \mathcal G(z),&\FireCheck(z,p)=1,\\
  \perp,&\text{otherwise},
 \end{cases}\\[2pt]
 N(z,p)
 &=
 \begin{cases}
  h(z),&\FireCheck(z,p)=1,\\
  \perp,&\text{otherwise},
 \end{cases}\\[2pt]
 V(z,\eta)&=\ind[\eta=h(z)].
\end{aligned}
\end{equation}
For each bit $b$, the oracle $O_b$ supplies the hash value
$H_b(y\Vert z)$ upon receiving a valid one-bit OSS
signature $u$ on $b$ under $y$.
A valid tag then permits $O_2$ to return $\mathcal G(z)$
and $N$ to return $h(z)$.
The oracle $V$ performs a check. 
Malformed inputs return $\perp$, except that $V$ returns
zero. \\
\textbf{The public oracle collection} is
\begin{equation}\label{eq:native-public-oracles}
 \mathcal O=(\mathcal D,O_0,O_1,O_2,N,V).
\end{equation}
The functions $\mathcal G,\mathcal D_0,H_0,H_1,h$ are not directly queryable for public.

\end{construction}

\begin{remark}
\normalfont
    We note that, the validity of a supplied tag can
be checked using the public oracle $O_2$ by submitting the encoding of a forward query:
\begin{equation}\label{eq:cgs-forward-check}
 \FireCheck(\enc(\texttt{forward},x),p)
 =
 \ind\!\left[
  O_2(\enc(\texttt{forward},x),p)\ne\perp
 \right].
\end{equation}
Indeed, $\mathcal P(x)$ never returns $\perp$.
\end{remark}

\begin{remark}
\normalfont
    Now, to fully describe the Fire construction, we need to build the Tag generation machinery explicitly. Write $F=YU$ for the flame register.
The algorithms use a common procedure $\FireTag$ to obtain
the two hash responses needed for a tag in registers $R_0$ and $R_1$. For sake of continuity, we defer the detailed implementation of $\FireTag$ to \Cref{sec:coherent-tag-generation}.
\end{remark}

\par\smallskip\noindent
\textbf{Signing a classical query.}
Verification uses a classical request $z\in\mathcal Z$
and tests a response against $V$. For the verification procedure only, consider 
the routine $\FireHash(F,z)$. This obtains a response by
forming a tag and submitting it to $N$:
\begin{enumerate}[label=(\roman*),leftmargin=*,itemsep=3pt]
 \item Measure $Y$, obtaining $y$.
 Initialize $Z$ to $\ket z$ and all auxiliary registers
 to zero.

 \item Apply $\FireTag$ and measure $R_0,R_1$,
 obtaining $\tau_0,\tau_1$.

 \item Return $N(z,(y,\tau_0,\tau_1))$.
\end{enumerate}
Malformed tags give output $\perp$. We note, measuring the register $Y$, potentially destroys the alleged flame. But this is compatible with the definition \Cref{def:native-fire} where verification only returns a bit. Although, one party interested in maintaining the flame can keep a cloned copy of it.
\\~\\
We now are ready to describe the full Quantum Fire scheme of \cite{CGS}. Throughout this paper, we shall call this scheme \emph{CGS fire} and the fire state as \emph{CGS flame}.
\begin{construction}[Quantum Fire, adapted from 
{\cite[Section~2.2]{CGS}}]
\label{con:cgs-algorithms}
\label{con:native-test}
\label{fact:native-effect}
\normalfont
\par\smallskip\noindent
 
The authority
has direct access to $\mathcal G$ for spark, while
verification and cloning use only $\mathcal O$.
The procedures of $\mathsf{Fire}_{\mathrm{CGS}}$ are as follows.
\begin{enumerate}[label=\arabic*.,leftmargin=*,itemsep=5pt]
 \item $\mathcal O\leftarrow\FireGenOracles(1^\mu)$:
 sample the oracles as in \Cref{con:cgs-source} and return
 the public collection \eqref{eq:native-public-oracles}.
 Retain private access to $\mathcal G$ for the authority.

 \item $(pk,sk)\leftarrow\FireSetup(1^\mu)$:
 Sample an instance label $s\leftarrow\{0,1\}^{\mu}$,
 independently of the sampled oracles, and set $pk=s$. The
 authority's private oracle $\mathcal G$ is going to serve the purpose of its secret key $sk$.
 The oracles remain fixed throughout the rest of the experiment.

 \item $(s,F)\leftarrow\FireSpark(sk)$:
 run $\FireGen$ (recall \Cref{eq:K}) on fresh zero registers and return
 \[
  (s,F),\qquad F=YU,\qquad \rho_F=\proj K.
 \]

 \item $b\leftarrow\FireVer(pk,s,F)$:
 reject unless $s=pk$ and $F$ has the prescribed size.
 Otherwise sample $z\leftarrow\mathcal Z$ uniformly,
 compute $\eta\leftarrow\FireHash(F,z)$, and return
 $b=V(z,\eta)$.

 \item $(F_0,F_1)\leftarrow\FireClone(pk,s,F)$:
on malformed inputs, return two zero registers of the
prescribed flame size.

We write $F=YU$ for the input flame.
The intention is to prepare a new flame on fresh registers by following
the circuit of \Cref{fact:sz-generation}, leaving its
verification-key register $Y$ unmeasured. but note to that, one needs access to the oracle $\mathcal G$ which is not available publicly. Instead, we 
implement its forward and inverse oracle calls using
the public procedure below.

\begin{equation}\label{eq:native-generation-query}
 \widetilde U_{\mathcal G}^{\mathcal O}
 =(\FireTag)^\dagger U_{O_2}\FireTag.
\end{equation}

where $U_{O_2}$ is the unitary implementing the oracle $O_2$. The details are given in, \Cref{sec:flame-cloning-details}.

After completing the preparation circuit, return the
original flame register as $F_0$ and the newly prepared
flame register as $F_1$.
Discard all remaining workspace.
\end{enumerate}
\end{construction}

Now, We state the correctness and security guarantees of CGS fire
at security parameter $\mu$.

\begin{theorem}[Correctness and untelegraphability,
adapted from {\cite[Theorems~16--17]{CGS}}]
\label{fact:cgs}
\label{cor:cgs-fire}
\label{thm:cgs-correctness}
\normalfont
Consider \Cref{con:cgs-algorithms} with the oracles of
\Cref{con:cgs-source} at security parameter $\mu$.
The construction satisfies verification correctness and cloning
correctness in \Cref{def:native-fire}, and untelegraphability in
\Cref{def:native-fire-untelegraphability}.

\end{theorem}

The next two remarks record stronger correctness
guarantees for the cloning and coherent tag circuits.

\begin{remark}[Strong Cloning Correctness, adapted from {\cite[Section 9.4]{CGS}}]
\label{rem:native-cloning}
\normalfont
Sample $\mathcal O\leftarrow\FireGenOracles(1^\mu)$ and
$(pk,sk)\leftarrow\FireSetup(1^\mu)$, and put $s=pk$.
Let $\ket K$ be the honest CGS flame in \eqref{eq:K} for these
oracles. The strong cloning guarantee, in trace-distance form,
gives
\begin{equation}\label{eq:native-cloning-distance}
 \mathbb E\!\left[
  D( \FireClone(pk,s,\proj K),\proj K^{\otimes2})
 \right]
 \le\negl(\mu).
\end{equation}
The expectation is over oracle generation and setup.
\end{remark}

\begin{remark}[Coherent tag correctness,
adapted from {\cite[Section 9.4]{CGS}}]
\label{rem:native-coherent-tag}
\label{lem:cgs-tags}
\normalfont
For each oracle and setup sample in \Cref{rem:native-cloning},
let $\Pi$ be the diagonal projector on $(Z,Y,R_0,R_1)$
onto the basis values satisfying
\[
 \FireCheck(z,(y,\tau_0,\tau_1))=1.
\]
It acts as the identity on all other registers.
For every finite-dimensional register $E$, extend $\FireTag$
and $\Pi$ by the identity on $E$. Then
\begin{equation}\label{eq:native-coherent-tag}
 \mathbb E\!\left[
  \sup_{\|\ket\psi_{ZE}\|=1}
  \left\|
   (I-\Pi)\FireTag
   \bigl(\ket K\otimes\ket\psi_{ZE}\otimes\ket0\bigr)
  \right\|^2
 \right]
 \le\negl(\mu).
\end{equation}
Here $\FireTag$ is the circuit in \eqref{eq:native-tag-circuit},
and $\ket0$ comprises its response and signing workspace registers.
The expectation is over oracle generation and setup.
\end{remark}

\begin{remark}[Verification effect]
\label{rem:native-verification-effect}
\normalfont
Fix the oracles, $pk$, and $s=pk$.
There is an operator $B$ such that
\begin{equation}\label{eq:B}
 0\preceq B\preceq I,
 \qquad
 \Pr[\FireVer(pk,s,\rho_F)=1]=\operatorname{Tr}(B\rho_F)
\end{equation}
for every density operator $\rho_F$ on the prescribed flame register.
\end{remark}
\section{Quantum Fire with Classical Cloning Key}
\label{sec:classical-secret-key-Q-Fire}
In this section we define and present a privately verifiable variant of quantum fire that only allows delegated cloners with a classical token to sign. Cloning is impossible without such a token. We begin by formally defining this variant, as our syntax slightly deviates from that of \cref{def:native-fire}.

\begin{definition}[Privately Verifiable Quantum Fire with Classical Cloning Key]
    \normalfont
    Let $\mu \in \mathbb N$ be the security parameter. A \emph{Privately Verifiable Quantum Fire with Classical Cloning Key} Scheme consists of an oracle generation procedure and a tuple of three $\QPT$ algorithms \[
        \mathsf{PVQF} = (\mathsf{PVQF.GenOracles}, \mathsf{PVQF.Spark}, \mathsf{PVQF.Clone}, \mathsf{PVQF.Ver})
    \] which work as follows:

    \begin{enumerate}
        \item $\mathcal{O}, \aux \leftarrow \mathsf{PVQF.GenOracles}(1^{\mu}):$ an oracle-generation procedure that samples a collection $\mathcal{O}$ of classical oracles queryable in superposition.

        \item $s,F,sk_{clone},sk_{ver} \leftarrow \mathsf{PVQF.Spark}^{\mathcal{O}, \aux}(1^{\mu}):$ On input $1^{\mu}$, this algorithm outputs a flame state $F$ and serial number $s$, as well as secret keys $sk_{clone}$ and $sk_{ver}$ for cloning and verifying respectively.

        \item $F_0, F_1 \leftarrow \mathsf{PVQF.Clone}^{\mathcal{O}}(sk_{clone}, s, F)$: On input a secret cloning key $sk_{clone}$, serial number $s$, and flame state $F$, the cloning algorithm outputs $2$ flame registers $F_0, F_1$.

        \item $b \leftarrow \mathsf{PVQF.Ver}^{\mathcal{O}}(s, F, sk_{ver})$: On input a secret verification key $sk_{ver}$, a serial number $s$, and an alleged flame state $F$, outputs $0$ or $1$.
    \end{enumerate}
\end{definition}\label{def:QFire-Classical-Cloning-Key}

\noindent Verification correctness, cloning correctness, and untelegraphability are analogous to \cref{def:native-fire}.  

\begin{construction}[Privately Verifiable Quantum Fire with Classical Cloning Key]
\label{con:PVQF-algorithms}
\normalfont
\par\smallskip\noindent
We now proceed to give a construction for this variant, which is  essentially a slight modification of the quantum fire scheme given in \cite{CNS26}. Similar to their scheme, this construction only allows for $O(\log(\mu))$ many clones, where $\mu$ denotes the security parameter. Here, all OSS operations refer to the long message variant described in \cref{con:long-oss}.

\begin{enumerate}
    \item $\mathsf{PVQF.GenOracles}(1^{\mu})$: 
    \begin{itemize}
        \item Sample $P,P^{-1},D$ by running the oracle generator given in \cref{con:sz-source} on input $1^{\mu}$.

        \item Set $\mathcal{O} = P,P^{-1}$ and $\aux = D$. Output $\mathcal{O}, \aux$. 
    \end{itemize}

    \item $\mathsf{PVQF.Spark}^{\mathcal{O},\aux}(1^{\mu})$: 
        \begin{itemize}
            \item Parse $\mathcal{O}$ as $P,P^{-1}$ and run the key generation procedure

            \item Set the $\mathsf{OSS}$ public key as $pk$ and the $\mathsf{OSS}$ signing key as  $\ket{sk}$

            \item Set $s = pk$, $F = \ket{sk}$, $D = sk_{clone} = sk_{ver}$. Output $s,F,sk_{clone}, sk_{ver}$.
        \end{itemize}  

    \item $\mathsf{PVQF.Clone}^{\mathcal{O}}(sk_{clone}, s, F)$: 
        \begin{itemize}
            \item Parse $sk_{clone}$ as the oracle $D$ and parse $F$ as
            \[
                (\mathsf{history},pk_{\mathrm{cur}},
                \ket{sk_{\mathrm{cur}}}).
            \]

            \item Run
            \[
                (pk_0,\ket{sk_0}),
                (pk_1,\ket{sk_1})
                \leftarrow
                \mathsf{OSS.KeyGen}^{\otimes 2}(1^\lambda).
            \] 
            
            and compute 
            \[
                \sigma
                \leftarrow
                \mathsf{OSS.Sign}^{D}
                \left(
                    \ket{sk_{\mathrm{cur}}},
                    pk_0\Vert pk_1
                \right).
            \]

            \item Set
            \[
                F_0
                    :=
                    \left(
                    \mathsf{history}\Vert(pk_0,pk_1,\sigma),
                    pk_0,
                    \ket{sk_0}
                \right)
            \]
            and
            \[
                F_1
                    :=
                    \left(
                    \mathsf{history}\Vert(pk_0,pk_1,\sigma),
                    pk_1,
                    \ket{sk_1}
                    \right).
            \]

            \item Output $F_0, F_1$.
        \end{itemize}

    \item $\mathsf{PQVF.Verify}^{\mathcal{O}}(s,F,sk_{ver})$: 
        \begin{itemize}
            \item Parse $s := pk_{root}$, $\mathcal{O}:= P,P^{-1}$ and 
            
            \[
                F :=
                    \left(
                    \mathsf{history},
                    pk_{\mathrm{cur}},
                    \ket{sk_{\mathrm{cur}}}
                \right),
            \] where 
            \[
        \mathsf{history}
            =
                \left(
                pk_0^{(1)},pk_1^{(1)},\sigma_1,
                \ldots,
                pk_0^{(T)},pk_1^{(T)},\sigma_T
                \right).
            \]

            \item  Define
            \[
                pk_0^{(0)}=pk_1^{(0)}:=pk_{\mathrm{root}}.
            \] For every \(t\in\{1,\ldots,T\}\), compute
            \[
                a_t
                    :=
                    \mathsf{OSS.Verify}_{pk_0^{(t-1)}}
                    \left(
                     pk_0^{(t)}\Vert pk_1^{(t)},\sigma_t
                 \right)
                \vee
                    \mathsf{OSS.Verify}_{pk_1^{(t-1)}}
                    \left(
                        pk_0^{(t)}\Vert pk_1^{(t)},\sigma_t
                \right).
            \]

            \item If \(T=0\), reject unless
            \[
                pk_{\mathrm{cur}}=pk_{\mathrm{root}}.
            \]
            If \(T\geq 1\), reject unless
            \[
                pk_{\mathrm{cur}}=pk_0^{(T)}
                \quad\text{or}\quad
                pk_{\mathrm{cur}}=pk_1^{(T)}.
            \]

            \item Run the coset verification procedure outlined in \cref{lem:token-test} for all single bit OSS public key/secret key pairs in $(pk_{cur}, \ket{sk_{cur}})$. Set $a := 1$ if all these tests pass, otherwise set $a:= 0$.
                
            \item Output
                \[
                    a\wedge\left(\bigwedge_{t=1}^{T}a_t\right).
                \]
        \end{itemize}
     
\end{enumerate}
\end{construction}

Correctness follows from OSS correctness. 

\begin{theorem}[Security]
    \cref{con:PVQF-algorithms} is a $c\log(\mu)$-clonable untelegraphable quantum fire scheme that is uncloneable for non-delegated cloners. 
\end{theorem}
\begin{proof}
    Untelegraphability follows from the same reasoning given for Construction 1a in \cite{CNS26}. Therefore, we now focus on proving unclonability for non-delegated cloners (i.e. ones without access to $sk_{clone}$).
    \\
    \\
    Let
    \[
    F=(\mathsf{history},y,\ket{sk_y})
    \]
    be the honestly generated flame given to \(\mathcal A\), where
    \[
    \mathsf{history}
    =
    \left(
      pk_0^{(1)},pk_1^{(1)},\sigma_1,
      \ldots,
      pk_0^{(T)},pk_1^{(T)},\sigma_T
    \right)
    \]
    is its classical signature chain, $y$ is its terminal public key, and
    $\ket{sk_y}$ is its terminal signing token.

    \noindent For every $t$, the signature \(\sigma_t\) authenticates the message
    \[
    M_t:=pk_0^{(t)}\Vert pk_1^{(t)}
    \]
    under one of the two public keys from the previous level. 
    \\
    \\
    \noindent Let $y_0, y_1$ be the terminal public keys at the end of the adversary's alleged clones. Let us condition on any of the following events occurring when running $\mathcal{A}$ on a flame state: namely that (1)$y_0 = y_1$, (2).$y_0 \neq y_1$ and one of these keys was used during a previous clone and is in $\mathsf{history}$, (3). $y_0 \neq y_1$ annd both public keys were leaves in $\mathsf{history}$ but were never used for a previous clone, (4). $y_0 \neq y_1$, neither were used for a previous clone, and at least of them is new (i.e. not in $\mathsf{history}$). In each case, we construct an adversary
    $(\mathcal B^{P,P^{-1},D})$ against $\mathsf{OSS}$ security.
    \\
    \\
    For each of the three cases, the reduction starts out the same. It honestly generates an input flame state. Whenever it applies
    the honest cloning algorithm and obtains two child signing keys, it continues the
    honest path using one child and privately retains the unused $\mathsf{OSS}$ key. It clones an arbitrary number (within the $O(\lambda)$ limit) of times. Let $F$ be the flame state the reduction passes to $\mathcal{A}^{P,P^{-1}}$. 
    \\
    \\
    \noindent Run
    \[
    (g_0,g_1)\leftarrow\mathcal A^{P,P^{-1}}(F).
    \]
    Parse the terminal public keys and alleged terminal tokens of the two outputs as
    \[
    (pk_0 = y_0,\rho_0)
    \qquad\text{and}\qquad
    (pk_1 = y_1,\rho_1).
    \]
    Write each string-signing public key as
    \[
    pk_i=(pk_{i,1},\ldots,pk_{i,L})
    \]
    and its corresponding signing token as an $L$-register state.
    \\
    \\
    Since \(g_0\) and \(g_1\) pass verification, their classical signature chains
    pass all checks and their terminal signing key registers pass the full
    coset verification procedure. Let
    \[
    \Pi_{y_i}
    =
    \bigotimes_{\ell=1}^{L}
    \ket{W_{y_{i,\ell}}}\!\bra{W_{y_{i,\ell}}}
    \]
    be the accepting projector for the OSS signing key under $y_i$. Conditioned on \cref{lem:token-test} accepting (recall we run this in flame state verification), the post-verification state is
    $\bigotimes_{\ell=1}^{L}\ket{W_{y_{i,\ell}}}$.
    \\
    \\
    \noindent We consider three cases.
    
    \begin{enumerate}
        \item \emph{The two output flames have the same terminal public key.}
    
        Suppose
        \[
        y := y_0 = y_1.
        \]
        Run the full token test on both alleged terminal tokens. Conditioned on both
        tests accepting, the reduction obtains two valid signing tokens under the
        same public key $y$.
    
        Fix any component \(\ell\in\{1,\ldots,L\}\). Let
        \[
        y:=y_{0,\ell}=y_{1,\ell}.
        \]
        Using the \(\ell\)-th component of the first accepted token, compute
        \[
        u_0
        \leftarrow
        \mathsf{Sign}^{P,P^{-1},D}
        \left(
          \ket{W_{y_0,\ell}},0
        \right).
        \]
        Using the \(\ell\)-th component of the second accepted token, compute
        \[
        u_1
        \leftarrow
        \mathsf{Sign}^{P,P^{-1},D}
        \left(
          \ket{W_{y_1, \ell}},1
        \right).
        \]
        By correctness,
        \[
        \mathsf{OSS.Verify}(y,0,u_0)=1
        \]
        and
        \[
        \mathsf{OSS.Verify}(y,1,u_1)=1.
        \]
        Therefore,
        \[
        (y,0,u_0,1,u_1)
        \]
        violates single-bit OSS unforgeability.
    
        \item \emph{One of the output flames terminates at an internal public key.}
    
        Suppose that, for some $i\in \{0,1\}$, the terminal public key $y_i$ is
        an internal public key in the graph determined by
        $\mathsf{history}$. Thus, $\mathsf{history}$ contains a message
        \[
        M=(M_1,\ldots,M_{L})
        \]
        and a signature
        \[
        \sigma=(u_1,\ldots,u_{L})
        \]
        satisfying
        \[
        \mathsf{OSS.Verify}(y_i, M,\sigma)=1.
        \]
        Componentwise, for every $\ell$, we have
        \[
        \mathsf{OSS.Verify}
        \left(
          y_{i,\ell},M_\ell,u_\ell
        \right)
        =1.
        \]
    
        Apply the full token test to \(\rho_i\). Conditioned on acceptance, its
        \(\ell\)-th component is the valid signing token
        \(\ket{W_{y_{i,\ell}}}\). Fix any component \(\ell\) and compute
        \[
        u'
        \leftarrow
        \mathsf{Sign}^{P,P^{-1},D}
        \left(
          \ket{W_{y_{i,\ell}}},
          1-M_\ell
        \right).
        \]
        Then
        \[
        \mathsf{OSS.Verify}
        \left(
          y_{i,\ell},1-M_\ell,u'
        \right)
        =1.
        \]
        Hence,
        \[
        \left(
          y_{i,\ell},
          M_\ell,u_\ell,
          1-M_\ell,u'
        \right)
        \]
        is a valid single-bit OSS forgery.
    
        \item \emph{The two output flames terminate at distinct leaf public keys.} 
    
        Suppose
        \[
        y_0\neq y_1
        \]
        and we are in case (3) described above.
    
        Note that $\mathcal{A}$ receives only one terminal public key/signing key pair. Let $y$ be the terminal public key. Then at least one of $y_0, y_1$ is different from $y$.
    
        Since $y_0\neq y_1$, at least one of the two output terminal keys differs from the input terminal key $y$. Choose $i^\star\in\{0,1\} \quad\text{such that}\quad y_{i^\star}\neq y.$
        
        Since $y_{i^\star}$ is an old leaf distinct from the current terminal key, it is an unused sibling leaf. When constructing the input flame, $\mathcal{B}$ generated and privately retained its honest signing token $\ket{sk_{y_{i^\star}}}$.
        \\
        \\
        Assuming the adversarial clone successfully verifies, we know $\rho_{i^*}$ is a valid signing key for $y_{i^*}$. Fix a
        component $\ell$, and write
        $y_{i^*}=(y_{i^*,1},\ldots,y_{i^*,L}).$
        Use the adversarial token component to compute
        \[
        u_0
        \leftarrow
        \mathsf{Sign}^{P,P^{-1},D}
        \left(
          \ket{W_{y_{i^*,\ell}}}^{\,\mathsf{adv}},0
        \right),
        \]
        and use the corresponding retained honest sibling-token component to compute
        \[
        u_1
        \leftarrow
        \mathsf{Sign}^{P,P^{-1},D}
        \left(
          \ket{W_{y_{i^*,\ell}}}^{\,\mathsf{honest}},1
        \right).
        \]
        Both signatures verify under the same single-bit public key:
        \[
        \mathsf{OSS.Verify}(y_{i^*,\ell},0,u_0)=1
        \]
        and
        \[
        \mathsf{OSS.Verify}(y_{i^*,\ell},1,u_1)=1.
        \]
        Thus,
        \[
        (y_{i^*,\ell},0,u_0,1,u_1)
        \]
        violates single-bit OSS unforgeability.
    \end{enumerate}
    
    \noindent Now we focus on the cases where the adversary returns two flame states where the terminal public keys are different from each other \emph{and} at least one of them has not appeared in $\mathsf{history}$ yet.
    \\
    \\
    Let the flame given to \(\mathcal A\) be
    $F=(\mathsf{history},y,\ket{sk_y})$, where $y=(y_1,\ldots,y_L), \ket{sk_y} = \bigotimes_{\ell=1}^{L}\ket{W_{y_\ell}}$.
    \\
    \\
    Choose an accepted output flame whose terminal public key is new, meaning that it did not appear in the flame before the flame was given to $\mathcal{A}.$ Traverse its authenticated chain from the root, and let $(p,\widehat M,\widehat\sigma)$
    be the first authenticated transition that does not occur in the original history, where a transition is identified by its parent key $p$ and signed message $\widehat M$. Because this is the first new transition, $p$ is a public key appearing in the original tree. If $p$ is an old internal key, then the original history already contains a signature under $p$ of a different message, immediately yielding an OSS forgery. If $p$ is an unused sibling leaf distinct from the current terminal key $y$, then $\mathcal{B}$ retained its honest signing token and can again obtain an OSS forgery.
    It remains to consider the case $p = y$, where $y$ is the current terminal public key of the flame given to $\mathcal{A}$. Although $y$ appears in the old tree, it has no old outgoing transition, and $\mathcal{B}$ gave its only signing token to $\mathcal A$. We handle this case using delayed collapsing.
    \\
    \\
    For a single-bit public key $y_i$, define $\mathsf{CosetTest}(y_i,\rho)$ as follows:
    \begin{enumerate}
        \item measure $\rho$ in the computational basis, obtaining $u$.
    
        \item accept if $P^{-1}(y_i,u)\neq\perp$.
    
    \end{enumerate}
    
    \noindent For a public key verifying longer message signatures, we can apply the same test component wise. Note that replacing the coset verification test performed in $\mathsf{PVQF.Verify}$ with $\mathsf{CosetTest}$ increases the winning probability for $\mathcal{A}$.
    \\
    \\
    Let $\mathsf{E_y}$ be the event that an accepted
    output contains a new authenticated transition whose parent is $y$,
    when its terminal signing key is checked using the component-wise
    $\mathsf{CosetTest}$. It suffices to prove that
    $\Pr[\mathsf{E_y}] \leq \mathsf{negl}(\lambda).
    $
    \\
    \\
    Define hybrids $\mathsf H_0,\ldots,\mathsf H_{L}$ where, in $\mathsf{H_j}$, the first $j$ registers $\ket{W_{y_1}},\ldots,\ket{W_{y_j}}$
    of the current terminal signing key are measured in the computational basis before the flame is given to $\mathcal{A}$. Thus,
    $\mathsf{H_0}$ is the unmeasured experiment and, in
    $\mathsf{H_L}$, the entire current signing key has been measured. Note that delayed collapsing implies that every pair of consecutive hybrids are computationally indistinguishable to each other. It therefore remains to prove that event $\mathsf{E_y}$ occurs in hybrid $\mathsf{H_L}$ with negligible probability.
    \\
    \\
    Suppose otherwise. We now show how to design the OSS adversary $\mathcal{B}$ for this case. First, $\mathcal{B}$ prepares a flame state
    $$ F=(\mathsf{history},y,\ket{sk_y})
    $$
    of the desired depth. It then generates two fresh string-OSS key pairs
    $$
    (pk_0^*,\ket{sk_0^*}),
    (pk_1^*,\ket{sk_1^*})
    \leftarrow
    \mathsf{OSS.KeyGen}^{\otimes 2}$$
    and retains both child signing keys. Define
    $$
    M^*:=pk_0^*|| pk_1^* =(M_1^*,\ldots,M_L^*).$$
    \\
    \\
    For every $\ell\in\{1,\ldots,L\}$, $\mathcal B$ uses $D$ to
    compute
    $$
    u_\ell^*
    \leftarrow
    \mathsf{OSS.Sign}^{P,P^{-1},D}
    (\ket{W_{y_\ell}},M_\ell^*)$$ and sets
    $$ \sigma^*=(u_1^*,\ldots,u_L^*).$$
    \\
    \\
    The reduction replaces the current signing-key register
    $\ket{sk_y}$ by the computational-basis state
    $\ket{\sigma^*}$, and gives the resulting flame to
    $\mathcal A^{P,P^{-1}}$.
    \\
    \\
    Note that $pk_0^*,pk_1^*$ are uniformly distributed, and therefore $M^*$ is a
    uniform message. Moreover, conditioned on
    $(u_{\ell}^*)_1=M_\ell^*$, the signing algorithm outputs a uniformly
    random element of the corresponding half of $W_{y_\ell}$. Averaging
    over the uniform bit $M_\ell^*$, $u_\ell^*$ is therefore uniform
    over $W_{y_\ell}$. Therefore,
    $\ket{\sigma^*}$
    has exactly the same distribution as the one obtained by measuring
    $\ket{sk_y}$ in the computational basis. Thus, the state given to $\mathcal{A}$ is distributed exactly as in $\mathsf{H_L}$.
    \\
    \\
    Suppose that the event $\mathsf{E_{y}}$ occurs. Let $(y, \widehat{M}, \sigma)$ be the first authenticated transition in the $\mathsf{history}$ leaving $y$, where $\widehat{M} = (\widehat{pk_0}||\widehat{pk_1}) = (\widehat{M_1},\ldots, \widehat{M_L})$ and $\widehat{\sigma} = (\widehat{\mu_1}\ldots \widehat{\mu_L})$.  
    There are two cases.
    \\
    \\
    \noindent\emph{Case 1: $\widehat M\neq M^*$.}
    
    \noindent Choose an index $\ell$ satisfying $\widehat M_\ell\neq M_\ell^*$. By construction, $\mathsf{OSS.Verify}
    (y_\ell,M_\ell^*,u_\ell^*)=1$
    whereas the validity of the adversary's clone $\mathsf{OSS.Verify}
    (y_\ell,\widehat M_\ell,\widehat u_\ell)=1.$
    Therefore $\bigl(
    y_\ell,
    M_\ell^*,u_\ell^*,
    \widehat M_\ell,\widehat u_\ell
    \bigr)$
    is a valid single-bit OSS forgery.
    \\
    \\
    \noindent\emph{Case 2: $\widehat M=M^*$.}
    \\
    \\
    In this case,
    $(\widehat{pk}_0,\widehat{pk}_1) = (pk_0^*,pk_1^*)$. Now if the history ends here, the verifier's final-key check
    implies that the terminal public key is $pk_a^*$ for some
    $a = 0,1$. Write
    $pk_a^*=(z^*_{a,1},\ldots,z^*_{a,L}).$
    Since the terminal register passes the componentwise
    \(\mathsf{CosetTest}\), the reduction records, for any coordinate $\ell$, $w_{\ell}$ such that $P^{-1}(z^*_{a,\ell},w_\ell)\neq\perp.$
    \\
    \\
    Let $b=(w_\ell)_1$.
    Then $\mathsf{OSS.Verify}
    (z^*_{a,\ell},b,w_\ell)=1.$
    The reduction retained the honest token component
    $\ket{W_{z^*_{a,\ell}}}$. Using $D$, it computes
    $$ w_\ell' \leftarrow
    \mathsf{OSS.Sign}^{P,P^{-1},D}
    (\ket{W_{z^*_{a,\ell}}},1-b).$$
    It follows that
    $$ \mathsf{OSS.Verify}
    (z^*_{a,\ell},1-b,w_\ell')=1.$$
    Thus, \(w_\ell,w_\ell'\) are signatures of two different bits under
    the same single-bit OSS public key.
    \\
    \\
    Finally, note that if the event $\mathsf{E_y}$ occurs but the adversary decides to extend the chain and clone further, the same reduction still works because the clone must contain a signature with respect to $y' = pk^*_a$, and the reduction can use its retained $\ket{sk^*_a}$ to sign another message.
    \\
    \\
    We conclude that every occurrence of
    \(\mathsf{E_y}\) in \(\mathsf H_L\) yields two valid
    signatures of different bits under a common single-bit OSS public key.
    Therefore,
    $$ \Pr[\mathsf H_L:\mathsf{E_y}]
    \leq\mathsf{negl}(\lambda).
    $$
    The delayed-collapsing hybrid argument then gives
    \[
    \Pr[\mathsf H_0:\mathsf{E_y}]
    \leq\mathsf{negl}(\lambda).
    \]
    Finally, because acceptance by the full token test implies acceptance
    by \(\mathsf{CosetTest}\), the corresponding probability in the
    original verification experiment is also negligible.
    \\
    \\
    We have covered all possible types of clones and reduced to OSS for each one. This lets us conclude that $\mathcal{A}$ can clone with negligible winning probability.
\end{proof}
\section{Definition: Torch-Fire}
\label{sec:definition}
In this section we give the formal definition of the torch-fire object in a classical oracle model.
\begin{definition}
\label{def:primitive}
\normalfont
Let $\lambda\in\mathbb N$ be the security parameter. Quantum inputs and outputs are specified by registers written in Capital letters. The states corresponding to these registers may be in mixed states. A \emph{torch-fire
scheme} consists of an
oracle-generation procedure and four $\QPT$ algorithms
\[
  \mathsf{TF}=(\TFGenOracles,\TFSetup,\TFSpark,\TFVer,\TFClone).
\]
\begin{enumerate}[label=\arabic*.,leftmargin=*,itemsep=5pt]
 \item $\mathcal O\leftarrow\TFGenOracles(1^\lambda)$:
 an oracle-generation procedure that samples a collection
 $\mathcal O$ of classical oracles query-able in superposition. All algorithms
and adversaries have query access to them. The sampled oracle
 remain fixed throughout each experiment.
\item $\bigl(pk,sk,(t_1,T_1),\ldots,(t_R,T_R)\bigr)
  \leftarrow\TFSetup(1^\lambda,1^R)$: This algorithm is run by the issuing authority. It takes in the security parameter and a polynomially bounded $R=R(\lambda)$ and outputs the public verification key $pk$, secret key $sk$ and the torch states $\{(t_i,T_i)\}_{i\in[R]}$. 
The values $pk,sk,t_1,\ldots,t_R$ are classical, and $T_1,\ldots,T_R$
are quantum torch registers. 

\item $(s,F)\leftarrow\TFSpark(sk)$: 
This algorithm is run by the issuing authority. It takes in the secret key and outputs a classical serial
number $s$ and a quantum flame register $F$. Each invocation uses
independent randomness.

\item $
  b\leftarrow\TFVer(pk,s,F),\:\: b\in\{0,1\}$: 
This public algorithm verifies an alleged flame $F$ under serial number $s$ and public key $pk$.

\item $(T',F_0,F_1)\leftarrow\TFClone(pk,t,s,T,F)$: 
This public algorithm returns a torch register $T'$, associated with the torch identifier 
$t$, and two flame registers $F_0,F_1$, associated with the serial number $s$.

\end{enumerate}
We require the following two correctness guarantees from the scheme.

\medskip
\noindent\textbf{Verification correctness.}
For every polynomially bounded $R$. 
We require
\begin{equation}
\label{eq:correct-ver}
  \Pr\biggl[
    \TFVer(pk,s,F)=1: 
    \substack{\mathcal{O}\leftarrow \TFGenOracles(1^\lambda)\\
    \bigl(pk,sk,(t_\ell,T_\ell)_{\ell\in[R]}\bigr)
    \leftarrow\TFSetup(1^\lambda,1^R)\\
  (s,F)\leftarrow\TFSpark(sk)}
  \biggr]\geq 1-\negl(\lambda).
\end{equation}

\medskip
\noindent\textbf{Cloning correctness.}
For every polynomially bounded $R,q$ sample
\[
\begin{aligned}
\mathcal{O}&\leftarrow \TFGenOracles(1^\lambda)\\
  \bigl(pk,sk,(t_\ell,T_\ell)_{\ell\in[R]}\bigr)
    &\leftarrow\TFSetup(1^\lambda,1^R)\\
  (s,F_0)&\leftarrow\TFSpark(sk)
\end{aligned}
\]
Now for any $j\in[R]$, set $T^{(0)}=T_j$. For $i=1,\ldots,q(\lambda)$, execute
\[
  \bigl(T^{(i)},F_{i-1},F_i\bigr)
  \leftarrow
  \TFClone\bigl(pk,t_j,s,T^{(i-1)},F_{i-1}\bigr).
\]
Note, the assignment to $F_{i-1}$ replaces the input register by the first
returned flame register. After all cloning operations, check
\[
  b_i\leftarrow\TFVer(pk,s,F_i),\qquad i=0,\ldots,q(\lambda).
\]
We require for every $j\in[R]$,
\begin{equation}
\label{eq:correct-clone}
  \Pr\left[
    \bigwedge_{i=0}^{q(\lambda)}(b_i=1)
  \right]\geq 1-\negl(\lambda).
\end{equation}
\end{definition}
All probabilities include oracle generation randomness, algorithmic randomness and measurement
outcomes. 

\begin{definition}[Global non-proliferation]
\label{def:global}
\normalfont
Let $R,m,N$ be polynomially bounded functions of $\lambda$, with
$0\leq m\leq R$, and let the adversary 
$\mathcal A^{\mathcal O}=(\mathcal A^{\mathcal O}_0,\mathcal A^{\mathcal O}_1,\ldots,\mathcal A^{\mathcal O}_N)$ be a tuple
of quantum polynomial-time algorithms. Define the game
\[
  \mathsf{SPLIT}^{\mathsf{TF}}_{\mathcal A^{\mathcal O}}
  (1^\lambda,1^R,m,N)
\]
between a challenger and an adversary as follows.

\begin{enumerate}[label=\arabic*.,leftmargin=*,itemsep=6pt]
\item The challenger samples \(\mathcal O\leftarrow \TFGenOracles\) and 
\(
  \bigl(pk,sk,(t_j,T_j)_{j\in[R]}\bigr)
  \leftarrow\TFSetup(1^\lambda,1^R)
\)
and gives $pk$ to $\mathcal A^{\mathcal O}_0$.

\item  The challenger answers polynomially many
spark queries from $\mathcal A^{\mathcal O}_0$. For each adaptive query, it
independently samples
\(
  (s,F)\leftarrow\TFSpark(sk)
\)
and gives $(s,F)$ to $\mathcal A^{\mathcal O}_0$. 

At any one point during this phase, $\mathcal A^{\mathcal O}_0$ submits a set
$J\subseteq[R]$ with $|J|=m$ and receives $(T_j)_{j\in J}$ from the challenger.
No other torch registers are supplied.

\item The algorithm $\mathcal A^{\mathcal O}_0$ outputs registers
\(
  B_1,\ldots,B_N
\)
and terminates. Register $B_i$ is assigned to $\mathcal A^{\mathcal O}_i$. Their
joint state may be arbitrarily entangled.

From this point onward, any classical or quantum pairwise communication between the parties $(\mathcal A^{\mathcal O}_1,\mathcal A^{\mathcal O}_2,\ldots,\mathcal A^{\mathcal O}_N)$ is
forbidden. No further
torch or flame issuance is permitted, except for the challenge inputs
specified in Step~4.

\item After the split, the challenger independently
samples
\(
  (s_i,F_i)\leftarrow\TFSpark(sk),\:\: i\in[N].
\)
It gives $(pk,s_1,\ldots,s_N)$ to every party and gives $F_i$ only to
$\mathcal A^{\mathcal O}_i$.

\item The party $i, \:\: i\in [N]$ executes
\[
  (F_{i,0}',F_{i,1}')
  \leftarrow\mathcal A^{\mathcal O}_i(pk,s_1,\ldots,s_N,B_i,F_i).
\]
and sends \((F_{i,0}',F_{i,1}')\) back to the challenger.

\item The challenger checks
\[
  b_{i,a}\leftarrow\TFVer(pk,s_i,F_{i,a}'),
  \qquad i\in[N],\ a\in\{0,1\},
\]
and sets $X_i=b_{i,0}b_{i,1}$. The game outputs $1$ if
\[
  \sum_{i=1}^{N}X_i\geq m+1,
\]
and outputs $0$ otherwise.
\end{enumerate}

The scheme satisfies the \emph{global non-proliferation} property if, for every
$R,m,N$ and every adversary $\mathcal A^{\mathcal O}$ as above,
\begin{equation}
\label{eq:G}
  \Pr\bigl[
    \mathsf{SPLIT}^{\mathsf{TF}}_{\mathcal A^{\mathcal O}}
    (1^\lambda,1^R,m,N)=1
  \bigr]\leq\negl(\lambda).
\end{equation}
\end{definition}
All probabilities include oracle generation randomness, algorithmic randomness and measurement
outcomes. 
\begin{remark}
\normalfont
    We note that for $m=0$, this property implies that an adversary without given a torch, cannot clone a flame state, except with negligible probability.
\end{remark}
\begin{definition}[Untelegraphability]
\label{def:untelegraphability}
\normalfont
Let $\mathcal A^{\mathcal O}=(\mathcal A^{\mathcal O}_1,\mathcal A^{\mathcal O}_2)$ be a pair of quantum
adversaries making polynomially many
queries to $\mathcal O$ in total, across both parties and
the entire experiment. 
 The initial joint state of \((\mathcal A^{\mathcal O}_1,\mathcal A^{\mathcal O}_2)\) is separable
across the two parties (i.e., they strictly do not have any shared entanglement) and independent of the challenger's setup
randomness. Shared classical randomness is permitted.
\\
Define the game
\[
  \mathsf{TELEGRAPH}^{\mathsf{TF}}_{\mathcal A^{\mathcal O}}(1^\lambda,1^R)
\]
between a challenger and an adversary as follows.

\begin{enumerate}[label=\arabic*.,leftmargin=*,itemsep=6pt]
\item The challenger samples \(\mathcal O \leftarrow \TFGenOracles\) and 
\(
  \bigl(pk,sk,(t_j,T_j)_{j\in[R]}\bigr)
  \leftarrow\TFSetup(1^\lambda,1^R)
\)
and gives $pk$ to both parties.

\item The adversary specifies disjoint sets
$J_1,J_2\subseteq[R]$. The challenger gives $(T_j)_{j\in J_1}$ to
$\mathcal A^{\mathcal O}_1$ and $(T_j)_{j\in J_2}$ to $\mathcal A^{\mathcal O}_2$.

\item For polynomially many rounds,
$\mathcal A^{\mathcal O}_1$ may submit a spark query. For each adaptive query, the
challenger independently samples
\(
  (s,F)\leftarrow\TFSpark(sk)
\)
and gives $(s,F)$ to $\mathcal A^{\mathcal O}_1$. No flame register is
supplied to $\mathcal A^{\mathcal O}_2$.

\item $\mathcal A^{\mathcal O}_2$ sends a classical message $\mathcal{M}$ to $\mathcal A^{\mathcal O}_1$.

\item Using $\mathcal{M}$ and its private state,
$\mathcal A^{\mathcal O}_1$ sends a classical string $L$ to
$\mathcal A^{\mathcal O}_2$.

\item Using $L$ and its retained state,
$\mathcal A^{\mathcal O}_2$ outputs a classical serial number $s^*$ and a quantum
register $F^*$.

\item The challenger verifies
\[
  b\leftarrow\TFVer(pk,s^*,F^*)
\]
and outputs $b$. The serial number $s^*$ need not have appeared during
the spark query phase.
\end{enumerate}


The scheme satisfies \emph{untelegraphability} if, for every polynomially
bounded $R$ and every such adversary $\mathcal A^{\mathcal O}$,
\begin{equation}
\label{eq:U}
  \Pr\bigl[
    \mathsf{TELEGRAPH}^{\mathsf{TF}}_{\mathcal A^{\mathcal O}}(1^\lambda,1^R)=1
  \bigr]\leq\negl(\lambda).
\end{equation}
All probabilities include oracle generation randomness, algorithmic randomness and measurement
outcomes. 
\end{definition}
\section{Construction of Torch-fire}
\label{sec:torch-fire}
\label{sec:construction}

\subsection{Construction}
\label{subsec:tf-construction}

The construction uses the long-message OSS signing states in
\eqref{eq:long} and the CGS flame state in \eqref{eq:K}.
A torch is one such OSS signing state.
A flame has two parts: an independently generated OSS signing state,
called its \emph{anchor}, and a CGS flame state.

The two parts are cloned separately.
For the CGS flame part, we use the public cloning procedure of CGS fire.
For the anchor, we use a procedure whose oracle queries
require signatures from the torch and anchor, together with a CGS tag.
We first describe formally the component wise  machinery and these oracle interfaces,
then give a construction of torch-fire.

\paragraph{Parameters.} Let $\lambda\ge2$ be the security parameter of the overall torch-fire scheme. 
\\
For the long-message OSS of \Cref{con:long-oss}, we use the parameters
\[
 L=\lambda,\qquad r=16\lambda(\lambda-1),\qquad
 n=r+24\lambda,\qquad k=n.
\]
For the CGS fire of \Cref{con:cgs-source,con:cgs-algorithms}, we use
the parameters below.
We distinguish its dimensions by writing a subscript $K$ throughout.
\[
\mu=\lambda^2,
\qquad
r_K=16\mu(\mu-1),\qquad
 n_K=r_K+24\mu,\qquad
 k_K=n_K
\]

\paragraph{Oracle collections for each instance.}
Fix a polynomially bounded integer $R\ge0$, the number of torches.
We use the per-instance oracle convention of \cite[Section~8]{CGS}:
each instance has its own sampled oracle collection, with public and
private access specified below. We can lift it to the standard oracle model where there is one global oracle and the Setup procedure samples public key, secret key to specify the corresponding oracles instead of outputting oracle tuples.
Oracle sampling and evaluation need
not be efficient; the procedures using these oracles are $\QPT$
with their prescribed access.
The public key is $pk=R$; the associated oracle collection specifies
the instance and is implicit throughout. 

For each instance, the construction uses one common OSS
oracle collection for its torches, a family of OSS oracle collections
for its anchors, and one CGS fire oracle collection.

Fix one such instance and sample the following collections.
\begin{enumerate}[label=(\roman*),leftmargin=*,itemsep=6pt]
\label{TF:OSS-Fire-Oracles}
 \item $\mathcal O_T\leftarrow\OSSGenOracles(1^\lambda)$:
 sample one long-message OSS oracle collection as in \Cref{con:long-oss}.
 All $R$ torches use this same collection.

 \item For every $\ell\in\{0,1\}^{\lambda}$, sample
 \[
  \mathcal O_{A,\ell}\leftarrow\OSSGenOracles(1^\lambda).
 \]
 as in \Cref{con:long-oss}. Each $\ell$ selects an independent long-message OSS oracle collection
 for anchors. 

 \item $\mathcal O_K\leftarrow\FireGenOracles(1^\mu)$:
 sample the CGS public  oracle collection in
 \eqref{eq:native-public-oracles}.
 Its underlying generation oracle is denoted by $\mathcal G$, as
 in \eqref{eq:native-dispatcher}.
 Direct access to $\mathcal G$ is private to the issuing authority
 and is not included in $\mathcal O_K$.
\end{enumerate}
Every long-message OSS oracle collection contains $L$ one-bit components
and its own digest oracle (\Cref{eq:Digest-Oracle-for-long-message-OSS}).

\paragraph{Torch verification keys.}
For the fixed instance, sample independently
\[
 t_j=(t_{j,1},\ldots,t_{j,L})
 \leftarrow(\{0,1\}^r)^L,
 \qquad j\in[R].
\]
These are long-message OSS verification keys relative to $\mathcal O_T$.
Uniform sampling agrees with the key distribution in
\Cref{fact:sz-generation,con:long-oss}.

Publish the keys through the oracle
\begin{equation}\label{eq:tf-torch-list}
 \mathcal T(j)=
 \begin{cases}
  t_j,&j\in[R],\\
  \perp,&\text{otherwise},
 \end{cases}
 \qquad
 \Tpub=\{t_j:j\in[R]\}.
\end{equation}

\paragraph{Private generation data for the issuing authority. }
For each published torch key $t_j$ and component $c\in[L]$,
the issuing authority receives the matrix and shift defining
$W_{t_{j,c}}$ (Recall \Cref{eq:sz-cosets}) in component $c$ of $\mathcal O_T$.
Write these private classical data as
\[
 \mathsf{td}_T
 =\bigl(A_{t_{j,c}},b_{t_{j,c}}\bigr)_{j\in[R],\,c\in[L]}.
\]
Their total size is polynomial in $\lambda$ for polynomially bounded $R$.
Knowing $A_{t_{j,c}}$ and $b_{t_{j,c}}$, by the relation \eqref{eq:sz-coset-state}, the authority can prepare the torch registers
\begin{equation}
\ket{T_j}=\bigotimes_{c=1}^L\ket{W_{t_{j,c}}}
\end{equation}
Separately, the authority retains private access to $\mathcal G$
to prepare the CGS flame part during spark.
Neither $\mathsf{td}_T$ nor direct access to $\mathcal G$ is public.

\paragraph{Anchor labels and flame serials.}
An anchor has a source label $\ell$ and a verification key $a$
relative to $\mathcal O_{A,\ell}$.
The pair is the flame serial:
\begin{equation}\label{eq:tf-serial}
 s=(\ell,a),\qquad
 \ell\in\{0,1\}^{\lambda},\qquad
 a=(a_1,\ldots,a_L)\in(\{0,1\}^r)^L.
\end{equation}

\paragraph{The signing registers.}
For a torch verification key $t\in T_{pub}$ and a flame serial $s=(\ell,a)$,
write the corresponding signing states from \eqref{eq:long} as
\begin{equation}\label{eq:tf-states}
\begin{aligned}
 \ket{T_t}_T&=\bigotimes_{c=1}^L\ket{W_{t_c}}_{T_c},
 &T&=T_1\cdots T_L,\\
 \ket{A_a}_A&=\bigotimes_{c=1}^L\ket{W_{a_c}}_{A_c},
 &A&=A_1\cdots A_L.
\end{aligned}
\end{equation}
In the first line, $W_{t_c}$ belongs to component $c$ of
$\mathcal O_T$.
In the second, $W_{a_c}$ belongs to component $c$ of
$\mathcal O_{A,\ell}$.

\paragraph{The CGS flame part and the flame.}
Let $\ket K$ be the state in \eqref{eq:K} relative to $\mathcal O_K$.
The honest flame with serial $s$ is the state of the register \(F=AK\). Hence, the flame is:
\begin{equation}\label{eq:flame-state}
 \ket{F_s}_F=\ket{A_a}_A\otimes\ket K_K.
\end{equation}
\begin{remark}
\normalfont
    Having described the structure of the flame and the torch, we now move towards building a torch-fire scheme as in \Cref{def:primitive}. To build that, along with the oracles \(\mathcal O_T,\{\mathcal O_{A,\ell}\}_\ell,\mathcal O_{K}\), we need some more fine-grained access structure which we build next.
\end{remark}
\paragraph{Affine functions for a fixed verification key.}
Cloning a flame with
serial $s=(\ell,a)$ needs to prepare an additional register in
$\ket{W_{a_c}}$ for each $c\in[L]$.
This task, ultimately requires preparing $\ket{W_y}$ for a specified
key $y$ in the corresponding one-bit OSS source.
We define the affine functions used in this preparation.

Consider one anchor component and a key $y\in\{0,1\}^r$.
Using its matrix $A_y$, shift $b_y$, and coset $W_y$ from
\Cref{con:sz-source}, define
\begin{equation}\label{eq:tf-rows}
\begin{aligned}
 \operatorname{affine}_{y,0}(w)&=A_yw+b_y,
 &&w\in\{0,1\}^{n-r},\\
 \operatorname{affine}_{y,1}(u)&=
 \begin{cases}
  w,&u=A_yw+b_y\text{ for some }w\in\{0,1\}^{n-r},\\
  \perp,&u\notin W_y,
 \end{cases}
 &&u\in\{0,1\}^k.
\end{aligned}
\end{equation}
Full column rank makes the inverse coordinate $w$ unique.
\begin{remark}\label{remark:row-functions-are-hidden}
\normalfont
    The crucial thing is, these affine functions do not have public access;
access to them is specified below. Critically, if one has access to these affine functions corresponding to a $y\in \{0,1\}^r$, that party can prepare the corresponding coset state as follows. Note, the contrast between this coset preparation procedure against that one in \Cref{fact:sz-generation} lies in which oracles they are allowed to access.
\end{remark}
\paragraph{The coset-state preparation circuit.}
Given the two affine queries for $y$, initialize registers $X,Y$ of
$n-r,k$ qubits to zero and apply
\begin{equation}\label{eq:tf-row-preparation}
\begin{aligned}
 \ket{0^{n-r}}_X\ket{0^k}_Y
 &\xrightarrow{H_X^{\otimes(n-r)}}
  2^{-(n-r)/2}\sum_w\ket w_X\ket{0^k}_Y\\
 &\xrightarrow{\operatorname{affine}_{y,0}:\,X\to Y}
  2^{-(n-r)/2}\sum_w\ket w_X\ket{A_yw+b_y}_Y\\
 &\xrightarrow{\operatorname{affine}_{y,1}:\,Y\to X}
  \ket{0^{n-r}}_X\ket{W_y}_Y.
\end{aligned}
\end{equation}

\paragraph{Messages for anchor cloning} We want to build a structure such that a party holding a torch can clone an anchor but a party without a torch cannot. To do so, we want, 
to clone the anchor part of a flame using the torch, signatures (of some underlying messages) are needed from both the anchor and the torch treating them as long-message OSS signing states. We set the underlying messages to be signed in the following way.
For $s=(\ell,a)$ and a torch verification key $t$, set the messages
\begin{equation}\label{eq:authorization-messages}
 M_A(s,t)=0\Vert\langle s,t\rangle,
 \qquad
 M_T(s,t)=1\Vert\langle s,t\rangle,
\end{equation}
where
\[
 \langle s,t\rangle
 =\ell\Vert a_1\Vert\cdots\Vert a_L
      \Vert t_1\Vert\cdots\Vert t_L.
\]
The anchor signs $M_A(s,t)$, and the torch signs $M_T(s,t)$, using
\Cref{con:long-oss}.
These messages have length $1+\lambda+2Lr$ and remain fixed
throughout one cloning call. Also note, the messages are hashed to $L$ bit digests by the corresponding digest oracles (\Cref{eq:Digest-Oracle-for-long-message-OSS}) in $\mathcal O_{A,\ell}$ and $\mathcal O_T$.

\paragraph{Anchor-affine queries.}
For an anchor, to request an affine value (\Cref{eq:tf-rows}), we need to specify
\[
 q=(s,t,c,e,u),\qquad c\in[L],\qquad e\in\{0,1\},\qquad s=(\ell,a)
\]
The query asks for $\operatorname{affine}_{a_c,e}(u)$ in component $c$
of $\mathcal O_{A,\ell}$.
Its input length is $n-r$ for $e=0$ and $k$ for $e=1$.

Let $\mathcal Q_{\rm affine}$ be the set of these well-formed queries,
with $s$ as in \eqref{eq:tf-serial} and
$t\in(\{0,1\}^r)^L$. The key $t$ specifies the torch under which the accompanying
torch signature is verified.

\paragraph{Encoding the query for a CGS tag.} We also want to enforce, a party who holds an existing torch-fire state should only be able to clone the anchor part. To enforce that, we observe a party holding an existing torch-fire state can generate a valid CGS tag using its CGS flame part as in \Cref{eq:native-tag-circuit}. Thus, in addition to the two OSS signatures, 
the gate keeping oracle $\Gamma$ in \Cref{con:gamma} checks
a CGS tag before answering a complete affine query
$q=(s,t,c,e,u)$.
Recall from \eqref{eq:native-encoding} and
\eqref{eq:native-dispatcher} that a $\texttt{forward}$
query to the CGS fire interface is specified by an
$n_K$-bit string. We encode $q$ into such a string.

Let $[c]_\lambda$ denote the binary representation of $c$
with leading zeros to make its length exactly $\lambda$.
Since $c\in[L]$ and $L=\lambda$, this suffices to represent
every component index.
Define
\begin{equation}\label{eq:tf-request-encoding}
 \mathsf{code}:\mathcal Q_{\rm affine}
       \longrightarrow\{0,1\}^{n_K}
\end{equation}
by concatenating the fields of $q$ in the following order
and appending zeros until the total length is $n_K$:
\[
 \mathsf{code}(q)
 =
 \langle s,t\rangle
 \Vert[c]_\lambda
 \Vert e
 \Vert u
 \Vert 0\cdots0.
\]
Here $\langle s,t\rangle$ is the concatenation already
defined in \eqref{eq:authorization-messages}.
It has length $\lambda+2Lr$, since $|\ell|=\lambda$.
The bit $e$ determines the length of $u$: $n-r$ for $e=0$
and $k$ for $e=1$. Thus all field boundaries are known,
and the encoding is injective and efficiently decodable.

Our choice of CGS fire security parameter $\mu=\lambda^2$
provides enough space for this encoding. Indeed,
\[
 n_K=16\mu^2+8\mu=16\lambda^4+8\lambda^2,
\]
whereas the concatenated fields before padding have length
at most
\[
 2\lambda+2Lr+1+k
 \le32\lambda^3
 \le n_K
 \qquad(\lambda\ge2).
\]

Using the CGS fire encoding from \eqref{eq:native-encoding},
set
\[
 x_q=\mathsf{code}(q),\qquad
 z_q=\enc(\texttt{forward},x_q)\in\mathcal Z.
\]
The encoding $\enc$ and query space $\mathcal Z$ belong
to the CGS fire source $\mathcal O_K$ of this instance.
The circuit $\mathsf{Fire.Tag}^{\mathcal O_K}$ from \eqref{eq:native-tag-circuit}
computes a tag $p$ with query input $z_q$, and $\Gamma$
checks the supplied tag $p$ using $\FireCheck(z_q,p)$.
By \eqref{eq:cgs-forward-check}, this check can also be
performed through the public oracle $O_2$ associated with the CGS fire oracle collection $\mathcal O_K$.

\begin{construction}[The \(\Gamma\) oracle]\label{con:gamma}
\normalfont
The public oracle $\Gamma$ receives a query $q=(s,t,c,e,u)$,
two OSS signatures $\sigma_A,\sigma_T$, and a CGS tag $p$.
Writing $s=(\ell,a)$, it checks
\begin{equation}\label{eq:tf-gate}
\begin{aligned}
 &\quad q\in\mathcal Q_{\rm affine}
   \quad\wedge\quad t\in\Tpub\\
 &\quad\wedge\quad
   \OSSVer^{\mathcal O_{A,\ell}}(a,M_A(s,t),\sigma_A)=1\\
 &\quad\wedge\quad
   \OSSVer^{\mathcal O_T}(t,M_T(s,t),\sigma_T)=1\\
 &\quad\wedge\quad \FireCheck(z_q,p)=1.
\end{aligned}
\end{equation}
Its response is
\begin{equation}\label{eq:Gamma}
 \Gamma(q,p,\sigma_A,\sigma_T)=
 \begin{cases}
  \operatorname{affine}_{a_c,e}(u),
      &\text{if \eqref{eq:tf-gate} holds},\\
  \perp,&\text{otherwise}.
 \end{cases}
\end{equation}
All component oracles belong to the fixed instance.
The returned affine query value belongs to component $c$ of $\mathcal O_{A,\ell}$.
The tag predicate is \eqref{eq:native-tag-check} for $\mathcal O_K$.
Malformed inputs return $\perp$.
\end{construction}

\paragraph{The full public oracle collection.}
For the fixed instance, the full public collection is
\begin{equation}\label{eq:tf-public-oracles}
 \mathcal O=\mathcal O_{\mathrm{pub}}
 =\Bigl(
 \mathcal O_T,
 (\mathcal O_{A,\ell})_{\ell\in\{0,1\}^{\lambda}},
 \mathcal O_K,\mathcal T,\Gamma
 \Bigr).
\end{equation}
All algorithms and adversaries have quantum query access to
$\mathcal O$. Anchor queries specify $\ell$; the additional interfaces
use their inputs defined above.
The issuing authority alone receives $\mathsf{td}_T$ and retains
private access to $\mathcal G$.
\\
\paragraph{Coherent preparation of an affine query.}
Fix a serial $s=(\ell,a)$, a published torch key
$t\in\Tpub$, a component $c\in[L]$, and $e\in\{0,1\}$.
Let, the unitary $\mathcal R^{\mathcal O}_{c,e}$ describes how the anchor $A$, torch $T$, and CGS flame part $K$
are used to request a value of $\operatorname{affine}_{a_c,e}$ from
$\Gamma$ in a coherent way. We defer its implementation details to \Cref{sec:coherent-affine-query}.
\\~\\
Now we are ready to describe the full torch-fire protocol.
\begin{construction}[Torch-fire]\label{con:torch-fire}
\normalfont
 The procedures use the oracle collection in
\eqref{eq:tf-public-oracles}, with the instance implicit throughout.
Setup also receives $\mathsf{td}_T$, and spark uses private access
to $\mathcal G$.
All Fire subroutines below use the CGS fire source $\mathcal O_K$.
Its public key and serial are both taken to be $0^\mu$ in these calls.
\begin{enumerate}[label=\arabic*.,leftmargin=*,itemsep=8pt]
 \item $\mathcal O\leftarrow\TFGenOracles(1^\lambda,1^R)$.

 Sample the oracles \(\Bigl(
 \mathcal O_T,
 (\mathcal O_{A,\ell})_{\ell\in\{0,1\}^{\lambda}},
 \mathcal O_K\Big)\) and torch verification keys $\mathcal T$ according
 to the distributions specified above.
 Define
 $\Gamma$ from the same underlying samples.
 Return the full public collection \eqref{eq:tf-public-oracles}.
 Provide the issuing authority with $\mathsf{td}_T$ and private
 access to $\mathcal G$ for this instance.
 These functions remain fixed throughout the experiment.

 \item $\bigl(pk,sk,(t_j,T_j)_{j\in[R]}\bigr)
       \leftarrow\TFSetup(1^\lambda,1^R)$.

 \begin{enumerate}[label=(\roman*),leftmargin=*,itemsep=4pt]
  \item \emph{Set the public and secret keys.}
  Set $pk=R$. The authority's private access to $\mathcal G$
  serves as its secret key $sk$.

  \item \emph{Read the torch keys using $\mathcal T$.}
  Obtain $t_j=\mathcal T(j)$ for every $j\in[R]$.
  Abort if two keys coincide.

  \item \emph{Prepare the torches using $\mathsf{td}_T=\bigl(A_{t_{j,c}},b_{t_{j,c}}\bigr)_{j\in[R],\,c\in[L]}$.}
  For every $j\in[R]$ and $c\in[L]$, run the coset-state
  preparation circuit \eqref{eq:tf-row-preparation} for
  $y=t_{j,c}$.
  Denote the resulting signing register by $T_{j,c}$ and put
  $T_j=T_{j,1}\cdots T_{j,L}$.
 \end{enumerate}
 Return $pk,sk,(t_j,T_j)_{j\in[R]}$.
 The state of each torch is
 \[
 \rho_{T_j}=\proj{T_{t_j}},
 \qquad
 \ket{T_{t_j}}=\bigotimes_{c=1}^L\ket{W_{t_{j,c}}}.
 \]

 \item $(s,F)\leftarrow\TFSpark(sk)$.

 Prepare a fresh anchor and CGS flame part.
 \begin{enumerate}[label=(\roman*),leftmargin=*,itemsep=4pt]
  \item \emph{Generate the anchor using $\mathcal O_{A,\ell}$.}
  Sample $\ell\leftarrow\{0,1\}^{\lambda}$ and run
  \[
  (a,A)\leftarrow\OSSGenQKey^{\mathcal O_{A,\ell}}(1^\lambda).
  \]
  This invokes \Cref{fact:sz-generation} in each anchor component,
  using its public $\mathcal P,\mathcal P^{-1}$ interfaces.
  It produces verification key $a$ and signing state $\proj{A_a}$.

  \item \emph{Generate the CGS flame part using $\mathcal G$.}
  Prepare $K$ by the circuit for \eqref{eq:K}, using the authority's
  private access to $\mathcal G$ for its forward and inverse queries.
  This prepares $\proj K$ for the same instance.
 \end{enumerate}
 Return
 \[
 s=(\ell,a),\qquad F=AK,\qquad
 \rho_F=\proj{A_a}\otimes\proj K=\proj{F_s}.
 \]
 Each invocation uses fresh randomness and fresh registers.

 \item $b\leftarrow\TFVer(pk,s,F)$.

 Reject malformed classical inputs or incorrect register sizes.
 Parse $s=(\ell,a)$ and $F=AK$, and test the two components.
 \begin{enumerate}[label=(\roman*),leftmargin=*,itemsep=4pt]
  \item \emph{Test the anchor using $\mathcal O_{A,\ell}$.}
  Apply the test of \Cref{lem:token-test} to every component of $A$.
  Each test uses $\mathcal P^{-1}$ for coset membership and
  $\mathcal D$ for dual-subspace membership.
  The joint accepting projector is
  \[
  \Pi_a=\bigotimes_{c=1}^L\Pi_{a_c}=\proj{A_a}.
  \]
  If any component rejects, return zero.

  \item \emph{Test the CGS flame part using $\mathcal O_K$.}
  Return the output of
  \[
  b\leftarrow\FireVer(0^\mu,0^\mu,K).
  \]
  This procedure uses $\mathcal D,O_0,O_1$ to form a tag,
  $N$ to obtain a response, and $V$ to check that response.
 \end{enumerate}
 With $B$ denoting the CGS flame verification effect from \eqref{eq:B},
 the accepting effect on $F$ is
 \begin{equation}\label{eq:flame}
 V_s^F=\Pi_a\otimes B.
 \end{equation}

 \item $(T',F_0,F_1)\leftarrow\TFClone(pk,t,s,T,F)$.

 Reject malformed classical inputs or incorrect register sizes by
 returning zero registers of the prescribed output sizes.
 Otherwise parse $s=(\ell,a)$ and $F=AK$.
 \begin{enumerate}[label=(\roman*),leftmargin=*,itemsep=4pt]
  \item \emph{Check the torch key using $\mathcal T$.}
  Read the published keys and check $t\in\Tpub$.
  If the check fails, return zero registers of the prescribed sizes.

  \item \emph{Clone the CGS flame part using $\mathcal O_K$.}
  Run
  \[
  (K_0,K_1)\leftarrow\FireClone(0^\mu,0^\mu,K).
  \]
  This uses $\mathcal D,O_0,O_1$ to prepare tags and $O_2$
  to implement the CGS generation queries.

  \item \emph{Prepare a second anchor through $\Gamma$.}
  Use $A,T,K_0$ as the resources for the affine-query.
  For each $c\in[L]$, initialize $X_c,Y_c$ to zero, of sizes
  $n-r,k$, respectively, and execute
  \[
  H_{X_c}^{\otimes(n-r)},\qquad
  \mathcal R_{c,0}^{\mathcal O}:X_c\to Y_c,\qquad
  \mathcal R_{c,1}^{\mathcal O}:Y_c\to X_c.
  \]
  This is essentially, the steps of the coset state the preparation procedure of \eqref{eq:tf-row-preparation}.
  Each affine query uses $\mathcal O_{A,\ell}$ and $\mathcal O_T$
  for the two signatures, $\mathcal O_K$ for the tag, and
  $\Gamma$ for the affine value.
 \end{enumerate}
 Return
 \[
 T'=T,\qquad F_0=AK_0,\qquad F_1=Y_1\cdots Y_LK_1,
 \]
 and discard the remaining workspace.
\end{enumerate}
\end{construction}

The four algorithms $\TFSetup,\TFSpark,\TFVer,\TFClone$ are
$\QPT$ with the oracle access and private data specified above,
for polynomially bounded $R$.

\subsection{Correctness}
\label{sec:correctness}

We prove verification and cloning correctness for the construction
in \Cref{subsec:tf-construction}.
An honest flame consists of an anchor $A$ and a CGS flame part $K$.
Verification tests these two registers separately.
Cloning first produces two CGS flame registers, then uses the
original anchor, the torch, and one CGS flame register to prepare
a second anchor through $\Gamma$.
It, though, requires a proof that these affine queries also preserve
the three resources used to make them. We establish this for one
query, then for one cloning call, and finally for repeated cloning.

Throughout, $L=\lambda$, $n-r=24\lambda$, $\mu=\lambda^2$, and
$D(\rho,\sigma)=\tfrac12\norm{\rho-\sigma}_1$.
Expectations are over the sampling of the oracles and classical keys. In principle, we want to argue all the steps relevant to the torch-fire construction in \Cref{subsec:tf-construction} succeeds together, as desired, except with negligible probability.


\begin{lemma}\label{lem:tf-honest-outputs}
\normalfont
Fix the oracles, published torch keys, and an honest spark
serial $s=(\ell,a)$. 
We have,
\begin{equation}\label{eq:tf-setup-abort}
 \Delta_{\rm setup}:=\Pr[\TFSetup\text{ aborts}]
 \le\binom R2 2^{-rL}=\negl(\lambda).
\end{equation}
\end{lemma}
\begin{proof}
\normalfont
Setup aborts precisely when two independently sampled torch keys
coincide, so
\[
 \Delta_{\rm setup}
 \le\sum_{1\le j<j'\le R}\Pr[t_j=t_{j'}]
 =\binom R2 2^{-rL}.
\]
\end{proof}

\paragraph{Signing errors.}
Affine access requires the two OSS signatures and the CGS tag in
\eqref{eq:tf-gate} to pass their checks.
For the OSS checks, let $\epsilon_{y,b}$ be the pointwise error
in \Cref{thm:one-bit-correctness}.
Set $e(\lambda)=2^{-\lambda}+2^{-24\lambda}=\negl(\lambda)$ and
\begin{equation}\label{eq:outer-sign-errors}
\begin{aligned}
 e_A&=\sum_{c=1}^L\sum_{b=0}^1\epsilon_{a_c,b},
 &\E[e_A]&\le2Le(\lambda),\\
 e_T&=\sum_{j=1}^R\sum_{c=1}^L\sum_{b=0}^1
                \epsilon_{t_{j,c},b},
 &\E[e_T]&\le2RLe(\lambda).
\end{aligned}
\end{equation}
If $b^A_c(M)$ and $b^T_c(M)$ denote the $c$th bits of the anchor
and torch digests of $M$, then \Cref{con:long-oss,rem:coherent-signing} gives
\[
\begin{aligned}
 \Pr[\text{anchor signature on }M\text{ is rejected}]
 &\le\sum_{c=1}^L\epsilon_{a_c,b^A_c(M)}\le e_A,\\
 \Pr[\text{torch }j\text{ signature on }M\text{ is rejected}]
 &\le\sum_{c=1}^L\epsilon_{t_{j,c},b^T_c(M)}\le e_T.
\end{aligned}
\]

\paragraph{CGS errors.}
For the CGS component, we need verification, joint cloning,
and tag correctness. With $B$ from \eqref{eq:B} and the tag
projector $P_K$ from \Cref{rem:native-coherent-tag}, define
\begin{equation}\label{eq:tf-native-errors}
\begin{aligned}
 e_{\rm ver}
 &=1-\Tr(B\proj K),\\
 e_{\rm cl}
 &=D\!\left(\FireClone(0^\mu,0^\mu,\proj K),
                         \proj K^{\otimes2}\right),\\
 e_{\rm tag}
 &=\sup_{\|\ket\psi_{ZE}\|=1}
   \norm{(I-P_K)\FireTag
     (\ket K\otimes\ket\psi_{ZE}\otimes\ket0)}^2.
\end{aligned}
\end{equation}
By \Cref{fact:cgs,rem:native-cloning,rem:native-coherent-tag},
\begin{equation}\label{eq:tf-native-mean-errors}
 \E[e_{\rm ver}+e_{\rm cl}+e_{\rm tag}]
 \le\negl(\lambda),
 \qquad \mu=\lambda^2.
\end{equation}

\paragraph{Affine query access.}
The anchor-preparation circuit \eqref{eq:tf-row-preparation}
uses the same anchor, torch, and CGS flame across its affine queries.
We therefore need the coherent implementation of \eqref{eq:tf-row-circuit} for the
requested affine operation while preserving these three resources.
In principle, we want to argue, since $\Gamma$ agrees with the ideal affine oracle whenever its
checks pass, the error probability is negligible.

\begin{lemma}\label{lem:tf-row-access}
\normalfont
Fix an honest serial $s=(\ell,a)$ and $t\in\Tpub$, and set
\begin{equation}\label{eq:authorization-error}
 d=\min\{1,e_A+e_T+e_{\rm tag}\}.
\end{equation}
For any normalized $\ket\phi_{XCE}$, let
\[
 \ket\psi
 =\ket{A_a}\ket{T_t}\ket K
   \otimes\ket\phi_{XCE}\otimes\ket0_W,
\]
where $X,C$ are the affine-input and answer registers, $E,W$ are auxiliary workspace registers.
Then, for every $c\in[L]$ and $e\in\{0,1\}$,
\begin{equation}\label{eq:approx-row-query}
 \norm{\bigl(\mathcal R_{c,e}^{\mathcal O}
             -U_{\operatorname{affine}_{a_c,e}}\bigr)\ket\psi}
 \le2\sqrt d,
\end{equation}
where the ideal query acts only on $X,C$. Moreover,
\begin{equation}\label{eq:mean-authorization-error}
 \E[d]\le2L(R+1)e(\lambda)+\E[e_{\rm tag}]
 \le\negl(\lambda).
\end{equation}
\end{lemma}
\begin{proof}
    The proof is given at \Cref{sec:proof:of-lemma-tf-row-access}.
\end{proof}

\begin{theorem}[Verification and cloning correctness]\label{thm:correct}
\normalfont
\Cref{con:torch-fire} satisfies \eqref{eq:correct-ver} for every
polynomially bounded $R\ge0$, and \eqref{eq:correct-clone} for
every polynomially bounded $R,q$ and every $j\in[R]$.
\end{theorem}
\begin{proof}
\normalfont
Fix a successful setup and the honest states from
\Cref{lem:tf-honest-outputs}.

\paragraph{1. Verification.}
By \eqref{eq:flame} and \Cref{lem:token-test},
\[
 \Tr\!\left[(\Pi_a\otimes B)\proj{F_s}\right]
 =\Tr(\Pi_a\proj{A_a})\Tr(B\proj K)
 =1-e_{\rm ver}.
\]
Hence \eqref{eq:tf-setup-abort} and \eqref{eq:tf-native-mean-errors}
give verification error at most
$\Delta_{\rm setup}+\E[e_{\rm ver}]=\negl(\lambda)$.

\paragraph{2. One cloning call.}
Let $\rho$ be the cloning output on
$\proj{T_t}\otimes\proj{F_s}$.
Compare it with $\widetilde\rho$, obtained by replacing the
CGS cloning output with $\proj K^{\otimes2}$ before preparing
the second anchor.
This replacement changes the state by trace distance $e_{\rm cl}$ i.e. \(D(\rho,\widetilde\rho)\le e_{\rm cl}.\)

It remains to compare the $2L$ anchor-affine queries (2 calls per component) with the ideal
queries in \eqref{eq:tf-row-preparation}.
With workspace retained, let $\ket{\Psi_h}$ be the final state
when the first $h$ queries are ideal and the rest actual,
for $0\le h\le2L$.
Write $\mathcal R_h^{\mathcal O},U_h$ for the two versions of query
$h$, and $\ket{\psi_h}$ for its input after the ideal prefix.
That prefix preserves $A,T,K_0$, so \Cref{lem:tf-row-access} gives
\[
 \norm{\ket{\Psi_{h-1}}-\ket{\Psi_h}}
 =\norm{(\mathcal R_h^{\mathcal O}-U_h)\ket{\psi_h}}
 \le2\sqrt d.
\]
Summing over the queries,
\[
 \norm{\ket{\Psi_0}-\ket{\Psi_{2L}}}
 \le\sum_{h=1}^{2L}
       \norm{\ket{\Psi_{h-1}}-\ket{\Psi_h}}
 \le4L\sqrt d.
\]
The ideal circuit \eqref{eq:tf-row-preparation} has output
\[
 \ket{\Psi_{2L}}
 =\ket{T_t}\otimes\ket{F_s}^{\otimes2}
   \otimes\ket0_{\rm work}.
\]
Passing to trace distance, tracing out workspace, and adding the
CGS cloning error yields
\begin{equation}\label{eq:one-clone-distance}
 D\!\left(\rho_{T'F_0F_1},
          \proj{T_t}\otimes\proj{F_s}^{\otimes2}\right)
 \le e_{\rm cl}+4L\sqrt d=\negl(\lambda).
\end{equation}

\paragraph{3. Repeated cloning.}
Fix $j\in[R]$ and $t=t_j$.
Let $\rho_i$ be the joint state of the torch and all $i+1$ flames
after $i$ calls in \Cref{def:primitive}, and set
\[
 \tau_i=\proj{T_t}\otimes\proj{F_s}^{\otimes(i+1)}.
\]
The $i$th cloning operation $\mathcal C_i$ acts on $T,F_{i-1}$
and leaves $F_0,\ldots,F_{i-2}$ untouched.
Contractivity and \eqref{eq:one-clone-distance} imply
\[
\begin{aligned}
 D(\rho_i,\tau_i)
 &\le D\!\left(\mathcal C_i(\rho_{i-1}),
                \mathcal C_i(\tau_{i-1})\right)
     +D\!\left(\mathcal C_i(\tau_{i-1}),\tau_i\right)\\
 &\le D(\rho_{i-1},\tau_{i-1})+e_{\rm cl}+4L\sqrt d.
\end{aligned}
\]
Since $\rho_0=\tau_0$, induction gives
\begin{equation}\label{eq:tf-repeated-cloning}
 D(\rho_i,\tau_i)\le i\bigl(e_{\rm cl}+4L\sqrt d\bigr).
\end{equation}

\paragraph{4. Acceptance of all flames.}
For the failure event
$\mathcal E=\{\exists i\in\{0,\ldots,q\}:b_i=0\}$,
\eqref{eq:tf-repeated-cloning} and a union bound give
\[
\begin{aligned}
 \Pr_{\rho_q}[\mathcal E]
 &\le\Pr_{\tau_q}[\mathcal E]+D(\rho_q,\tau_q)\\
 &\le(q+1)e_{\rm ver}+q\bigl(e_{\rm cl}+4L\sqrt d\bigr).
\end{aligned}
\]
Averaging and using $\E[\sqrt d]\le\sqrt{\E[d]}$, we obtain
\begin{equation}\label{eq:correctness-total-error}
\begin{aligned}
 \Pr[\text{cloning correctness fails}]
 &\le\Delta_{\rm setup}
    +q\bigl(\E[e_{\rm cl}]+4L\sqrt{\E[d]}\bigr)
    +(q+1)\E[e_{\rm ver}]\\
 &\le\negl(\lambda).
\end{aligned}
\end{equation}
The last inequality uses \eqref{eq:tf-setup-abort},
\eqref{eq:tf-native-mean-errors}, and
\eqref{eq:mean-authorization-error}, with $R,q,L$ polynomially bounded.
\end{proof}
\subsection{Untelegraphability}
\label{sec:untelegraphability}
We now prove the \Cref{con:torch-fire} satisfies the untelegraphability property as in \Cref{def:untelegraphability}.
By \eqref{eq:flame}, the accepting effect satisfies
\[
 \Pi_a\otimes B\preceq I_A\otimes B.
\]
We therefore reduce to CGS untelegraphability by discarding the
output anchor. 
\begin{theorem}\label{thm:tf-untelegraphability}
\normalfont
for every polynomially bounded $R$ and every polynomial query bounded pair of adversaries 
$\mathcal A^{\mathcal O}=(\mathcal A_1^{\mathcal O},
\mathcal A_2^{\mathcal O})$,
\begin{equation}\label{eq:tf-untelegraphability}
 \Pr\!\left[
  \mathsf{TELEGRAPH}^{\mathsf{TF}}_{\mathcal A^{\mathcal O}}
       (1^\lambda,1^R)=1
 \right]\le\negl(\lambda).
\end{equation}
\end{theorem}
\begin{proof}
\normalfont
Suppose $\mathcal A$ is an adversary to the untelegraphability game of the torch-fire object (\Cref{def:untelegraphability}) making at most $Q$ public-oracle queries
and $v$ spark requests. Using that, 
we construct an adversary 
$(\mathcal B_1^{\mathcal O_K},\mathcal B_2^{\mathcal O_K})$
for the untelegraphability game of the CGS fire scheme as per \Cref{def:native-fire-untelegraphability} at $\mu=\lambda^2$,
with $\mathcal B_i$ running $\mathcal A_i$.

\paragraph{1. Simulating the oracles.}
Use the challenged CGS collection as $\mathcal O_K$ in
\eqref{eq:tf-public-oracles}; write $\widehat s$ for its public
key and serial. Sample the remaining torch and anchor data,
including digest oracles and published keys, as in
\Cref{subsec:tf-construction}, using shared classical randomness
independent of $\mathcal O_K$.
Both simulated adversaries $(\mathcal B_1^{\mathcal O_K},\mathcal B_2^{\mathcal O_K})$ can then evaluate
$\mathcal O_T$, $(\mathcal O_{A,\ell})_\ell$, and $\mathcal T$
locally, forward $\mathcal O_K$ queries to the challenger,
and simulate setup with $pk=R$.

\paragraph{2. Simulating $\Gamma$.}
For a well-formed query $q=(s,t,c,e,u)$, let
$\chi(q,\sigma_A,\sigma_T)$ be the predicate which is 1 if the
torch membership check and two OSS checks in \eqref{eq:tf-gate} pass. 
Both $\chi$ and $\operatorname{affine}_{a_c,e}(u)$ are determined
by the sampled oracles and data. For the tag check, put
\[
 z_q=\enc(\texttt{forward},\mathsf{code}(q)).
\]
Since $\mathcal G(z_q)=\mathcal P(\mathsf{code}(q))\ne\perp$,
\eqref{eq:cgs-forward-check} gives
\begin{equation}\label{eq:tf-telegraph-tag-check}
 \FireCheck(z_q,p)=\ind[O_2(z_q,p)\ne\perp].
\end{equation}
Consequently, \eqref{eq:Gamma} becomes
\[
 \Gamma(q,p,\sigma_A,\sigma_T)=
 \begin{cases}
  \operatorname{affine}_{a_c,e}(u),
    &\chi(q,\sigma_A,\sigma_T)=1
       \ \wedge\ O_2(z_q,p)\ne\perp,\\
  \perp,&\text{otherwise}.
 \end{cases}
\]

\paragraph{3. Simulating issuance and output.}
\begin{enumerate}[label=(\roman*),leftmargin=*,itemsep=5pt]
 \item \emph{Torches.}
 For the chosen disjoint sets $J_1,J_2$, each $\mathcal B_i$
 prepares its assigned torches from the sampled matrices and shifts.
 By \Cref{lem:tf-honest-outputs}, this gives the required joint state
 \[
  \rho_{T_{J_1}T_{J_2}}
  =\bigotimes_{j\in J_1\cup J_2}\proj{T_{t_j}}.
 \]

 \item \emph{Spark.}
 For each request, $\mathcal B_1$ obtains a fresh $K$ from the
 CGS challenger, samples $\ell\leftarrow\{0,1\}^{\lambda}$,
 and runs $\OSSGenQKey^{\mathcal O_{A,\ell}}(1^\lambda)$
 to obtain $(a,A)$. It returns
 \[
  (s,F)=((\ell,a),AK),
  \qquad \rho_{F\mid s}=\proj{A_a}\otimes\proj K,
 \]
 exactly as in \Cref{con:torch-fire}.

 \item \emph{Output.}
 Given $s^*=(\ell,a)$ and $F^*=AK$, party $\mathcal B_2$ returns
 \[
  (s^*,\rho_{AK})\longmapsto
       (\widehat s,\Tr_A\rho_{AK}).
 \]
\end{enumerate}
The classical messages exchanged between $(\mathcal A_1^{\mathcal O},
\mathcal A_2^{\mathcal O})$ are kept unchanged.

\paragraph{4. Comparing success probabilities.}
Let $p_{\rm TF},p_{\rm CGS}$ be the respective success probabilities,
and let $\rho_{AK}^{s}$ be the (possibly subnormalized) state on valid
outputs with serial $s$. The CGS effect $B$ is the same for every $s$.
Thus, \Cref{fact:cgs} gives
\begin{equation}\label{eq:tf-native-domination}
\begin{aligned}
 p_{\rm TF}
 &=\sum_{s=(\ell,a)}
   \Pr\!\left[s^*=s\ \wedge\ \TFVer\text{ accepts}\right]\\
 &=\E\!\left[
   \sum_{s=(\ell,a)}
    \Tr\!\left[(\Pi_a\otimes B)\rho_{AK}^{s}\right]
   \right]\\
 &\le\E\!\left[
   \sum_{s=(\ell,a)}
    \Tr\!\left[(I_A\otimes B)\rho_{AK}^{s}\right]
   \right]\\
 &=\E\!\left[
   \sum_{s=(\ell,a)}
    \Tr\!\left[B\,\Tr_A(\rho_{AK}^{s})\right]
   \right]\\
 &\le p_{\rm CGS}
 \le\negl(\mu)=\negl(\lambda).
\end{aligned}
\end{equation}
The first inequality uses
$\Pi_a\preceq I_A$ and $B\succeq0$.
\end{proof}
\section{Extraction and Global Non-proliferation}
\label{sec:extraction}
\label{sec:global}

We prove the global non-proliferation property (\Cref{def:global}) for \Cref{con:torch-fire} in the
per-instance oracle model of \Cref{subsec:tf-construction}.
All public oracles are as shown in 
\eqref{eq:tf-public-oracles}.

The idea is to convert successful cloning at several separated locations $i\in [N]$
into extracting valid torch signatures for their respective challenge labels $\ell_i$.

\paragraph{The experiment.}
Fix polynomially bounded $R,m,N$, with $0\le m\le R$, and an
adversary
\[
 \mathcal A^{\mathcal O}
 =(\mathcal A_0^{\mathcal O},\mathcal A_1^{\mathcal O},\ldots,
   \mathcal A_N^{\mathcal O})
\]
as in \Cref{def:global}.
Let $v$ is an upper bound on the number of pre-split spark queries, and let
$Q_{\rm pre}$ is upper bound on the number of pre-split oracle queries.
The challenge serial at location $i$ is:
\[
 s_i=(\ell_i,a_i),\qquad
 \ell_i\in\{0,1\}^{\lambda},\qquad
 a_i=(a_{i,1},\ldots,a_{i,L})\in (\{0,1\}^r)^L.
\]
Let $X_i$ indicate that both output flames (outcome of the cloning experiment) at location $i$ verify
under the serial $s_i$. The winning event is $\sum_iX_i\ge m+1$.
It is trivially impossible when $m\ge N$, so throughout the proof we may assume
$m<N$ and hence $N\ge1$.

\paragraph{Distinct labels $\ell\in \{0,1\}^\lambda$.}
Define the \emph{ideal experiment} by additionally declaring
failure if any two labels $\ell$ among the $v$ pre-split and $N$ challenge spark
queries coincide. If $p_{\rm real}$ and $p_{\rm ideal}$ are the real and ideal
winning probabilities for an adversary, respectively, then a union bound gives
\begin{equation}\label{eq:global-reference-comparison}
 p_{\rm real}\le p_{\rm ideal}+\binom{v+N}{2}2^{-\lambda}=p_{\rm ideal}+\negl(\lambda)
\end{equation}
for polynomial $v(\lambda), N(\lambda)$.

\subsection{Extraction from separated computations}
\label{subsec:separated-extraction}

We prove a general extraction lemma for $N$ local oracle
computations on a possibly entangled input.
Its application in the context of torch-fire game is given in
\Cref{subsec:extraction-oss}.
Previous work considers simultaneous extraction from two entangled
oracle algorithms~\cite{CMP24}, or from a common quantum
input with essentially fixed outputs~\cite{CLPY25}.
Our lemma compares, for every $d$, the probability that at least
$d$ non-communicating parties accept with the probability that
at least $d$ local extractors output witnesses.
The extractors use only the modified oracles. 

\paragraph{The input spaces.}
For each location $i\in[N]$, let $\mathcal H_i$ be its input
space and $\mathcal W_i$ its workspace.
Let $\rho$ be a possibly subnormalized joint density operator on
$\mathcal H=\bigotimes_{i=1}^N\mathcal H_i$, and let $\rho_i$
be its reduced state at location $i$:
\[
 \rho\succeq0,\qquad
 \Tr\rho\le1,\qquad
 \rho_i=\Tr_{\bigotimes_{k\ne i}\mathcal H_k}(\rho).
\]
Each computation operates only on its own input and workspace;
the locations do not communicate.

\paragraph{The two oracle computations.}
At location $i$, consider the q-query, original and modified mode of computation given by:
\begin{equation}\label{eq:original-modified-computation}
\begin{aligned}
 U_i^1
 &=G_{i,q}O_{i,q}^1G_{i,q-1}\cdots
   G_{i,1}O_{i,1}^1G_{i,0},\\
 U_i^0
 &=G_{i,q}O_{i,q}^0G_{i,q-1}\cdots
   G_{i,1}O_{i,1}^0G_{i,0}.
\end{aligned}
\end{equation}
Both act on the joint space $\mathcal H_i\otimes\mathcal W_i$.
Here $G_{i,0},\ldots,G_{i,q}$ are are common  unitary operators, whereas
 $O_{i,j}^1$ and $O_{i,j}^0$
are the original and modified oracle query unitaries. The unitaries $G_{i,j}$ may also include identical oracle
queries.
\paragraph{Acceptance.} At location $i\in [N]$,
Let $J_i$ be the isometry that appends zero workspace to an input space.
\[
 J_i:\mathcal H_i\longrightarrow
       \mathcal H_i\otimes\mathcal W_i,
 \qquad
 J_i\ket\psi=\ket\psi\ket0_{\mathcal W_i}.
\]
Let $\Pi_i$ be an
efficiently measurable final acceptance projector.
The accepting maps for the original and modified computations are
\begin{equation}\label{eq:extraction-local-maps}
 A_i=\Pi_iU_i^1J_i,
 \qquad
 C_i=\Pi_iU_i^0J_i.
\end{equation}
Their acceptance probabilities on the local input $\rho_i$ are
\[
 \Pr[\text{original accepts}]
   =\Tr(A_i\rho_iA_i^\dagger),
 \qquad
 \Pr[\text{modified accepts}]
   =\Tr(C_i\rho_iC_i^\dagger).
\]

\paragraph{Witness at a position}
Let $Q_i$ be a computational-basis diagonal projector on the
input register (i.e. on the Hilbert space $\mathcal H_i$) such that
\begin{equation}\label{eq:witness-proj}
 (O_{i,j}^1-O_{i,j}^0)(I-Q_i)=0
 \qquad\text{for every }j\in[q].
\end{equation}
Thus the two oracles agree on every query outside the support of $Q_i$.
A classical query input $x$ is a \emph{witness} at location $i\in [N]$ if it is in the support of the projector $Q_i$ i.e, 
\begin{equation}\label{eq:witness}
 Q_i\ket x=\ket x.
\end{equation}

\begin{remark}
\label{rem:extraction-oracle-access}
\normalfont
Let $U_{i,<j}^0$ be the unitary part of the modified computation
before its $j$th query:
\begin{equation}\label{eq:U-prefix-ops}
\begin{aligned}
 U_{i,<1}^0&=G_{i,0},\\
 U_{i,<j}^0
 &=(G_{i,j-1}O_{i,j-1}^0)\cdots
   (G_{i,1}O_{i,1}^0)G_{i,0},
   \qquad 2\le j\le q.
\end{aligned}
\end{equation}
Appending the zero workspace providing isometry $J_i$, we also define
\begin{equation}\label{eq:V-prefix-ops}
    V_{i,j}=U_{i,<j}^0J_i
\end{equation}
We assume that all of $U_{i,<j}^0$, $(U_{i,<j}^0)^\dagger, Q_i$
can be implemented efficiently, using only the oracles 
available to $U_i^0$.
\end{remark}

\paragraph{Sampling a witness.}
A witness sampler chooses an index $j\in [q]$ uniformly,
runs the modified computation up to that position, and measures
its query input.
It succeeds if the measured input is a witness, and 
fails otherwise.

The effect describing the success probability of this sampler is
\begin{equation}\label{eq:extraction-sampler}
 E_i=\frac1q\sum_{j=1}^qV_{i,j}^\dagger Q_iV_{i,j},
 \qquad
 0\preceq E_i\preceq I.
\end{equation}
In particular,
\[
 \Pr[\text{the sampler outputs a witness on }\rho_i]
 =\Tr(E_i\rho_i).
\]

We note that, measuring a uniformly chosen oracle query is a standard
extraction technique as done in ~\cite{CMP24}.
Here the sampler runs the modified computation.

\paragraph{Filtering the input.}
On input components $i\in [N]$, where the above witness sampling procedure succeeds with low probability, 
we will bound the probability of acceptance in the original experiment using that of the modified
experiment. To do that, first fix an integer
$\beta\ge1$ and define a filtering operator
\begin{equation}\label{eq:extraction-filter-definition}
 L_i=(I-E_i)^\beta.
\end{equation}
For every unit eigenvector $\ket e$ of $E_i$,
\[
 E_i\ket e=u\ket e
 \quad\Longrightarrow\quad
 L_i\ket e=(1-u)^\beta\ket e,
 \qquad u\in[0,1].
\]
Here $u$ is the witness-sampling probability on $\ket e$.
The filter therefore suppresses components with large
witness-sampling probability.

The implementation follows the modified Marriott--Watrous
procedure of \cite{HLM17}.
As shown in the proof, each successful cycle applies $I-E_i$,
so $\beta$ successful cycles implement the branch
\[
 \rho_i\longmapsto L_i\rho_iL_i.
\]

\paragraph{Acceptance probability after filtering for the modified experiment.}
Let $\varepsilon_\beta$ be the sum, over all locations $i\in [N]$, of the
probabilities that the filter $L_i$ succeeds and the modified
computation then accepts:
\begin{equation}\label{eq:extraction-filtered-null}
 \varepsilon_\beta
 =\sum_{i=1}^N
   \E\!\left[
     \Tr\!\left(
       C_iL_i\rho_iL_iC_i^\dagger
     \right)
   \right].
\end{equation}
The expectation is over the distribution of the
classical oracles.
\\
Now, we are ready to formally state the extraction lemma.
\begin{lemma}
\label{lem:compact-extraction}
\normalfont
Fix an integer $\beta\ge1$. 
Consider the original and modified $q$-query oracle aided computation as described in \Cref{eq:original-modified-computation}. Let $X_i$ indicate acceptance by
the original computation at location $i\in [N]$.
For every $\delta\in(0,1)$, for every location $i\in [N]$ there are local extracting procedures that output either a 
a witness (\ref{eq:witness-proj},\ref{eq:witness}) or a failure. If $Y_i$ indicates a successful witness output, then,
for every $d\in[N]$,
\begin{equation}\label{eq:compact-extraction}
\begin{aligned}
 \Pr\!\left[\sum_{i=1}^N X_i\ge d\right]
 &\le \Pr\!\left[\sum_{i=1}^N Y_i\ge d\right]
       +N\delta(2+\delta)\\
 &\quad +(1+\delta^{-1})
       \left(2\varepsilon_\beta
             +\frac{8Nq^2}{2\beta+1}\right).
\end{aligned}
\end{equation}
The two probabilities refer to separate executions on the same
input state $\rho$.
Each witness extracting procedure receives only its local input and fresh
workspace. It uses the the unitaries $U^0_{i,<j}, {U^0_{i,<j}}^\dagger$
and witness testing projectors $Q_i$, as in \Cref{rem:extraction-oracle-access}, with no communication between separate locations.
Moreover, the witness extraction procedure is $\QPT$ under the assumption of \Cref{rem:extraction-oracle-access} and 
polynomially bounded  $\beta$, and $1/\delta$. 
\end{lemma}

\begin{proof}
\normalfont
We first work at one location and omit its index.

\paragraph{1. Comparing the two computations.}
We use the standard oracle-hybrid argument of
\cite{BBBV97}, as in
\Cref{lem:bbbv-hybrid}, written as an operator inequality.
For each query position $j$, let $B_j$ be the suffix unitary of
the original computation after the $j$-th query:
\[
 B_q=G_q,\qquad
 B_j=G_qO_q^1G_{q-1}\cdots O_{j+1}^1G_j
 \quad (1\le j<q).
\]
Replacing the queries (original and modified oracle queries) one at a time gives the telescoping identity
\[
 A-C=\sum_{j=1}^q\Pi B_j(O_j^1-O_j^0)V_j.
\]
Recall $V_j$ is the prefix relative to the modified oracles defined as in $\Cref{eq:V-prefix-ops}$.
By the witness condition (\Cref{eq:witness-proj}) and unitarity, we also have,
\[
 O_j^1-O_j^0=(O_j^1-O_j^0)Q,
 \qquad
 \norm{O_j^1-O_j^0}_{\rm op}\le2,
 \qquad
 \norm{\Pi B_j}_{\rm op}\le1.
\]
Consequently, for every input vector $\ket\psi$, by triangle-inequality,
\begin{align*}
 \norm{(A-C)\ket\psi}^2
 &\le4\left(\sum_{j=1}^q\norm{QV_j\ket\psi}\right)^2\\
 &\le4q\sum_{j=1}^q\norm{QV_j\ket\psi}^2
  =4q^2\bra\psi E\ket\psi.
\end{align*}
The second inequality is Cauchy--Schwarz and the final equality is due the definition of the operator $E$ as in \Cref{eq:extraction-sampler}.
Since this is true for all unit vector $\ket{\psi}$,
\begin{equation}\label{eq:extraction-hybrid}
 (A-C)^\dagger(A-C)\preceq4q^2E.
\end{equation}

\paragraph{2. Implementing the filter.}
We construct a procedure whose successful branch applies the operator $I-E$ on the input register.
Repeating it $\beta$ times,  we will implement the filter operator $L=(I-E)^\beta$. The filter follows the modified Marriott--Watrous procedure of
Harrow, Lin, and Montanaro~\cite{HLM17}.
For this purpose, we first implement the witness sampling experiment coherently.

Let $R$ be a register of $\log_2q$ qubits, with computational-basis
states labelled by $j\in[q]$.
Let $J_0$ append both $R$ and the circuit workspace $\mathcal W$
in their zero states:
\[
 J_0\ket\psi
 =\ket0_RJ\ket\psi
 =\ket0_R\ket\psi\ket0_{\mathcal W},
 \qquad J_0^\dagger J_0=I.
\]
We recall from \Cref{rem:extraction-oracle-access} $V_j=U_{<j}^0J$, where $U_{<j}^0$ is the unitary part of
the modified computation before query $j$.

Apply Hadamards on $R$ and then apply $U_{<j}^0$ controlled on
its value $j$, without measuring $R$.
Let, the resulting joint unitary is $V$. It satisfies
\[
 VJ_0\ket\psi
 =\frac1{\sqrt q}\sum_{j=1}^q
      \ket j_RU_{<j}^0J\ket\psi
 =\frac1{\sqrt q}\sum_{j=1}^q
      \ket j_RV_j\ket\psi.
\]
Under the assumption of \Cref{rem:extraction-oracle-access}, we note that both $V$ and $V^\dagger$ have efficient implementation.
\\~\\
Now, On the enlarged space, define the projector
\begin{equation}\label{eq:proj-pi0-piQ}
 \Pi_0=J_0J_0^\dagger
      =\proj{0}_R\otimes I\otimes\proj{0}_{\mathcal W},
 \qquad
 \Pi_Q=I_R\otimes Q.
\end{equation}
The projector $\Pi_0$ checks whether both appended registers
are zero, leaving the input register unrestricted.
The projector $\Pi_Q$ tests the witness condition on the query
input and acts identically on $R$.
\\
Combining the operators, we have the following important identity:
\begin{equation}\label{eq:extraction-compression}
\begin{aligned}
 J_0^\dagger V^\dagger\Pi_QVJ_0
 &=\frac1q\sum_{k,j=1}^q
      \langle k|j\rangle\,V_k^\dagger QV_j
 =\frac1q\sum_{j=1}^qV_j^\dagger QV_j=E
\end{aligned}
\end{equation}
Also,
\begin{equation}
    \Pi_0V^\dagger(I-\Pi_Q)VJ_0
 =J_0(I-E)
\end{equation}
The last identity follows from $\Pi_0=J_0J_0^\dagger$,
$V^\dagger V=I$, and $J_0^\dagger J_0=I$.
\\
We are now ready to describe the filter operator. One filter cycle consists of the following operations.
\begin{enumerate}[label=(\roman*),leftmargin=*,itemsep=3pt]
\item Apply $J_0$.
 \item Apply $V$ and measure $\{\Pi_Q,I-\Pi_Q\}$.
       Continue only on outcome $I-\Pi_Q$.
 \item Apply $V^\dagger$ and measure $\{\Pi_0,I-\Pi_0\}$.
       Continue only on outcome $\Pi_0$.
\end{enumerate}
Every other outcome is failure.

By \eqref{eq:extraction-compression}, the (unnormalized) vector
after one successful cycle is
\[
 \Pi_0V^\dagger(I-\Pi_Q)VJ_0\ket\psi
 =J_0(I-E)\ket\psi.
\]
Since the appended registers are again zero, the same cycle
can be repeated on the remaining input.
After $\beta$ successful cycles, the (unnormalized) vector is
\[
 J_0(I-E)^\beta\ket\psi
 =J_0L\ket\psi
 =\ket0_R\bigl(L\ket\psi\bigr)\ket0_{\mathcal W}.
\]

In particular, for a unit eigenvector $E\ket e=u\ket e$,
\[
 \Pr[\text{all $\beta$ filter cycles succeed on }\ket e]
 =\norm{J_0L\ket e}^2
 =\norm{L\ket e}^2
 =(1-u)^{2\beta}.
\]

\paragraph{3. Bounding acceptance after filtering.}
We will compare the  acceptance probabilities of the original and modified computation on the filtered input.
Their successful maps are $AL$ and $CL$, respectively.
We first bound their difference $(A-C)L$.

Since $L=(I-E)^\beta$ is Hermitian and commutes with $E$,
\eqref{eq:extraction-hybrid} gives
\[
 L(A-C)^\dagger(A-C)L
 \preceq4q^2LEL
 =4q^2E(I-E)^{2\beta}.
\]
The eigenvalues of $E(I-E)^{2\beta}$ are
$u(1-u)^{2\beta}$, where $u\in[0,1]$ ranges over the eigenvalues
of $E$.
This scalar function is maximized at $u=1/(2\beta+1)$, so
\[
 u(1-u)^{2\beta}
 \le\frac1{2\beta+1}
      \left(\frac{2\beta}{2\beta+1}\right)^{2\beta}
 \le\frac1{2\beta+1}.
\]
Consequently,
\[
 L(A-C)^\dagger(A-C)L
 \preceq\frac{4q^2}{2\beta+1}I.
\]

Now decompose
\[
 AL=CL+(A-C)L.
\]
Using $\norm{x+y}^2\le2\norm{x}^2+2\norm{y}^2$, for every
input vector $\ket\psi$ we obtain
\begin{align*}
 \norm{AL\ket\psi}^2
 &\le2\norm{CL\ket\psi}^2
     +2\norm{(A-C)L\ket\psi}^2\\
 &\le2\norm{CL\ket\psi}^2
     +\frac{8q^2}{2\beta+1}\norm{\ket\psi}^2.
\end{align*}
Since this holds for every $\ket\psi$, it is equivalent to
\begin{equation}\label{eq:extraction-filtered-acceptance}
 LA^\dagger AL
 \preceq2LC^\dagger CL+\frac{8q^2}{2\beta+1}I.
\end{equation}
Restoring the location indices, the same bound applies to each
local mixed input $\rho_i$:
\[
\begin{aligned}
 \Tr(A_iL_i\rho_iL_iA_i^\dagger)
 &\le2\Tr(C_iL_i\rho_iL_iC_i^\dagger)
       +\frac{8q^2}{2\beta+1}\Tr\rho_i.
\end{aligned}
\]
The two trace terms are the probabilities of successful filtering
followed by acceptance for the original and modified computation, respectively.

Finally, $\Tr\rho_i\le1$.
Taking expectation over the oracle distributions, summing over all locations,
and using the definition of $\varepsilon_\beta$ in
\eqref{eq:extraction-filtered-null}, we obtain
\begin{equation}\label{eq:extraction-total-filtered}
 \sum_{i=1}^N
 \E\!\left[\Tr(A_iL_i\rho_iL_iA_i^\dagger)\right]
 \le2\varepsilon_\beta+\frac{8Nq^2}{2\beta+1}.
\end{equation}

\paragraph{4. The witness extraction procedure.}
We now construct a witness extractor on the original local input.
Its purpose is to find a witness with high probability on the
components strongly suppressed by the filter $L$.
\\
Consider the projectors $\Pi_0,\Pi_Q$ from \Cref{eq:proj-pi0-piQ}. On the enlarged space, define the projector and the parameter
\[
 S=V\Pi_0V^\dagger=(VJ_0)(VJ_0)^\dagger,
 \qquad
 T=\Pi_Q,
 \qquad
 t=\left\lceil\frac{2\beta}{\delta^2}\right\rceil.
\]
Since $VJ_0$ is an isometry $\mathcal H \rightarrow \mathcal H \otimes \mathcal W$, $S$ is the orthogonal projector onto
the image of the map $VJ_0$.
The projector $T$ tests whether the current query input is a witness.

After rejection by $T$, the state is supported in the range of
$I-T$. An immediate repetition of the witness test would therefore
reject again, since $T(I-T)=0$.
We interleave witness tests with measurements of $S$, which can kickback some nonzero weight in the witness subspace.

Use the following alternating-projector procedure, based on
Marriott--Watrous amplification~\cite{MW05}.
\begin{enumerate}[label=(\roman*),leftmargin=*,itemsep=3pt]
 \item Append zero registers using $J_0$, apply $V$, and measure
       $\{T,I-T\}$.

 \item Whenever the witness test accepts, measure the complete
       query input in the computational basis, output it, and stop.

 \item After the first rejection, perform at most $t$ further
       rounds. In each round, measure $\{S,I-S\}$, discard that
       outcome, and then measure $\{T,I-T\}$.
       On acceptance, output the query as above.
       If every witness test rejects, output failure.
\end{enumerate}
Both outcomes of the $S$-measurement lead to the next witness
test; neither outcome causes the procedure to abort.



\paragraph{5. Computing the extraction probability.}
By Jordan's lemma, the sampler space decomposes into mutually
orthogonal subspaces of dimension at most two that are invariant
under both $S$ and $T$~\cite{NWZ09,CMSZ21}.
We include the standard two-projector calculation for
completeness~\cite{NWZ09,CMSZ21}. It gives the exact extraction
probability, which we then use to compare the success effect
$W$ with $(I-L)^2$.

Fix a unit eigenvector $E\ket e=u\ket e$ with $0<u<1$, and let
\[
 \ket v=VJ_0\ket e
\]
Since $VJ_0$ is an isometry and
$J_0^\dagger V^\dagger TVJ_0=E$, by \Cref{eq:extraction-compression}, we have
\[
 \norm{\ket v}=1,\qquad
 \norm{T\ket v}^2=u,\qquad
 \norm{(I-T)\ket v}^2=1-u.
\]
Thus the normalized states after witness acceptance and rejection
are, respectively,
\[
 \ket w=\frac{T\ket v}{\sqrt u},
 \qquad
 \ket z=\frac{(I-T)\ket v}{\sqrt{1-u}}.
\]
These vectors are orthogonal and
\[
 \ket v=\sqrt u\ket w+\sqrt{1-u}\ket z.
\]

To determine how the $S$-measurement changes these
states, use $S=(VJ_0)(VJ_0)^\dagger$ to obtain
\[
 S\ket v=\ket v,
 \qquad
 ST\ket v=VJ_0E\ket e=u\ket v.
\]
It follows that
\[
\begin{aligned}
 S\ket w
 &=\sqrt u\ket v
  =u\ket w+\sqrt{u(1-u)}\ket z,\\
 S\ket z
 &=\sqrt{1-u}\ket v
  =\sqrt{u(1-u)}\ket w+(1-u)\ket z.
\end{aligned}
\]
Hence the two-dimensional space
$\mathcal K_e=\operatorname{span}\{\ket w,\ket z\}$
is invariant under $S$, $T$.
In the ordered basis $(\ket w,\ket z)$, their restrictions are
\begin{equation}\label{eq:extraction-two-projectors}
 T\big|_{\mathcal K_e}
 =\begin{pmatrix}1&0\\0&0\end{pmatrix},
 \qquad
 S\big|_{\mathcal K_e}
 =\begin{pmatrix}
 u&\sqrt{u(1-u)}\\
 \sqrt{u(1-u)}&1-u
 \end{pmatrix}.
\end{equation}

The first witness test succeeds with probability $u$.
On rejection, the normalized state is $\ket z$.
In the next round, the two possible $S$-outcomes followed by
witness acceptance give the (unnormalized) vectors
\[
 TS\ket z=\sqrt{u(1-u)}\ket w,
 \qquad
 T(I-S)\ket z=-\sqrt{u(1-u)}\ket w.
\]
Since these branches correspond to distinct measured outcomes of $\{S,I-S\}$,
their probabilities add. Thus one further round succeeds with
probability
\[
 \norm{TS\ket z}^2+\norm{T(I-S)\ket z}^2
 =2u(1-u).
\]
If the witness test rejects again, the normalized state returns
to $\ket z$.


Consequently, each further round, conditioned on all earlier
rejections, fails with probability $1-2u(1-u)$.
Including the initial rejection and the $t$ further rounds,
the total witness extraction probability is
\begin{equation}\label{eq:extraction-witness-function}
 g_t(u)=1-(1-u)[1-2u(1-u)]^t.
\end{equation}
We note that, for $u=0$, every witness test rejects; for $u=1$, the first accepts.
These endpoint cases agree with $g_t(0)=0$ and $g_t(1)=1$.
\\~\\
To extend this calculation to arbitrary inputs, choose an
orthonormal eigenbasis $\{\ket{e_k}\}_k$ of $E$, with eigenvalues
$u_k$, and write $\ket{v_k}=VJ_0\ket{e_k}$.
For $k\ne\ell$,
\[
\begin{aligned}
 \langle Tv_k,Tv_\ell\rangle
 &=\bra{e_k}E\ket{e_\ell}=0,\\
 \langle (I-T)v_k,(I-T)v_\ell\rangle
 &=\bra{e_k}(I-E)\ket{e_\ell}=0.
\end{aligned}
\]
Therefore, the spaces
\[
 \operatorname{span}\{T\ket{v_k},(I-T)\ket{v_k}\}
\]
are mutually orthogonal and invariant under $S$ and $T$.
Every sequence of these measurements preserves the spaces,
so distinct eigenvector components contribute without
interference to the success probability.
For a unit input vector $\ket\psi=\sum_k c_k\ket{e_k}$,
\[
 \Pr[\text{extract a witness on }\ket\psi]
 =\sum_k |c_k|^2g_t(u_k).
\]
Thus the witness procedure has success effect
\[
 W=g_t(E),\qquad 0\preceq W\preceq I.
\]
That is, its probability of outputting a witness on an input
state $\rho$ is $\Tr(W\rho)$.
\\~\\
We next control $A(I-L)$ in the decomposition
\(
 A=A(I-L)+AL.
\)
The term corresponding to $AL$ is bounded by
\eqref{eq:extraction-filtered-acceptance}.
For the other term, since $A=\Pi U^1J\preceq I$ is a contraction,
\[
 \bigl(A(I-L)\bigr)^\dagger A(I-L)
 =(I-L)A^\dagger A(I-L)
 \preceq(I-L)^2.
\]
We will show that $(I-L)^2\preceq W+\delta I$.
This will bound the squared norm of $A(I-L)\ket\psi$ by the
witness-extraction probability on $\ket\psi$, up to
$\delta\norm{\ket\psi}^2$.

Both $W=g_t(E)$ and $L=(I-E)^\beta$ are functions of $E$.
For each eigenvector $E\ket e=u\ket e$,
\[
\begin{aligned}
 W\ket e&=g_t(u)\ket e,\\
 (I-L)^2\ket e&=[1-(1-u)^\beta]^2\ket e.
\end{aligned}
\]
Thus $W-(I-L)^2+\delta I$ is diagonal in an orthonormal
eigenbasis of $E$, with eigenvalues
\[
 g_t(u)-[1-(1-u)^\beta]^2+\delta.
\]
To prove that this operator is positive semidefinite, it therefore
suffices to show that these scalar quantities are nonnegative.

Split the eigenvalue range at
\[
 \alpha=\frac{\delta}{2\beta}.
\]
Since $t=\lceil2\beta/\delta^2\rceil$,
\[
 t\alpha
 \ge\frac{2\beta}{\delta^2}\frac{\delta}{2\beta}
 =\frac1\delta.
\]

For $u\ge\alpha$, we bound the extraction failure probability.
Using \eqref{eq:extraction-witness-function} and
$1-x\le e^{-x}$,
\[
\begin{aligned}
 1-g_t(u)
 &=(1-u)[1-2u(1-u)]^t\\
 &\le(1-u)e^{-2tu(1-u)}\\
 &\le(1-u)e^{-2t\alpha(1-u)}\\
 &\le\frac1{2\mathrm e\,t\alpha}
 \le\delta.
\end{aligned}
\]
The penultimate inequality applies
$xe^{-ax}\le1/(\mathrm e a)$ with
$x=1-u$ and $a=2t\alpha$.
Hence $g_t(u)\ge1-\delta$ throughout this range.

For $0\le u<\alpha$, it is enough to bound the comparison term.
Bernoulli's inequality gives
\[
 (1-u)^\beta\ge1-\beta u,
 \qquad
 0\le1-(1-u)^\beta\le\beta u.
\]
Consequently,
\[
 [1-(1-u)^\beta]^2
 \le\beta^2u^2
 \le\beta^2\alpha^2
 =\frac{\delta^2}{4}
 \le\delta.
\]

For $u\ge\alpha$, use $g_t(u)\ge1-\delta$ and
$[1-(1-u)^\beta]^2\le1$.
For $u<\alpha$, use $g_t(u)\ge0$ and the bound just established.
Together these give
\[
 g_t(u)-[1-(1-u)^\beta]^2+\delta\ge0,
 \qquad u\in[0,1].
\]
Every eigenvalue of $W-(I-L)^2+\delta I$ is therefore
nonnegative, proving
\begin{equation}\label{eq:extraction-witness-effect}
 W\succeq(I-L)^2-\delta I.
\end{equation}
\paragraph{6. Comparing acceptance with extraction.}
We compare the original acceptance effect $A^\dagger A$ with
the witness success effect $W$, keeping the contribution of
filtered acceptance explicit.
For every input vector,
\[
 A\ket\psi=A(I-L)\ket\psi+AL\ket\psi.
\]
These two terms need not be orthogonal.
The weighted Cauchy--Schwarz inequality bounds their cross term
and gives
\[
 \norm{A\ket\psi}^2
 \le(1+\delta)\norm{A(I-L)\ket\psi}^2
    +(1+\delta^{-1})\norm{AL\ket\psi}^2.
\]

Since $L^\dagger=L$ and the inequality holds for every
$\ket\psi$, we obtain
\[
\begin{aligned}
 A^\dagger A
 &\preceq
 (1+\delta)(I-L)A^\dagger A(I-L)
 +(1+\delta^{-1})LA^\dagger AL\\
 &\preceq
 (1+\delta)(I-L)^2
 +(1+\delta^{-1})LA^\dagger AL.
\end{aligned}
\]
The second line uses $A^\dagger A\preceq I$. 
The extraction bound \eqref{eq:extraction-witness-effect} gives
\[
 (I-L)^2\preceq W+\delta I.
\]
Substituting this bound yields,
\begin{equation}\label{eq:extraction-local-comparison}
\begin{aligned}
A^\dagger A
 &\preceq
 (1+\delta)W+\delta(1+\delta)I
 +(1+\delta^{-1})LA^\dagger AL\\
 \implies A^\dagger A-W
 &\preceq
 \delta W+\delta(1+\delta)I
 +(1+\delta^{-1})LA^\dagger AL\\
 &\preceq
 \delta(2+\delta)I
 +(1+\delta^{-1})LA^\dagger AL.
\end{aligned}
\end{equation}
Here $LA^\dagger AL$ is the effect of successful filtering
followed by original acceptance.
Its total contribution over the locations is bounded by
\eqref{eq:extraction-total-filtered}.

\paragraph{7. Comparing the joint success thresholds.}
We now use the local comparison
\eqref{eq:extraction-local-comparison} to compare the probabilities
of at least $d$ original acceptances and at least $d$ witness
outputs on the joint input $\rho$.

For local success effects $0\preceq Z_i\preceq I$, the effects
of success and failure at location $i$ are $Z_i$ and $I-Z_i$,
respectively.
Thus the joint outcome $z=(z_1,\ldots,z_N)\in\{0,1\}^N$ has effect
\[
 \bigotimes_{i=1}^N Z_i^{z_i}(I-Z_i)^{1-z_i}.
\]
Summing the outcomes with at least $d$ successes gives the
threshold effect
\begin{equation}\label{eq:extraction-threshold-effect}
 T_d(Z_1,\ldots,Z_N)=
 \sum_{\substack{z\in\{0,1\}^N\\\sum_i z_i\ge d}}
 \bigotimes_{i=1}^N Z_i^{z_i}(I-Z_i)^{1-z_i}.
\end{equation}
Thus, 
\[
 \Pr[\text{at least $d$ successes}]
 =\Tr\!\left(T_d(Z_1,\ldots,Z_N)\rho\right).
\]
We replace the original tests by the witness extracting tests one location
at a time.
For $0\le i\le N$, define
\[
 F_i=T_d(W_1,\ldots,W_i,
              A_{i+1}^\dagger A_{i+1},\ldots,A_N^\dagger A_N).
\]
Thus $F_i$ uses witness extraction at locations $1,\ldots,i$
and original acceptance at locations $i+1,\ldots,N$.
In particular,
\[
 F_0=T_d(A_1^\dagger A_1,\ldots,A_N^\dagger A_N),
 \qquad
 F_N=T_d(W_1,\ldots,W_N).
\]

Now, fix a particular step $i\in[N]$.
Only the test at location $i$ changes between $F_{i-1}$ and $F_i$.
For this comparison, we consider the unchanged effects at all other locations as
\[
 Z_j=
 \begin{cases}
  W_j,&j<i,\\
  A_j^\dagger A_j,&j>i.
 \end{cases}
\]
Changing the outcome at location $i$ affects the global threshold $d$
only when the other locations have exactly $d-1$ successes.
Because, with at least $d$ other successes the threshold is already
met, while with at most $d-2$ it cannot be met by either local outcome.

Define the effect of exactly $d-1$ successes among the other
locations by
\[
 G_i=
 \sum_{\substack{(z_j)_{j\ne i}\in\{0,1\}^{N-1}\\
                  \sum_{j\ne i}z_j=d-1}}
 \bigotimes_{j\ne i}Z_j^{z_j}(I-Z_j)^{1-z_j}.
\]
Let $H_i$ be the effect of at least $d$ successes among those
same locations, i,e,
\[
 H_i=
 \sum_{\substack{(z_j)_{j\ne i}\in\{0,1\}^{N-1}\\
                  \sum_{j\ne i}z_j\geq d}}
 \bigotimes_{j\ne i}Z_j^{z_j}(I-Z_j)^{1-z_j}.
\]


The global threshold $d$ is met either through the event represented by
$H_i$, regardless of the local outcome, or through the event
represented by $G_i$ together with a local success at $i$.
Consequently, with identity factors implicit,
\[
 \begin{aligned}
  F_{i-1}&=H_i+(A_i^\dagger A_i)G_i,\\
  F_i&=H_i+W_iG_i.
 \end{aligned}
\]
Subtracting the common contribution $H_i$ we get:
\[
 F_{i-1}-F_i=(A_i^\dagger A_i-W_i)G_i.
\]
Summing the consecutive differences cancels all intermediate
effects $F_1,\ldots,F_{N-1}$, giving
\begin{equation}\label{eq:extraction-threshold-identity}
\begin{aligned}
 &T_d(A_1^\dagger A_1,\ldots,A_N^\dagger A_N)
  -T_d(W_1,\ldots,W_N)\\
 &\qquad=F_0-F_N
  =\sum_{i=1}^N(F_{i-1}-F_i)\\
 &\qquad=\sum_{i=1}^N(A_i^\dagger A_i-W_i)G_i.
\end{aligned}
\end{equation}

Since $G_i\succeq0$ acts outside location $i$, tensoring the
local comparison by $G_i$ preserves operator order.
Thus \eqref{eq:extraction-local-comparison} gives
\begin{align*}
 (A_i^\dagger A_i-W_i)G_i
 &\preceq
 \delta(2+\delta)G_i
 +(1+\delta^{-1})L_iA_i^\dagger A_iL_iG_i\\
 &\preceq
 \delta(2+\delta)I
 +(1+\delta^{-1})L_iA_i^\dagger A_iL_i.
\end{align*}
For the second line, use $G_i\preceq I$.

Combining these bounds with
\eqref{eq:extraction-threshold-identity}, we obtain
\[
 F_0-F_N
 \preceq
 N\delta(2+\delta)I
 +(1+\delta^{-1})
   \sum_{i=1}^N L_iA_i^\dagger A_iL_i.
\]
Now, taking the trace against $\rho$ (note, $\Tr\rho\le1$),
and taking expectation over over the oracle distributions, we get
\[
\begin{aligned}
 \E[\Tr(F_0\rho)]
 &=\Pr\!\left[\sum_iX_i\ge d\right],\\
 \E[\Tr(F_N\rho)]
 &=\Pr\!\left[\sum_iY_i\ge d\right].
\end{aligned}
\]
 we therefore obtain
\begin{align*}
 \Pr\!\left[\sum_iX_i\ge d\right]
 -\Pr\!\left[\sum_iY_i\ge d\right]
 &\le N\delta(2+\delta)\\
 &\quad +(1+\delta^{-1})\sum_{i=1}^N
     \E\!\left[\Tr(A_iL_i\rho_iL_iA_i^\dagger)\right].
\end{align*}

Finally, this along with the bound established in \eqref{eq:extraction-total-filtered} proves \eqref{eq:compact-extraction}.
We also note that, 
each local extractor performs one initial witness test and at
most $t=\lceil2\beta/\delta^2\rceil$ additional rounds.
The implementations of $V$, $V^\dagger$, and the witness test
therefore give the claimed polynomial complexity under the assumption of \Cref{rem:extraction-oracle-access}.
\end{proof}

\subsection{From cloning to torch signatures}
\label{subsec:extraction-oss}

We now apply \Cref{lem:compact-extraction} to prove global
non-proliferation for \Cref{con:torch-fire}.

Write \(s_i=(\ell_i,a_i)\) for the serial of the challenge flame
issued to location \(i\) after the split in the global non-proliferation experiment as in \Cref{def:global} of the torch-fire object.
We call \(\ell_i\in\{0,1\}^{\lambda}\) the \emph{challenge label}:
it selects the anchor source \(\mathcal O_{A,\ell_i}\), while
\(a_i\) is the anchor's verification key within that source
(\eqref{eq:tf-serial}).
Recall that each flame contains an anchor and a CGS part
(\eqref{eq:flame-state}). Successful cloning therefore requires
two anchors that pass the test under \(a_i\)
(\eqref{eq:flame}).

The proof uses the two OSS schemes as follows.
\begin{enumerate}[label=(\roman*),leftmargin=*,itemsep=4pt]
 \item Disable affine access to the anchor source
 \(\mathcal O_{A,\ell_i}\).
 We show that successful filtering followed by two accepted
 anchors would give a strong one-bit OSS forgery
 (\Cref{fact:sz}). This bounds the \(2\varepsilon_\beta\) term in \eqref{eq:compact-extraction}.

 \item Every query whose answer changes under this oracle disabling modification
 contains a valid torch signature (\Cref{con:gamma}).
 Its signed message includes \(\ell_i\).
 We then show that obtaining such signatures for \(m+1\)
 distinct challenge labels from only \(m\) issued torches
 would violate long-message OSS security (\Cref{lem:oss}).
\end{enumerate}
The extraction \Cref{lem:compact-extraction} connects these two steps even when the
locations share entanglement.
We first define the modified oracle and the witness queries.

\paragraph{The local input and acceptance test.}
Let \(\rho\) be the joint state of the local inputs
\(B_1F_1,\ldots,B_NF_N\) at the start of the global non proliferation experiment as in \Cref{def:global} of the torch-fire object. Let \(\rho_i\) be its reduced
state on \(B_iF_i\).
Due to setup aborts and label collisions,
assume \(\rho\) may be subnormalized.
For the operator definitions below, fix the classical oracle
functions and public inputs \(pk,s_1,\ldots,s_N\), where
\(s_i=(\ell_i,a_i)\).
All the expectation calculations below average over their distributions.

At location \(i\), let \(U_i^1\) coherently run
\(\mathcal A_i^{\mathcal O}\) on \(B_iF_i\), followed by verification
of both output flames under \(s_i\).
Retain the workspace and measurement records as in
\Cref{fact:coherent-implementation}.
Let \(J_i\) append zero workspace, and let \(\Pi_i\) test that
both verifications accept. Thus
\[
 A_i=\Pi_iU_i^1J_i,
 \qquad
 \Pr[X_i=1]=\E\!\left[\Tr(A_i\rho_iA_i^\dagger)\right].
\]
We implement the anchor verification tests by the Hadamard-dual-Hadamard test as in \Cref{lem:token-test} which corresponds to the projector \(\Pi_{a_i}=\proj{A_{a_i}}\),


\paragraph{Removing the challenge label's affine access.} Now we define the modified version of the computation as in \Cref{eq:original-modified-computation}.
The modified computation uses the same oracles except
that \(\Gamma\) returns \(\perp\) on every affine request carrying
\(\ell_i\). For a well-formed request
\(q'=(s',t,c,e,u)\), with \(s'=(\ell',a')\), define
\begin{equation}\label{eq:global-disabled-oracle}
 \Gamma_i^0(q',p,\sigma_A,\sigma_T)=
 \begin{cases}
  \perp,&\ell'=\ell_i,\\
  \Gamma(q',p,\sigma_A,\sigma_T),&\ell'\ne\ell_i.
 \end{cases}
\end{equation}
Malformed inputs also return \(\perp\).
Write \(\mathcal O_i^0\) for the public collection obtained by this
replacement, and let \(U_i^0\) be the mofied version of \(U_i^1\) with access to \(\mathcal O_i^0\).
Its accepting map is
\[
 C_i=\Pi_iU_i^0J_i.
\]


\paragraph{The witness in a distinguishing query.}
For the two affine oracles to differ, the original oracle must
answer a request carrying \(\ell_i\).
By \eqref{eq:tf-gate}, such a request must include a valid
signature under a published torch key.
We therefore let \(Q_i\) be the computational-basis diagonal
projector onto complete \(\Gamma\)-query inputs satisfying
\begin{equation}\label{eq:global-witness-test}
 \begin{aligned}
  &q'=(s',t,c,e,u)\in\mathcal Q_{\rm affine},
       \qquad s'=(\ell_i,a'),\\
  &t\in\Tpub,\qquad
    \OSSVer^{\mathcal O_T}
       \bigl(t,M_T(s',t),\sigma_T\bigr)=1.
 \end{aligned}
\end{equation}
Set \(Q_i=0\) on malformed inputs and designated padding queries.
The predicate is coherently testable using \(\mathcal T\) and
public torch OSS verification, both available in
\(\mathcal O_i^0\).

Outside the support of \(Q_i\), either the request has a
different label, or both affine oracles reject it.
Consequently,
\begin{equation}\label{eq:global-query-support}
 (U_\Gamma-U_{\Gamma_i^0})(I-Q_i)=0,
\end{equation}
which is precisely the witness condition of
\Cref{lem:compact-extraction}. Crucially,
The predicate need not check the anchor signature or the CGS
tag (\eqref{eq:tf-gate}). Those checks are necessary for a affine response, but the
reduction needs only the torch signature.
Thus a witness may satisfy \(Q_i\) even when \(\Gamma\) rejects it.

Measuring a witness and retaining its torch-signature fields
gives
\begin{equation}\label{eq:global-torch-witness}
 \bigl(t,M_T(s',t),\sigma_T\bigr),
 \qquad s'=(\ell_i,a'),\qquad t\in\Tpub,
\end{equation}
with a valid signature under \(t\).
The key \(a'\) need not equal the challenge key \(a_i\):
the relevant information is the label \(\ell_i\).
Indeed, \eqref{eq:authorization-messages} gives
\[
 M_T(s',t)
 =1\Vert\ell_i\Vert a'_1\Vert\cdots\Vert a'_L
       \Vert t_1\Vert\cdots\Vert t_L.
\]
These messages have length \(1+\lambda+2Lr\).
Distinct challenge labels therefore give distinct signed
messages, including when the witnesses use the same torch key.

\paragraph{The quantity left to bound.}
Choose a common power-of-two bound \(q\ge1\) on the local
\(\Gamma\)-query positions, padding with invalid queries.
All other oracle queries are identical in \(U_i^1\) and \(U_i^0\).
For the modified prefixes \(V_{i,j}=U_{i,<j}^0J_i\), define
\[
 E_i=\frac1q\sum_{j=1}^qV_{i,j}^\dagger Q_iV_{i,j},
 \qquad
 L_i=(I-E_i)^\beta,
 \qquad \beta\ge1,
\]
as in \Cref{subsec:separated-extraction}.
Consider, the error term in the extraction \Cref{lem:compact-extraction} is
\[
 \varepsilon_\beta
 =\sum_{i=1}^N
   \E\!\left[\Tr(C_iL_i\rho_iL_iC_i^\dagger)\right].
\]
Each summand includes both successful filtering and acceptance
by the modified computation.
We will bound it by anchor OSS security, then use torch OSS
security to bound the number of successful witness outputs.

\paragraph{Why the preparation phase must also be modified.}
The input \(\rho\) was prepared before the split in the experiment \Cref{def:global} using the
original oracle \(\Gamma\).
It could therefore already contain information obtained from
the affine queries that \(U_i^0\) later disables.
To use anchor OSS security, we must remove that access from
the preparation phase as well.
The challenge label is uniform and unknown during this phase;
the next lemma bounds the effect of the additional replacement.

\begin{lemma}
\label{lem:fresh-label}
\normalfont
Fix \(i\in[N]\).
Replace \(\Gamma\) by \(\Gamma_i^0\) also before the split in th experiment \Cref{def:global},
keeping the other oracles and honest issuance unchanged.
Let \(\widetilde\rho_i\) be the resulting joint input state after
challenge issuance and the experiment's abort checks.
For every effect \(0\preceq Z\preceq I\) on the joint input,
possibly depending on the classical oracle functions and
public inputs,
\begin{equation}\label{eq:fresh-label-distance}
 \left|\E[\Tr(Z\rho)]
       -\E[\Tr(Z\widetilde\rho_i)]\right|
 \le2Q_{\rm pre}\,2^{-\lambda/2}.
\end{equation}
Here \(Q_{\rm pre}\) is the pre-split query bound fixed at the
start of \Cref{sec:extraction}.
The subscript on \(\widetilde\rho_i\) specifies which label's
affine queries are disabled; this is a state of all locations.
\end{lemma}
\begin{proof}
\normalfont
Fix the oracle functions. Sample
\(\ell_i\leftarrow\{0,1\}^{\lambda}\) before running
\(\mathcal A_0\), keeping it secret until challenge issuance,
and defer the label-collision checks.
Let \(\omega_j\) be the query-register state immediately before
the \(j\)-th query in the original pre-split execution.
These states are independent of \(\ell_i\).

Let \(\Pi_\ell^{\rm affine}\) project onto \(\Gamma\)-query inputs
carrying label \(\ell\). Since
\(\sum_\ell\Pi_\ell^{\rm affine}\preceq I\), the total query mass
at the hidden label satisfies
\[
 \begin{aligned}
 M_{\ell_i}
 &:=\sum_{j=1}^{Q_{\rm pre}}
       \Tr(\Pi_{\ell_i}^{\rm affine}\omega_j),\\
 \E_{\ell_i}[M_{\ell_i}]
 &=2^{-\lambda}\sum_{j=1}^{Q_{\rm pre}}
       \Tr\!\left[
          \left(\sum_\ell\Pi_\ell^{\rm affine}\right)\omega_j
       \right]
 \le Q_{\rm pre}\,2^{-\lambda}.
 \end{aligned}
\]

By \eqref{eq:global-disabled-oracle},
\[
 (U_\Gamma-U_{\Gamma_i^0})
       (I-\Pi_{\ell_i}^{\rm affine})=0,
\]
and all other oracles are unchanged.
Thus, writing \(\tau^1,\tau_i^0\) for the two pre-split states,
\Cref{lem:bbbv-hybrid} and Jensen's inequality give
\[
 \begin{aligned}
 \E_{\ell_i}\!\left[D(\tau^1,\tau_i^0)\right]
 &\le 2\sqrt{Q_{\rm pre}}\,
       \E_{\ell_i}\!\left[\sqrt{M_{\ell_i}}\right]\\
 &\le 2\sqrt{Q_{\rm pre}\,\E_{\ell_i}[M_{\ell_i}]}\\
 &\le 2Q_{\rm pre}\,2^{-\lambda/2}.
 \end{aligned}
\]

Applying the common challenge-issuance and collision-check,
recording aborts as failure and extending \(Z\) by zero on failure, we have
\[
 \left|
   \E\Tr(Z\rho)-\E\Tr(Z\widetilde\rho_i)
 \right|
 \le \E D(\tau^1,\tau_i^0)
 \le 2Q_{\rm pre}\,2^{-\lambda/2},
\]
where \(\E\) also averages over the oracle functions.
This proves \eqref{eq:fresh-label-distance} without conditioning
on non-abort.
\end{proof}

We now bound the \(2\varepsilon_\beta\) term in
\eqref{eq:compact-extraction} using anchor OSS security.
The reduction runs the filter followed by the modified
computation, counting filter failures as failures of the
reduction. Two accepted anchors will then provide a
strong one-bit OSS forgery.

\begin{lemma}
\label{lem:global-filtered-anchor}
\normalfont
For every polynomially bounded integer \(\beta\ge1\),
\begin{equation}\label{eq:global-null-small}
 \varepsilon_\beta
 =\sum_{i=1}^N
   \E\!\left[\Tr(L_iC_i^\dagger\rho_i C_iL_i)\right]
 \le\negl(\lambda),
\end{equation}
where local operators act as the identity on all other
locations.
\end{lemma}

\begin{proof}
\normalfont
We construct a strong one-bit OSS forger in the following way.
\paragraph{1. Disable the affine queries during preparation.}
The reduction must also avoid the challenge label's private
affine queries before the split.
Since \(C_i\) is a contraction and \(0\preceq L_i\preceq I\),
\[
 Z_i:=L_iC_i^\dagger C_iL_i,
 \qquad
 0\preceq Z_i\preceq L_i^2\preceq I.
\]
Thus \Cref{lem:fresh-label}, applied to \(Z_i\), gives
\begin{equation}\label{eq:global-filtered-preparation}
 \begin{aligned}
 p_i
 &:=\E\!\left[\Tr(Z_i\widetilde\rho_i)\right],\\
 \varepsilon_\beta
 &\le \sum_{i=1}^N p_i
       +2NQ_{\rm pre}\,2^{-\lambda/2}.
 \end{aligned}
\end{equation}
We construct an OSS forger whose success probability is
proportional to \(\sum_i p_i\).

\paragraph{2. Set up the OSS reduction.}
In the OSS game of \Cref{fact:sz}, the challenger
\(\mathcal C_{\rm OSS}\) samples
\[
 \mathcal O^\star
 \leftarrow \OSSGenOracles(1^\lambda)
\]
and gives the forger \(\mathcal B_{\rm A}\) access to its
public oracles.
The forger attacks this source by running the torch-fire
adversary \(\mathcal A\) in the modified experiment.

The forger chooses \(i\leftarrow[N]\) uniformly and secretly
samples \(\ell_i\leftarrow\{0,1\}^{\lambda}\).
It answers public oracle queries addressed to anchor
component \((\ell_i,1)\) using \(\mathcal O^\star\).
It samples the torch source \(\mathcal O_T\), the CGS source
\(\mathcal O_K\), all other anchor components, and the
remaining construction data as prescribed by
\Cref{con:torch-fire}.
Since the anchor components are independent one-bit OSS
sources, this gives the prescribed joint distribution
of the underlying sources.

The adversary receives \(\Gamma_i^0\) from
\eqref{eq:global-disabled-oracle} throughout the experiment:
requests carrying \(\ell_i\) return \(\perp\), and all other
affine responses and checks are computed from the sources
sampled by \(\mathcal B_{\rm A}\).

For the first component of the challenge anchor, the forger runs
\[
 (y,\ket{W_y})
 \leftarrow
 \OSSGenQKey^{\mathcal O^\star}(1^\lambda),
 \qquad a_{i,1}:=y.
\]
By \Cref{fact:sz-generation}, this produces the required
uniform verification key and signing state using public
oracle queries.
The remaining challenge components and honest issuances
use the sampled sources.
Taking into account the abort possibilities, the resulting
joint input is \(\widetilde\rho_i\).

\paragraph{3. Run the filter and the modified computation.}
At the chosen location, \(\mathcal B_{\rm A}\) runs \(L_i\),
then \(U_i^0\) and the acceptance test \(\Pi_i\).
The joint successful branch has map \(C_iL_i\), and hence
probability
\[
 p_i
 =\E\!\left[
   \Tr(C_iL_i\widetilde\rho_iL_iC_i^\dagger)
 \right].
\]
By \Cref{rem:extraction-oracle-access} and
\eqref{eq:global-witness-test}, the complete reduction makes
polynomially many queries to \(\mathcal O^\star\) for
polynomial \(\beta\).

\paragraph{4. Forge under the challenged OSS source.}
Retain the two output anchor registers after their projective
tests.
By \Cref{lem:token-test} and \eqref{eq:flame}, the accepting
map on each anchor is the rank-one projector
\[
 \Pi_{a_i}=\proj{A_{a_i}}.
\]
The joint successful state therefore has the form
\[
 \proj{A_{a_i}}_{A^{(0)}}
 \otimes
 \proj{A_{a_i}}_{A^{(1)}}
 \otimes\tau,
 \qquad
 \tau\succeq0,
 \qquad
 \E[\Tr(\tau)]=p_i.
\]

For fixed source data and \(y=a_{i,1}\), measuring component
\(1\) of each anchor gives independent uniform samples
\(u_0,u_1\in W_y\). Thus
\[
 \Pr[u_0\ne u_1\mid\text{joint success}]
 =1-\frac1{|W_y|}
 =1-2^{-(n-r)}.
\]
When \(u_0\ne u_1\), the forger returns to
\(\mathcal C_{\rm OSS}\) the key and signature pairs
\[
 y,\qquad
 \bigl((u_0)_1,u_0\bigr),\qquad
 \bigl((u_1)_1,u_1\bigr).
\]
These are signatures for the source sampled by the OSS
challenger: by \eqref{eq:one-bit-verification},
\[
 \OSSVer^{\mathcal O^\star}
       \bigl(y,(u_b)_1,u_b\bigr)=1
 \qquad (b\in\{0,1\}).
\]
The pairs are distinct because \(u_0\ne u_1\), even if
their message bits coincide.
The strong OSS game in \Cref{fact:sz} permits the forger
to output this key \(y\).

For each chosen location, the forging probability is
\((1-2^{-(n-r)})p_i\).
Averaging over \(i\leftarrow[N]\) gives
\[
 \Pr[\mathcal B_{\rm A}\text{ forges}]
 =\frac{1-2^{-(n-r)}}{N}\sum_{i=1}^N p_i
 \le\negl(\lambda),
\]
where the inequality follows from \Cref{fact:sz}.
Finally, \eqref{eq:global-filtered-preparation} and
\(n-r=24\lambda\) yield
\[
 \begin{aligned}
 \varepsilon_\beta
 &\le
 \frac{N}{1-2^{-24\lambda}}\,
 \Pr[\mathcal B_{\rm A}\text{ forges}]
 +2NQ_{\rm pre}\,2^{-\lambda/2}\\
 &\le\negl(\lambda),
 \end{aligned}
\]
since \(N\) and \(Q_{\rm pre}\) are polynomially bounded.
\end{proof}

We now bound the witness extraction probability term in \eqref{eq:compact-extraction}
at threshold \(m+1\).
Distinct labels give distinct torch messages by
\eqref{eq:authorization-messages}.
Among \(m+1\) witnesses, either two use the same key of an
issued torch, or one uses a key whose torch was not issued.

\begin{lemma}
\label{lem:global-witness-count}
\normalfont
Let \(\mathcal A_{\rm W}\) be a quantum procedure following
the global non-proliferation experiment's (\Cref{def:global}) issuance rules, receiving only
the \(m\) issued torch registers, and making polynomially
many public-oracle queries.
For witnesses defined by \eqref{eq:global-torch-witness},
\[
 p_{\rm wit}
 :=\Pr\!\left[
   \begin{array}{c}
   \mathcal A_{\rm W}\text{ outputs valid witnesses}\\
   \text{for }m+1\text{ distinct labels}
   \end{array}
 \right]
 \le\negl(\lambda).
\]
The procedure may process all pre-split and challenge flames
jointly.
\end{lemma}

\begin{proof}
\normalfont
If \(R=0\), then \(\Tpub=\varnothing\), so no witness is valid.
Assume \(R\ge1\).
We construct a long-message OSS forger \(\mathcal B_{\rm T}\)
using \(\mathcal A_{\rm W}\).

\paragraph{1. Simulate setup and issuance.}
The OSS challenger \(\mathcal C_{\rm OSS}\) samples
\[
 \mathcal O^\star
 \leftarrow\OSSGenOracles(1^\lambda)
\]
for the long-message scheme of \Cref{con:long-oss}, and
gives \(\mathcal B_{\rm T}\) access to its public oracles.
The forger answers simulated torch-source queries using
\(\mathcal O^\star\), and samples all anchor and CGS sources
as prescribed by \Cref{con:torch-fire}.
Its target is a forgery under \(\mathcal O^\star\).

To generate the published torches, \(\mathcal B_{\rm T}\) runs
\[
 (t_j,T_j)
 \leftarrow
 \OSSGenQKey^{\mathcal O^\star}(1^\lambda),
 \qquad j\in[R].
\]
By \Cref{fact:sz-generation,con:long-oss}, for each fixed
source the keys are independent and uniform in
\((\{0,1\}^r)^L\), and their corresponding signing state is
\[
 \rho_{T_1\cdots T_R}
 =\bigotimes_{j=1}^R\proj{T_{t_j}}.
\]
This is the setup distribution of \Cref{con:torch-fire},
obtained here without \(\mathsf{td}_T\).
Publish the keys through \(\mathcal T\) as in
\eqref{eq:tf-torch-list}.

For the selected set \(J\subseteq[R]\), \(|J|=m\), give
\(\mathcal A_{\rm W}\) precisely the registers
\((T_j)_{j\in J}\) and retain \((T_j)_{j\notin J}\).
Generate all flames from the sampled anchor and CGS sources.
The oracle \(\Gamma\) uses the challenged source only for
the public torch-verification test in \eqref{eq:tf-gate};
its remaining checks and affine responses use the sampled data.

\paragraph{2. Turn the witnesses into a forgery.}
On witness success, retain \(m+1\) witnesses with distinct
labels \(\ell_1,\ldots,\ell_{m+1}\).
Write their signature fields as
\[
 (t_{j_h},M_h,\sigma_h),
 \qquad j_h\in[R],\quad h\in[m+1].
\]
Witness validity and \eqref{eq:authorization-messages} give
\[
 \begin{aligned}
 \OSSVer^{\mathcal O^\star}
       (t_{j_h},M_h,\sigma_h)&=1,\\
 h\ne k
 \quad\Longrightarrow\quad
 \ell_h\ne\ell_k
 &\quad\Longrightarrow\quad M_h\ne M_k.
 \end{aligned}
\]
The forger proceeds as follows.

\begin{enumerate}[label=(\roman*),leftmargin=*,itemsep=6pt]
 \item \emph{Every \(j_h\) belongs to \(J\).}
 Since \(|J|=m\), there are \(h\ne k\) with
 \(j_h=j_k=:j\).
 Return to the OSS challenger
 \[
  t_j,\qquad (M_h,\sigma_h),\qquad(M_k,\sigma_k).
 \]
 Both signatures verify under the same key, and
 \(M_h\ne M_k\).

 \item \emph{Some \(j_h\) does not belong to \(J\).}
 Choose such an \(h\) and put \(j:=j_h\).
 Using the retained register \(T_j\), compute
 \[
  M_\star:=0^{\,1+\lambda+2Lr},
  \qquad
  \sigma_\star
  \leftarrow
  \OSSSign^{\mathcal O^\star}(t_j,T_j,M_\star).
 \]
 By \eqref{eq:authorization-messages}, \(M_h\) starts
 with \(1\), so \(M_\star\ne M_h\).
 Return
 \[
  t_j,\qquad(M_h,\sigma_h),\qquad(M_\star,\sigma_\star).
 \]
 This is a forgery whenever \(\sigma_\star\) verifies.
\end{enumerate}

Both cases attack the source sampled by
\(\mathcal C_{\rm OSS}\).
The OSS game of \Cref{lem:oss} permits the forger to output
the selected verification key \(t_j\).

\paragraph{3. Bound the final signing error.}
The bound \(e_T\) in \eqref{eq:outer-sign-errors} holds
for every published key and message, so it applies also
to the untouched torch selected from the witnesses in
case~(ii).
Since case~(i) always yields a forgery,
\[
 \begin{aligned}
 p_{\rm wit}
 &\le \Pr[\mathcal B_{\rm T}\text{ forges}]+\E[e_T]\\
 &\le \negl(\lambda)
      +\negl(\lambda)\\
 &\le \negl(\lambda),
 \end{aligned}
\]
\end{proof}

\subsection{Global non-proliferation}
\label{subsec:global-theorem}

We now combine the two reductions with the threshold guarantee
of \Cref{lem:compact-extraction}.
The anchor reduction makes its filtered-acceptance error
negligible. The torch reduction rules out \(m+1\) extracted
witnesses with distinct labels.

\begin{theorem}[Global non-proliferation]
\label{thm:compact-global}
\label{thm:tf-global}
\normalfont
\Cref{con:torch-fire} satisfies \Cref{def:global}.
For polynomially bounded \(R,m,N\), with \(0\le m\le R\), and
every adversary allowed in that definition,
\begin{equation}\label{eq:tf-global-security}
 \Pr\!\left[
 \mathsf{SPLIT}^{\mathsf{TF}}_{\mathcal A^{\mathcal O}}
       (1^\lambda,1^R,m,N)=1
 \right]\le\negl(\lambda).
\end{equation}
\end{theorem}
\begin{proof}
\normalfont
If \(m\ge N\), the winning event is empty.
Assume \(m<N\).

\paragraph{1. Apply extraction to the ideal experiment (\Cref{eq:global-reference-comparison}).}
Let
\[
 p_{\rm ideal}
 =\Pr\!\left[\sum_{i=1}^N X_i\ge m+1\right].
\]
For parameters \(\beta\ge1\) and \(\delta\in(0,1)\), apply
\Cref{lem:compact-extraction} to the computations and witnesses
defined in \Cref{subsec:extraction-oss}.
Let \(Y_i\) indicate that the extractor at location \(i\)
outputs a witness, and put
\[
 p_{\rm ext}
 =\Pr\!\left[\sum_{i=1}^N Y_i\ge m+1\right].
\]
The extraction \Cref{lem:compact-extraction} gives
\begin{equation}\label{eq:global-threshold-application}
 \begin{aligned}
 p_{\rm ideal}
 &\le p_{\rm ext}+N\delta(2+\delta)\\
 &\quad +(1+\delta^{-1})
       \left(2\varepsilon_\beta
             +\frac{8Nq^2}{2\beta+1}\right).
 \end{aligned}
\end{equation}
Note, it does
not require the local cloning or extraction outcomes to be
independent.

\paragraph{2. Bound the extraction event by torch OSS security.}
The ideal experiment enforces distinct challenge labels,
and \(Y_i=1\) supplies a valid witness for \(\ell_i\) by
\eqref{eq:global-torch-witness}. Thus
\[
 \bigl|\{\ell_i:Y_i=1\}\bigr|
 =\sum_{i=1}^N Y_i,
\]
so \(m+1\) extraction successes give witnesses for \(m+1\)
distinct labels.

Applying \Cref{lem:global-witness-count} gives
\[
 p_{\rm ext}
 =\Pr\!\left[\sum_{i=1}^N Y_i\ge m+1\right]
 \le\negl(\lambda).
\]

\paragraph{3. Choose the extraction parameters.}
Fix any integer \(c\ge1\) and set
\begin{equation}\label{eq:global-extraction-parameters}
 \delta=\frac{1}{2N\lambda^c},
 \qquad
 \beta=\left\lceil N^2q^2\lambda^{2c}\right\rceil.
\end{equation}
For this fixed \(c\), Step~2 and
\Cref{lem:global-filtered-anchor} make
\(p_{\rm ext}+2(1+\delta^{-1})\varepsilon_\beta\)
negligible. Therefore,
\eqref{eq:global-threshold-application} gives
\[
 \begin{aligned}
 p_{\rm ideal}
 &\le \negl(\lambda)
      +N\delta(2+\delta)
      +(1+\delta^{-1})\frac{8Nq^2}{2\beta+1}\\
 &\le \negl(\lambda)+O(\lambda^{-c}).
 \end{aligned}
\]
Since \(c\) is arbitrary and \(p_{\rm ideal}\) is independent
of the extraction parameters, \(p_{\rm ideal}\) is negligible.

\paragraph{4. Return to the real experiment.}
The ideal experiment differs from the real one only by
declaring failure on repeated spark labels.
By \eqref{eq:global-reference-comparison},
\[
 p_{\rm real}
 \le p_{\rm ideal}+\Delta_{\rm lab},
 \qquad
 \Delta_{\rm lab}=\binom{v+N}{2}2^{-\lambda}.
\]
Both terms are negligible for polynomially bounded \(v,N\).
This proves \eqref{eq:tf-global-security}.
\end{proof}
\section*{Acknowledgments}
AM is supported by the National Research Foundation, Singapore, through the 
National Quantum Office, hosted in A*STAR, under its National Quantum Scholarship Scheme (NQSS) 
Funding Initiative (PhD).

\clearpage
\printbibliography
\appendix
\section{Deferred details from \protect\Cref{con:cgs-algorithms}}
\subsection{Coherent Tag Generation}
\label{sec:coherent-tag-generation}
We start with the following tool. For, the signing algorithm of \Cref{con:one-bit-oss}, 
the coherent signing circuit of
\Cref{rem:coherent-signing} is available using the public
collection $\mathcal O$.
We denote it by $U_{y,b}^{\mathcal O}$ and use the signing
algorithm with security parameter $\mu$.

\par\smallskip\noindent
\textbf{Coherent tag circuit.}
Recall that the flame register is $F=YU$.
For each $b\in\{0,1\}$, let $E_b$ be signing workspace
and let $\Sigma_b$ hold the signature.
The coherent signing circuit $U_{y,b}^{\mathcal O}$
acts on $UE_b\Sigma_b$.
\\
To apply this circuit without measuring the verification
key in $Y$, define
\[
 \widehat U_b^{\mathcal O}
 =\sum_{y\in\{0,1\}^{r}}
   \proj{y}_Y\otimes U_{y,b}^{\mathcal O}.
\]
Thus $\widehat U_b^{\mathcal O}$ acts on $YUE_b\Sigma_b$:
$Y$ is the control register, and $UE_b\Sigma_b$ are the
registers on which the selected signing circuit acts.

Let $Z$ hold the query, and let $R_0,R_1$ hold the two
oracle responses.
Initialize $E_b,\Sigma_b,R_b$ to zero for each $b$.
The oracle query unitary $U_{O_b}$ acts on
$Y\Sigma_bZR_b$: it uses $(Y,\Sigma_b,Z)$ as its input
registers and writes the response into $R_b$, leaving
$Y,\Sigma_b,Z$ unchanged.

For each bit $b$, perform the following operations:
\begin{enumerate}[label=(\roman*),leftmargin=*,itemsep=3pt]
 \item Apply $\widehat U_b^{\mathcal O}$ to
 $YUE_b\Sigma_b$ to compute the signature in $\Sigma_b$.

 \item Apply $U_{O_b}$ to $Y\Sigma_bZR_b$ to obtain the
 oracle response in $R_b$.

 \item Apply $(\widehat U_b^{\mathcal O})^\dagger$ to
 $YUE_b\Sigma_b$, leaving $Z$ and $R_b$ untouched.
\end{enumerate}
Each operation acts as the identity on registers not
listed above.
Writing $T_b^{\mathcal O}$ for this sequence, define
\begin{equation}\label{eq:native-tag-circuit}
 \begin{aligned}
  T_b^{\mathcal O}
    &=(\widehat U_b^{\mathcal O})^\dagger
      U_{O_b}\widehat U_b^{\mathcal O},
      &&b\in\{0,1\},\\
  \FireTag&=T_1^{\mathcal O}T_0^{\mathcal O}.
 \end{aligned}
\end{equation}
Thus $T_0^{\mathcal O}$ is applied first, followed by
$T_1^{\mathcal O}$.
Both use the same registers $Y,U,Z$, with separate
auxiliary registers $E_b,\Sigma_b,R_b$.
The tag is held in $(Y,R_0,R_1)$.
All auxiliary registers are retained so that
$(\FireTag)^\dagger$ can later be applied to the same
registers.

\subsection{Implementation Details of Cloning a flame}
\label{sec:flame-cloning-details}
Encode each oracle call in a query register $Z$ using
$\enc$ from \eqref{eq:native-encoding}.
Let $C$ denote its answer register.
For each call, introduce fresh zero auxiliary registers
$E_0,\Sigma_0,R_0,E_1,\Sigma_1,R_1$ and execute:
\begin{enumerate}[label=(\roman*),leftmargin=*,itemsep=3pt]
 \item Apply $\FireTag$ to
 \[
  FZ E_0\Sigma_0R_0 E_1\Sigma_1R_1.
 \]
 The resulting tag occupies $(Y,R_0,R_1)$.

 \item Apply the public query unitary $U_{O_2}$.
 Its input registers are $(Z,Y,R_0,R_1)$, and its
 answer register is $C$.
 For a valid tag, $O_2$ returns the value of the
 encoded forward or inverse call.

 \item Apply $(\FireTag)^\dagger$ to the same registers
 as in Step~(i), leaving $C$ untouched.
\end{enumerate}
Thus the public circuit replacing each \(\mathcal G\)
query is
\begin{equation}
 \widetilde U_{\mathcal G}^{\mathcal O}
 =(\FireTag)^\dagger U_{O_2}\FireTag.
\end{equation}

\section{Deferred details from \protect\Cref{con:torch-fire}}
\subsection{Coherent execution of an affine query}
\label{sec:coherent-affine-query}
\paragraph{Coherent preparation of an affine query.}
Fix a serial $s=(\ell,a)$, a published torch key
$t\in\Tpub$, a component $c\in[L]$, and $e\in\{0,1\}$.
We describe how the anchor $A$, torch $T$, and CGS flame part $K$
are used to request a value of $\operatorname{affine}_{a_c,e}$ from
$\Gamma$ in a coherent way.

Let $X$ hold the affine input $u$ and let $C$ be the answer register.
The preparation circuit $S_{c,e}^{\mathcal O}$ performs the following
operations, with fresh auxiliary registers initialized to zero.
\begin{enumerate}[label=(\roman*),leftmargin=*,itemsep=4pt]
 \item \emph{Form the request.}
 Coherently compute
 \[
 q=(s,t,c,e,u),\qquad z_q=\enc(\texttt{forward},\mathsf{code}(q))
 \]
 into registers $Q,Z$, respectively, leaving $X$ unchanged.
 The encoding is the one defined in \eqref{eq:tf-request-encoding}.

 \item \emph{Prepare the CGS tag.}
 Write $K=YU$ and apply $\FireTag$ from
 \eqref{eq:native-tag-circuit} to $K$ with request register $Z$.
 This uses $\mathcal D,O_0,O_1$ from the CGS fire source
 $\mathcal O_K$.
 If $R_0,R_1$ denote the response registers of that circuit,
 the tag occupies $(Y,R_0,R_1)$.
 Retain all its auxiliary registers for the inverse operation.

 \item \emph{Prepare the two signatures.}
 Using the messages in \eqref{eq:authorization-messages}, apply
 \[
 U_{a,M_A(s,t)}^{\mathcal O_{A,\ell}}
 \quad\text{to }A,
 \qquad
 U_{t,M_T(s,t)}^{\mathcal O_T}
 \quad\text{to }T,
 \]
 with their zero signing workspace, and store the signatures in
 $\Sigma_A,\Sigma_T$.
 These are the coherent long-message signing circuits of
 \eqref{eq:long}; each uses its source's digest oracle and component
 $\mathcal D$ interfaces.
 Retain the signing workspace as in \Cref{rem:coherent-signing}.
\end{enumerate}
The circuit $S_{c,e}^{\mathcal O}$ preserves $X$ and acts trivially
on $C$. The signatures and tag are kept coherent throughout.

\paragraph{Querying and reversing the preparation.}
After applying $S_{c,e}^{\mathcal O}$, query $\Gamma$ with inputs
\[
 (Q,(Y,R_0,R_1),\Sigma_A,\Sigma_T)
\]
and answer register $C$.
The oracle checks the signatures and tag as in \eqref{eq:tf-gate}
and returns the requested affine value when the checks pass.
Then apply $(S_{c,e}^{\mathcal O})^\dagger$ to the preparation
registers, leaving $C$ untouched. This defines
\begin{equation}\label{eq:tf-row-circuit}
 \mathcal R_{c,e}^{\mathcal O}
 =(S_{c,e}^{\mathcal O})^\dagger
   U_\Gamma S_{c,e}^{\mathcal O}.
\end{equation}
All workspace is retained until the inverse is applied.
No signature, tag, or answer is measured.
The approximation to an ideal affine query on honest resources is
proved in \Cref{lem:tf-row-access}.

\subsection[Proof of Lemma~\ref{lem:tf-row-access}]{Proof of \Cref{lem:tf-row-access}}
\label{sec:proof:of-lemma-tf-row-access}
\begin{proof}
\normalfont
Here \(U=U_{\operatorname{affine}_{a_c,e}}\) is the ideal affine-query unitary. For fixed \(a,c,e\), write \(f(u)=\operatorname{affine}_{a_c,e}(u)\). Then
\[
U\bigl(\ket u_X\ket z_C\otimes\ket\eta\bigr)
=
\ket u_X\ket{z\oplus f(u)}_C\otimes\ket\eta.
\]
It writes the affine value into \(C\) and leaves every other register unchanged.
\\~\\
Also, recall the preparation circuit $S=S_{c,e}^{\mathcal O}$ from \eqref{eq:tf-row-circuit}.
Let $P$ project onto correctly encoded requests for the input
in $X$ that pass \eqref{eq:tf-gate}.
The preparation encodes these requests correctly, so the signing
and tag bounds \eqref{eq:outer-sign-errors} and
\eqref{eq:tf-native-errors} give
$\norm{(I-P)S\ket\psi}^2\le d$.

On this subspace, $\Gamma$ returns the ideal affine value:
$(U_\Gamma-U)P=0$.
Also, $S,U$ commute since $S$ preserves $X$ and does not act on $C$.
Consequently,
\[
\begin{aligned}
 \norm{\bigl(\mathcal R_{c,e}^{\mathcal O}
             -U_{\operatorname{affine}_{a_c,e}}\bigr)\ket\psi}=\norm{(S^\dagger U_\Gamma S-U)\ket\psi}
 &=\norm{(U_\Gamma-U)(I-P)S\ket\psi}\\
 &\le2\norm{(I-P)S\ket\psi}
 \le2\sqrt d.
\end{aligned}
\]
The expectation bound follows from \eqref{eq:outer-sign-errors}
and \eqref{eq:tf-native-mean-errors}.
\end{proof}

\end{document}